\documentclass[11pt,fleqn]{article}
\usepackage{etex}
\usepackage{amssymb, latexsym, amsmath, amsfonts, graphicx, ebezier, enumitem}
\usepackage{hyperref}
\usepackage{pictex,color}
\usepackage{pict2e}
\usepackage{comment}
\usepackage{todonotes}
    \advance\textheight by \topskip

\usepackage{cancel}
\usepackage{ulem}

\usepackage{tikz}
\usetikzlibrary{decorations.markings}

\numberwithin{equation}{section}

\usepackage{bm}
\definecolor{shadecolor}{rgb}{0.9, 0.9, 0.86}

\makeatletter
\def\eqalign#1{\null\vcenter{\def\\{\cr}\openup\jot\m@th
  \ialign{\strut$\displaystyle{##}$\hfil&$\displaystyle{{}##}$\hfil
      \crcr#1\crcr}}\,}
\makeatother

\def\Xint#1{\mathchoice
   {\XXint\displaystyle\textstyle{#1}}%
   {\XXint\textstyle\scriptstyle{#1}}%
   {\XXint\scriptstyle\scriptscriptstyle{#1}}%
   {\XXint\scriptscriptstyle\scriptscriptstyle{#1}}%
   \!\int}
\def\XXint#1#2#3{{\setbox0=\hbox{$#1{#2#3}{\int}$}
     \vcenter{\hbox{$#2#3$}}\kern-.5\wd0}}

\def\dashint{\Xint-}

\usepackage{color}

\newcommand{\clr}{\textcolor[rgb]{1.00,0.00,0.00}}
\newcommand{\clu}{\textcolor[rgb]{0.00,0.00,0.00}}

\def\beq{\begin{equation}}
\def\eeq{\end{equation}}

\newcommand{\be}{\begin{equation}}
\newcommand{\ee}{\end{equation}}

\newcommand{\wt}{\widetilde}

\newcommand{\al}{\alpha}
\newcommand{\bt}{\beta}

\newcommand{\lb}{\lambda}

\newcommand{\sP}{\sigma_{\hskip-1mm_P}}

\renewcommand{\(}{\left(}
\renewcommand{\)}{\right)}

\renewcommand{\i}{{i}}
\newcommand{\e}{\mathrm{e}}

\renewcommand{\d}{{d}}

\renewcommand{\ln}{\log}

    \def\Re{{\rm Re \,}}
    \def\Im{{\rm Im \,}}

    \def\P2n{{\rm P}_{{\rm II}}^{(n)}}

    \newtheorem{theorem}{Theorem}[section]
    \newtheorem{lemma}[theorem]{Lemma}
    \newtheorem{corollary}[theorem]{Corollary}
    \newtheorem{proposition}[theorem]{Proposition}
    
    \newtheorem{Definition}[theorem]{Definition}
    
    \newtheorem{Remark}[theorem]{Remark}
    \newenvironment{remark}{\begin{Remark}\rm}{\end{Remark}}
    \newtheorem{Example}[theorem]{Example}
    
    \newtheorem{Assumptions}[theorem]{Assumptions}

    \newenvironment{proof}%
    {\rm \trivlist \item[\hskip \labelsep{\bf Proof. }]}%
    {\hspace*{\fill}$\Box$\endtrivlist}
    {\rm \trivlist \item[\hskip \labelsep{\bf Proof}]}%
    {\hspace*{\fill}$\Box$\endtrivlist}

    \DeclareMathOperator*{\Tr}{Tr}

\begin{document}
\title{Weak and strong confinement in the Freud random matrix ensemble and gap probabilities}
\author{T. Claeys\footnote
{
Institut de Recherche en Math\'ematique et Physique,  UCLouvain, Chemin du Cyclotron 2, B-1348
Louvain-La-Neuve, Belgium; e-mail: tom.claeys@uclouvain.be}, 
I. Krasovsky\footnote{
Department of Mathematics, Imperial College London, Huxley Building,
South Kensington Campus, London SW7\,2AZ, UK; e-mail:i.krasovsky@imperial.ac.uk}, 
O. Minakov
\footnote
{
Department of Mathematical Analysis, Faculty of Mathematics and Physics, Charles University, Sokolovska 8, Prague~8, 186\,75 Czech Republic; e-mail: minakov(at)karlin(dot)mff(dot)cuni(dot)cz}
}

\maketitle
\begin{abstract}
The Freud ensemble of random matrices is the unitary invariant ensemble corresponding to the weight $\exp(-n |x|^{\bt})$, $\bt>0$, on the real line.
We consider the local behaviour of eigenvalues near zero, which exhibits a transition in $\bt$. If $\bt\ge 1$, it is described by the standard sine process.
Below the critical value $\beta=1$, it is described by a process depending on the value of $\bt$, and we determine the first two terms of the large gap probability in it. This so called weak confinement range $0<\bt<1$ corresponds to 
the Freud weight with the indeterminate moment problem. We also find the multiplicative constant in the asymptotic expansion of the Freud multiple integral for $\bt\ge 1$.
\end{abstract}
\tableofcontents

\section{Introduction}
Consider a random matrix ensemble on the space of $n\times n$ Hermitian matrices $M$, defined by the probability measure
\begin{equation}\label{MatrixEnsemble}
\frac{1}{\widehat Z_n}\exp\left(-n\Tr |M|^{\beta}\right)dM,\qquad dM=\prod_{j=1}^{n}dM_{jj}\, \prod_{1\leq j<k\leq n}d\Re M_{jk}d\Im M_{jk},
\end{equation}
where $\beta>0$ is a parameter, $\widehat Z_n=\widehat Z_n(\beta)$ is a normalization constant, and $\Tr |M|^{\beta}=\sum_{j=1}^n|\lambda_j|^\beta$, where $\lambda_1,\ldots, \lambda_n$ are the eigenvalues of $M$. The case $\beta=2$ corresponds to the Gaussian Unitary Ensemble (GUE).
For general $\beta>0$, \eqref{MatrixEnsemble} induces a probability distribution on the eigenvalues of $M$ which is given by (see, e.g., \cite{Deift})
\begin{equation}\label{FreudEnsemble}
p(\lb_1,\dots,\lb_n)\prod_{j=1}^nd\lb_j=
\frac{1}{Z_n}\prod_{1\leq j<k\leq n}|\lambda_j-\lambda_k|^2\prod_{j=1}^ne^{-n|\lambda_j|^\beta}d\lambda_j,\qquad \lambda_1,\ldots, \lambda_n\in\mathbb R,
\end{equation}
where the normalisation constant $Z_n=Z_n(\beta)$ is given by
\be\label{Zdef}
Z_n=
{\int_{\mathbb R^n}\prod_{1\leq j<k\leq n}|\lambda_j-\lambda_k|^2\prod_{j=1}^ne^{-n|\lambda_j|^\beta}d\lambda_j}.
\ee
The probability distribution \eqref{FreudEnsemble} has two qualitative features, namely repulsion of the eigenvalues caused by the factor $\prod_{1\leq j<k\leq n}|\lambda_j-\lambda_k|^2$, and confinement of the eigenvalues around $0$ because of the decay of the weight function $e^{-n|x|^\beta}$ for large $|x|$.

One is interested in the statistical properties of the eigenvalues when $n$ is large. These are well-understood
 when $\beta=2m$, $m=1,2,\dots$, and more generally, when $x^{2m}$ is replaced by an analytic function with sufficient
 growth at infinity and satisfying generic regularity conditions, see e.g. \cite{Deift, DKMVZ2, DKMVZ1}: the eigenvalues for large $n$ tend to
 accumulate on one or several intervals; in the local scaling regime, the behaviour of the eigenvalues near the edges of the intervals is governed by the Airy point process, and near any point inside the intervals (bulk of the spectrum), by the sine process, see below.

 In our case of the ensemble \eqref{MatrixEnsemble}, $\bt>0$, the eigenvalues accumulate on the single interval
 $[-A,A]$, where (in terms of Euler's $B$ and $\Gamma$ functions)
\be\label{defAbeta}
A=A(\beta)=B\left(\frac{\beta}{2},\frac{1}{2}\right)^{\frac{1}{\beta}}=
\left(\frac{\Gamma(\frac{1}{2})\Gamma(\frac{\beta}{2})}{\Gamma(\frac{\beta+1}{2})}\right)^{\frac{1}{\beta}}.
\ee
Locally, the behaviour of the eigenvalues near the points $\pm A$ is described by the Airy point process, and near
any $x\in (-A,A)\setminus\{0\}$, by the sine process for any $\bt>0$. 
However, as observed by Canali, Wallin, and Kravtsov \cite{CWK}
(see also \cite{AV, Canali, Janik, Kanzieper}), the local eigenvalue behaviour near $x=0$ shows remarkably different features if $\bt<1$.
This regime is known as {\em weak confinement} (as opposed to strong confinement if $\bt>1$). Namely, in the weak 
confinement, as $\bt$ decreases from the critical value $\bt=1$, the local statistics of eigenvalues near $x=0$ depend on 
the value of $\bt$ in a crucial way, and one observes a transition between the Wigner-Dyson statistics (which correspond to the sine process) and Poissonian statistics (scaling $\beta$ appropriately). The authors in \cite{CWK} made these observations by numerical simulations
and also by examining the behaviour of the mean eigenvalue density close to $0$. 
Transitions of this type, describing a crossover between strongly repulsive points and independent points or deterministic points, occur also in other models, like Dyson's Brownian motion, finite temperature free fermion models, and random thinnings of determinantal point processes. The ensemble \eqref{FreudEnsemble} has the convenient feature that its general structure (that of an orthogonal polynomial ensemble, see below) is preserved as $\beta>0$ varies.
\medskip

Indeed, the distribution \eqref{FreudEnsemble} for any $\beta>0$ is a determinantal point process on $\mathbb R$ with correlation kernel given in terms of orthogonal polynomials,
\begin{align}\nonumber
p(\lb_1,\dots,\lb_n)&=\frac{1}{n!} \det (K_n(\lb_j,\lb_k))_{j,k=1}^n,
\\
\label{kernel}
K_n(x,y)&=e^{-\frac{n}{2}(|x|^\beta+|y|^\beta)}  \sum_{k=0}^{n-1} P_{k}(x) P_{k}(y)=
e^{-\frac{n}{2}(|x|^\beta+|y|^\beta)}\frac{\chi_{n-1}}{\chi_n}
\frac{P_n(x)P_{n-1}(y)-P_n(y)P_{n-1}(x)}{x-y},
\end{align}
where $P_{k}=P_k^{(\beta)}$ is the degree $k$ normalized orthogonal polynomial with positive leading coefficient $\chi_k=\chi_k^{(\beta)}$ with respect to the weight $e^{-n|x|^\beta}$, characterized by the conditions
\[\int_{-\infty}^{+\infty}P_j(x)P_k(x)e^{-n|x|^\beta}dx=\delta_{jk},\qquad j,k=0,1,\ldots\]
These orthogonal polynomials are known as the Freud polynomials \cite{Freud} (usually after the rescaling $x=n^{-1/\bt}y$), and have a rich history in approximation theory. The weight $e^{-|x|^\beta}$ possesses a remarkable feature: 
for $\bt\ge 1$ the corresponding moment problem is determinate, while for $\bt<1$ it is {\em indeterminate} (see \cite{Berg}).
That is, given the moments $\int_{-\infty}^{+\infty}  x^k d\mu(x)$, $k=0,1,\dots$, the moment problem is determinate
if there exists a unique measure $\mu$ with these moments, and indeterminate if there is more than one such measure.
If $\bt<1$, then 
the measures $e^{-|x|^\beta}dx$ and \[w^{(\gamma)}(x;\beta)dx:=e^{-|x|^\beta}(1+\gamma \cos(|x|^{\bt}\tan(\bt\pi/2))dx\quad\mbox{ with 
any }\gamma\in [-1,1]\] have the same moments.
Thus, in the indeterminate case, the same orthogonal polynomials possess a one-parameter family of orthogonality measures. Moreover, the ensemble \eqref{FreudEnsemble} with $e^{-n|\lambda_j|^\beta}$ replaced by $w^{(\gamma)}(n^{1/\beta}\lambda_j;\beta)$, $\gamma\in[-1,1]$, defines a one-parameter family of determinantal point process, with correlation kernels that are equal up to a prefactor, namely
\begin{equation}\label{Knlambda}K_n^{(\gamma)}(x,y)=\sqrt{(1+\gamma \cos(n|x|^{\bt}\tan(\bt\pi/2))(1+\gamma \cos(n|y|^{\bt}\tan(\bt\pi/2))}\ K_n(x,y).\end{equation}

The large $n$ asymptotics for the recurrence coefficients associated to the Freud polynomials have been studied extensively, see e.g.\ \cite{Magnus, LMS, Nevai}. Kriecherbauer and McLaughlin \cite{KM} established global asymptotics in the complex plane for the Freud polynomials, except in the vicinity of $0$, where their construction is implicit.
Wong and Zhao \cite{WongZhao} contributed to the description of the polynomials near $0$ by expressing relevant quantities in terms of certain integral equations.

The results of \cite{KM} on the asymptotics of the leading coefficients of the Freud polynomials easily imply (see Section \ref{section:FreudHankel})
the large $n$ asymptotics of the Freud multiple integral \eqref{Zdef} but without identification of the constant
term $c(\bt)$:

\begin{proposition}\label{prop:H0}
As $n\to\infty$, there hold the following asymptotics for any $\bt>0$,
\be\label{asZ}
\log \frac{Z_n}{n!}= \left(\log\frac{A}{2}-\frac{3}{2\beta}\right)n^2 +n\log(2\pi) -\frac{1}{12}\log n +c(\beta)+o(1),
\ee
for some $c(\beta)\in\mathbb R$ independent of $n$, and with $A$ given by \eqref{defAbeta}.
\end{proposition}
\begin{remark}
For $\beta=1$, $Z_n$ is, up to a prefactor, a special case of the partition function in the six-vertex model with domain wall boundary conditions, and an expansion equivalent to \eqref{asZ} was in this case already proved in \cite[Theorem 2 ($x=0$)]{BleherBothner}. 
\end{remark}

\medskip
We are interested in the behaviour of the eigenvalues near $0$, in the local scaling regime. To gain insight into it, we consider the probability of an interval (gap) $[-\lambda,\lambda]$ around $0$ without eigenvalues, where $\lambda=\lambda(n)$ decreases with $n$ in a certain way. 
In order to determine the scaling, we need to consider the limiting (large $n$) distribution of the eigenvalues. As is known in a very general setting (see, e.g., \cite{Deift, DeiftKriecherbauerMcLaughlin, SandierSerfaty}), the limiting distribution of the eigenvalues is given as the 
{\em equilibrium measure}: the unique probability measure $\nu$ on $\mathbb R$ minimizing a logarithmic energy functional, which in our case is given by
\begin{equation}
\label{eq:logenergy}\iint_{\mathbb R^2} \log\frac{1}{|x-y|}d\nu(x)d\nu(y)+\int_{\mathbb R} |x|^\beta d\nu(x).
\end{equation}
The measure $\nu=\nu^{(\beta)}$ minimizing this quantity is (see \cite[Chapter IV.5]{SaffTotik} or \cite{AV, CWK, Canali, Janik, Kanzieper})
supported on the interval $[-A,A]$, where $A$ is  defined in \eqref{defAbeta}, and given by
\begin{equation}
\label{eqdensity}
\nu(x)=\psi(x)dx,\qquad
\psi(x)=\frac{\beta |x|^{\beta-1}}{\pi A^{\beta}}\int_{1}^{A/|x|}\frac{u^{\beta-1}}{\sqrt{u^2-1}}du.
\end{equation}
As it must, it satisfies the normalisation
\be\label{normalisation}\nonumber
\int_{-A}^A \psi(x)dx=1.
\ee

In the limit $x\to 0$, we obtain the expansion 
\be\label{exppsi}
\psi(x)\sim\begin{cases}
\frac{\beta}{\beta-1}\frac{1}{\pi A},&\mbox{for $\beta>1$,}\\
\frac{1}{\pi A}\log |x|^{-1},&\mbox{for $\beta=1$,}\\
|x|^{\bt-1}\frac{\bt\tan(\pi\bt/2)}{2\pi},&\mbox{for $0<\beta<1$.}
\end{cases}
\ee

For $\beta>1$, where $\psi(x)$ is constant to the main order near $0$, we can guess the suitable scaling of the local variable $s$ by imposing the condition
\be\label{cond1}
\int_0^{\lambda}\psi(x)dx=\frac{s}{n}(1+o(1)),\qquad n\to\infty,
\ee
which implies that the average number of eigenvalues in $[0,\lambda]$ is approximately equal to $s$ for large $n$. 
Using this condition also for $\bt= 1$ and using \eqref{exppsi},
we define the scaling as follows,
\be\label{scaling-ge1}
\lambda(n,s)=\begin{cases}\frac{s}{n\psi(0)}=
\frac{\beta-1}{\beta}\pi A\frac{s}{n},&\mbox{for $\beta>1$,}\\
\pi A\frac{s}{n\log n},&\mbox{for $\beta=1$.}
\end{cases}
\ee

For $\beta<1$, the function $\psi(x)$ depends on $x$ strongly near $0$, a somewhat similar situation occurs, e.g., in the rescaling for the Airy point process at the edge of the GUE. So, similarly to the Airy point process, we apply a {\it linear} rescaling $s(\lambda)$ corresponding to the condition
\be\label{cond2}
\int_0^{\lambda}\psi(x)dx=\frac{s^{\beta}}{n}(1+o(1)),\qquad \lambda>0,\quad n\to\infty,
\ee
according to which we define the local variable $s$ by
\be\label{scaling-less1}
\lambda(n,s)=\left(\frac{2\pi}{\tan(\pi\beta/2)}\right)^{1/\beta}\frac{s}{n^{1/\beta}},\qquad \beta<1.
\ee

Note that for $\beta=2$, $A=\sqrt{2}$, $\psi(x)=\frac{1}{\pi}\sqrt{2-x^2}$, and \eqref{scaling-ge1} is the usual scaling for GUE at $0$,
$\lambda=\frac{s}{n\psi(0)}$. In this case, it is well-known that the
kernel \eqref{kernel} converges to the sine-kernel $K^{(\beta=2)}$:
\[
\frac{1}{n\psi(0)}K_n\left(\frac{u}{n\psi(0)},\frac{v}{n\psi(0)}\right)\rightarrow K^{(\beta=2)}(u,v)=
\frac{\sin(\pi (u-v))}{\pi(u-v)},
\qquad n\to\infty,
\]
uniformly for $u,v$ in compact subsets of the real line,
and the probability of the gap $(-s,s)$ in the local variable becomes the Fredholm determinant 
\be\nonumber
\det(I-K^{(\beta=2)})_{(-s,s)}=\det(I-K^{(\beta=2)}_s)_{(-1,1)},\qquad K^{(\beta=2)}_s(u,v)=
\frac{\sin(\pi s(u-v))}{\pi(u-v)}.
\ee
Its large $s$ asymptotics are given by the expression \cite{W,DIZ,K04,Ehrhardt,DIKZ} (see also \cite{Kreview}, \cite{FK})
\be\label{beta2exp}
\log\det(I-K^{(\beta=2)})_{(-s,s)}=
-\frac{\pi^2}{2}s^2-\frac{1}{4}\log(\pi s)+3\zeta'(-1)+\frac{1}{12}\log 2+o(1),\qquad s\to\infty,
\ee
where $\zeta'(x)$ is the derivative of the Riemann zeta-function. Below (see Remark \ref{rem:sine}) we provide yet another
proof of \eqref{beta2exp}.

For general $\beta>0$, we have as a corollary of the results in \cite{KM} the following (proved in Section \ref{sec:kernel}).
\begin{proposition}\label{prop:Kernels}
For any $\beta>0$, there exists a function $K^{(\beta)}(u,v)$ independent of $n$ such that 
with the scalings \eqref{scaling-ge1}, \eqref{scaling-less1},
\be\nonumber
\lb(n,1) K_n\left(\lb(n,u),\lb(n,v)\right)\rightarrow K^{(\beta)}(u,v),\qquad n\to\infty,
\ee
uniformly for $u,v$ in compact subsets of the real line,
and $K^{(\beta)}(u,v)=K^{(\beta=2)}(u,v)$ for any $\beta\ge 1$.

Furthermore, the Fredholm determinant $F^{(\beta)}(s)=\det(I-K^{(\beta)})_{(-s,s)}$ exists for any $\bt>0$ and $s>0.$
\end{proposition}

Thus, if $\bt\ge 1$, the kernel \eqref{kernel} converges to the standard sine-kernel $K^{(\beta=2)}(u,v)$ {\it independent} of $\beta$. (Note a connection with the determinacy of the moment problem.)
However, if $0<\bt<1$, the local asymptotics of the Freud polynomials around $0$ are not known explicitly, and therefore
we cannot write the explicit local limit for the kernel in this case: $K^{(\beta)}(u,v)$, $0<\bt<1$, can be written though in terms
of the solution of a Riemann-Hilbert problem \cite[eq. (6.41)--(6.43)]{KM}, whose existence was proven in \cite{KM}.

It follows by standard arguments from Proposition \ref{prop:Kernels} that the limiting kernel $K^{(\beta)}$ also defines a determinantal point process; the Fredholm determinant $F^{(\beta)}(s)$ is the probability of the gap $(-s,s)$ in this point process.
\begin{remark}Observe that for $\beta<1$, with the same scaling, a one-parameter family of limiting kernels \[K^{(\beta,\gamma)}(u,v)=\sqrt{\left(1+\gamma\cos\left(2\pi|u|^\beta\right)\right)\left(1+\gamma\cos\left(2\pi|v|^\beta\right)\right)}\, K^{(\beta)}(u,v)\]
arises if we replace, in Proposition \ref{prop:Kernels}, $K_n$ by $K_n^{(\gamma)}$ (from \eqref{Knlambda}) with $\gamma\in[-1,1]$, each of which defines a determinantal point process.
\end{remark}

We will now see that although $K^{(\beta)}(u,v)$ is not given explicitly for $0<\bt<1$, the 2 leading terms of the large $s$ asymptotics of $F^{(\beta)}(s)$ can be found.

\begin{theorem}\label{thm:Freudlargegap}
As $s\to\infty$, we have the asymptotics
\begin{equation}\label{asF}
\log F^{(\beta)}(s)=\begin{cases} -\frac{\pi^2}{2}s^2-\frac{1}{4}\log s +C+o(1),&\beta\ge 1,\\
\\
-\frac{\beta}{2}B\left(\frac{\beta}{2},\frac{1}{2}\right)^2 s^{2\beta}
-\frac{\beta}{4}\log s+C(\beta)+o(1),&0<\beta<1,
\end{cases}
\end{equation}
for some $C(\beta)\in\mathbb R$ independent of $s$. For $\bt\ge 1$, $C(\bt)=C$ is independent of $\beta$
and equals 
\begin{equation}\label{def:C}C=3\zeta'(-1)+\frac{1}{12}\log 2-\frac{1}{4}\log\pi.\end{equation}
\end{theorem}
\begin{remark}
We see from \eqref{asF} that the large gap probability is larger if $\bt<1$. This effect was noticed 
numerically in \cite{CWK} for the nearest-neighbor spacing distribution.
\end{remark}
\begin{remark}
Each of the 2 main terms in \eqref{asF} is continuous
as a function of $\beta$  at the critical point $\beta=1$. 
\end{remark}
\begin{remark}\label{rem:sine} 
We only use the fact that $K^{(\beta)}(u,v)=K^{(\beta=2)}(u,v)$ for $\beta\ge 1$ in the proof of this theorem
to identify the constant $C(\bt)$ for $\beta\ge 1$ by \eqref{beta2exp}. (The proof of \eqref{asF} with unidentified $C(\bt)$
does not rely on this fact and \eqref{beta2exp}).
Now ignoring this fact, and 
using Propositions \ref{prop:H0} and \ref{prop:Hankel} below, we can write the constant $C(\beta)$ in terms of the unknown (in general) constant $c(\beta)$
 in the expansion \eqref{asZ} of the Freud multiple integral \eqref{Zdef}. If $\beta=2$,  this integral is the well-known Selberg integral related to 
 Hermite polynomials and $c(2)=\zeta'(-1)$, which in turn implies that $C(2)=3\zeta'(-1)+\frac{1}{12}\log 2-\frac{1}{4}\log\pi$. Thus \eqref{asF} reproduces \eqref{beta2exp} in this case (see Appendix A).
 \end{remark}
\medskip

As discussed in Remark \ref{rem:sine}, we can write $C(\beta)$ in terms of the constant $c(\beta)$ in the expansion 
\eqref{asZ}. Using the fact that $C=C(\bt)$ is identified in Theorem \ref{thm:Freudlargegap} for $\bt\ge 1$, we obtain in Section \ref{sec:ProofofThm1} the following

\begin{theorem}\label{thm:const}
The constant in the large $n$ expansion \eqref{asZ} of the Freud multiple integral \eqref{Zdef} is 
\be\nonumber
c(\bt)=\zeta'(-1)-\frac{1}{12}\log \frac{\beta}{2}\,,\qquad \bt\ge 1.
\ee
\end{theorem}
\medskip

We now describe our approach to prove Theorem \ref{thm:Freudlargegap}.
By standard arguments,
\eqref{FreudEnsemble} and Proposition \ref{prop:Kernels} imply that the gap probability $F^{(\beta)}(s)$ is given by the limit
\be\label{Flimit}
F^{(\beta)}(s)=\lim_{n\to\infty} F_n^{(\beta)}(\lambda(n,s)),
\ee
where
$\lambda(n,s)$ is as in \eqref{scaling-ge1}, \eqref{scaling-less1}, and
\be\nonumber
F_n^{(\beta)}(\lambda)=
\frac{1}{Z_n} 
\int_{\left(\mathbb R\setminus[-\lambda,\lambda]\right)^n}\prod_{1\leq j<k\leq n}|\lambda_j-\lambda_k|^2\prod_{j=1}^ne^{-n|\lambda_j|^\beta}d\lambda_j.
\ee
Moreover, by Heine's formula, the $n$-fold integrals above can be expressed in terms of the Hankel determinants
\[
H_n(\lambda):=\det\left(\int_{\mathbb R\setminus[-\lambda,\lambda]}|x|^{j+k}e^{-n|x|^\beta}dx\right)_{j,k=0}^{n-1}=\frac{1}{n!}
\int_{\left(\mathbb R\setminus[-\lambda,\lambda]\right)^n}\prod_{1\leq j<k\leq n}|\lambda_j-\lambda_k|^2\prod_{j=1}^ne^{-n|\lambda_j|^\beta}d\lambda_j.
\]
In particular, $Z_n=n! H_n(0)$, so that
\begin{equation}\label{FnHankel}
F_n^{(\beta)}(\lambda)=\frac{H_n(\lambda)}{H_n(0)}.
\end{equation}

The function $e^{-n|x|^\beta}$ is called the symbol (weight) of the Hankel determinant $H_n(\lambda)$, and the set $\mathbb R\setminus[-\lambda,\lambda]$ its support.
The determinants $H_n(\lambda)$ are closely connected to orthogonal polynomials: we have the well-known formula
\begin{equation}\label{Hnkappan}
H_n(\lambda)=
\prod_{k=0}^{n-1}(\chi_k^{(\lb)})^{-2},
\end{equation}
where $\chi_k^{(\lb)}$ are the leading coefficients of the orthonormal polynomials  $P_k^{(\lambda)}$ with respect to the weight $e^{-n|x|^\beta}$ on $\mathbb R\setminus[-\lambda,\lambda]$.

To prove Theorem \ref{thm:Freudlargegap}  using \eqref{FnHankel}, we will compute the asymptotics of the Hankel determinant $H_n(\lambda(n,s))$ as $n\to\infty$ and $s$ is large.
Here $\lambda(n,s)$ is given by \eqref{scaling-ge1}, \eqref{scaling-less1}. 
It turns out that the leading terms of the asymptotics of $H_n(\lambda)$
for fixed $\lambda$ remain valid also when $\lambda$ tends to 0 as $n\to\infty$ at a rate slower or equal to \eqref{scaling-ge1}, \eqref{scaling-less1}.
In fact, in view of future applications, we obtain a slightly more general result. 

It is easily verified (see Appendix B) that for any $\lambda>0$,
\be\label{HH}
H_{2n}(\lambda)=2^{-4n^2/\beta} \wt H_{n}^{(\alpha=1/2)}(\mu)\wt H_{n}^{(\alpha=-1/2)}(\mu),\qquad \mu=2^{2/\beta}\lambda^2,
\ee
where 
\be\label{def:tildeH}
\wt H_{n}^{(\alpha)}(\mu)=\det\left(\int_{\mu}^{+\infty}x^{j+k}e^{-n|x|^{\beta/2}}x^{\alpha} dx\right)_{j,k=0}^{n-1}.
\ee

The following result about Hankel determinant asymptotics will be the key to prove Theorem \ref{thm:Freudlargegap}. A large part of this paper (Sections \ref{section:3diffid}--\ref{sec:integration}) will be devoted to its proof.

\begin{proposition}\label{prop:Hankel}
Let $\alpha\in\mathbb{R}$, $\beta>0$, and define
\be\label{def_rho}
\rho(\mu)=\frac{\beta}{2\pi}\int_{\mu}^{a(\mu)}t^{\beta/2-1}\frac{dt}{\sqrt{(a(\mu)-t)(t-\mu)}},
\ee
where $a(\mu)>\mu$ is the unique solution of the equation
\be\label{defa}
\frac{\beta}{4\pi}\int_{\mu}^{a(\mu)}t^{\beta/2-1}\sqrt{\frac{t-\mu}{a(\mu)-t}}dt=1.
\ee
For any $\mu_0>0$ there exists $M>0$ such that
\begin{equation}\label{dHs}
\log\wt H_n^{(\alpha)}(\mu) = C_2(\mu)n^2+C_1(\mu)n-\frac{1}{6}\log n +C_0(\mu)+\mathcal O\left(\frac{1}{n\sqrt{\mu}\rho(\mu)}\right)+o(1)
\end{equation}
as $n\to\infty$ uniformly for $\frac{M^2}{n^2\rho(\mu)^2}  \leq \mu\leq \mu_0$,
where $C_j=C_j(\mu)$, $j=0,1,2$ are as follows:
\begin{equation}\label{C1234}\begin{split}
&C_2(\mu) = \log\frac{a(\mu)-\mu}{4}-\frac{3}{\beta}-\frac{1}{\beta}\mu \rho(\mu)-\frac{1}{4\beta}\mu(a(\mu)-\mu) \rho(\mu)^2,
\\
&C_1(\mu) = \log(2\pi)+ 2\alpha\log\frac{\sqrt{a(\mu)}+\sqrt{\mu}}{2}-\frac{2\alpha}{\beta}+\frac{\alpha}{\beta}\sqrt{\mu}(\sqrt{a(\mu)}-\sqrt{\mu})\rho(\mu),
\\
&C_0(\mu) = 2\zeta'(-1)-\frac{1}{6}\log\left(\frac{a(\mu)-\mu}{4} \rho(\mu)\right)+\frac{1}{24}\log a'(\mu)+\frac{\alpha^2}{2}\log\frac{(\sqrt{a(\mu)}+\sqrt{\mu})^2}{4\sqrt{\mu a(\mu)}}.
\end{split}
\end{equation}
\end{proposition}
\begin{remark}\label{remark:strongerestimate1}
One can show 
 (see, in particular, Remark \ref{D_0:full} below)
 that the term $o(1)$ in \eqref{dHs} is in fact $\mathcal{O}\(\frac{\ln n}{n}\)$.
 \end{remark}

\begin{remark}
The asymptotics of Hankel determinants with weights on $[\mu,+\infty)$  for {\it fixed} $\mu$ were obtained for more general weights (for which $x^{\alpha} e^{-x^{\beta/2}}$ is a particular case) by Charlier and Gharakhloo \cite{CG19}.
\end{remark}
\begin{remark}
Explicit small $\mu$ asymptotics for $a(\mu)$ and $\rho(\mu)$ are given in Lemma \ref{lemma:expaf}.
\end{remark}

Substituting the asymptotics of Proposition \ref{prop:Hankel} into \eqref{HH} and taking the scaling \eqref{scaling-ge1}, \eqref{scaling-less1},
it is easy to obtain the asymptotics of $H_{n}(\lambda(n,s))$. Recalling now that the asymptotics of
$H_n(0)=Z_n/n!$ are given by Proposition \ref{prop:H0} we can complete the proof of Theorem \ref{thm:Freudlargegap}
by \eqref{FnHankel}. 
We do this in Section \ref{sec:ProofofThm1}. 

Our main task will therefore be to prove Proposition \ref{prop:Hankel}.
We proceed with this proof as follows.
In  Section \ref{section:3diffid} we derive a differential identity, expressing the logarithmic $\mu$-derivative of 
$\wt H_n^{(\alpha)}(\mu)$ in terms of the related orthogonal polynomials $p_k^{(\alpha,\beta,\mu)}$. Then, we obtain asymptotics for $p_k^{(\alpha,\beta,\mu)}$ via a Riemann-Hilbert (RH) approach. Here we need to distinguish between $3$ different regimes. In Section \ref{section:6RH1}, we use a standard RH approach, which leads us to results when $\mu$ can tend to $0$ slowly with $n$. In Section \ref{section:7RH2}, we refine the RH analysis to allow $\mu$ to tend to $0$ faster. This first refinement, which involves an unusual normalization of the RH problem, allows us to obtain good enough asymptotics when $\beta>1$, but not when $\beta\leq 1$. In Section \ref{section:7RH2}, we resolve this issue using a further improvement of the RH analysis, and obtain the asymptotics of the polynomials in the full regime of Proposition \ref{prop:Hankel}. 
Finally, in Section \ref{section:9proofs}, we substitute the asymptotics from the RH analysis into the differential identity.
To integrate the differential identity, we need a suitable starting point where the full asymptotics of $\wt H_n^{(\alpha)}(\mu)$ are known or can be determined independently. 
In our case, the starting point is a fixed 
value of $\mu=\mu_0>0,$ where we use the result of \cite{CG19}.
We then integrate the  differential identity
in Section \ref{sec:integration} and obtain the expressions for the $C_j$'s.\footnote{The reason why we use a differential identity
rather than the formula of type \eqref{Hnkappan} to determine the asymptotics is that \eqref{Hnkappan} does not allow
to determine $C_0$ (cf. the proof of Proposition \ref{prop:H0}), but this coefficient is crucial for the proof of Theorem \ref{thm:Freudlargegap}.}

\section{Convergence of kernels and determinants. Proof of Proposition \ref{prop:Kernels}}\label{sec:kernel}
In \cite{KM}, the authors considered the orthonormal polynomials 
$\phi_k(x)=\gamma_k x^k+\dots$, $k=0,1,\dots$, with respect to the weight $e^{-|A x|^\beta}$ on $\mathbb R$.
Note that 
by the change of variable $y=n^{-\frac{1}{\beta}}A x$ in the orthonormality condition for the polynomials $\phi_k(x)$, with $A$ given by \eqref{defAbeta},
we obtain the identity
\be\label{Pphi}
P_k(y)=\frac{n^{\frac{1}{2\beta}}}{\sqrt{A}}\phi_k\left(n^{\frac{1}{\beta}}A^{-1} y\right),\qquad k=0,1,\dots,
\ee
where $P_k(y)=\chi_k y^k+\cdots$ are, as before, the orthonormal polynomials with respect to the weight $e^{-n|y|^\beta}$
on $\mathbb R$.
For the leading coefficients, we thus have
\begin{equation}\label{eq:chigamma}
\chi_k=\frac{n^{\frac{2k+1}{2\beta}}}{A^{k+\frac{1}{2}}}\gamma_k.
\end{equation}

It was shown in
\cite[Theorem 1.5]{KM} that for any $\beta>0$,
\be\label{KMformula}
\begin{aligned}
\log\gamma_m&=\log\gamma_m^{(approx)}+
\mathcal O\left(\frac{1}{m\log^2m}\right),\qquad m\to\infty,\\
\log\gamma_m^{(approx)}&=
-\frac{1}{2}\log\pi -\frac{1}{\beta}\left(m+\frac{1}{2}\right)\log m +\frac{m}{\beta}+m\log 2+\frac{\beta-4}{24\beta m}.
\end{aligned}
\ee

Moreover, it follows from \cite[Theorem 1.16]{KM} that uniformly for $z$ in a fixed $\varepsilon$-neighbourhood of $0$ as $n\to\infty$, for any $\bt\ge 1$,
\be\label{1asymp}
e^{-\frac{n}{2}|z|^{\bt}}\phi_n(n^{1/\bt}z/A)=
\sqrt{\frac{2}{\pi n^{1/\bt}}}\cos\left(\pi n \int_A^z\psi(t)dt+\frac{1}{2}\arcsin(z/A)\right)(1+\mathcal O(z^2)+\mathcal O(1/\log n)),
\ee
with $\psi$ as in \eqref{eqdensity}.
On the other hand,
for any $0<\bt<1$, for $z$ in the first quarter of the $\varepsilon$-neighbourhood of $0$,
\be\label{2asymp}
\begin{aligned}
e^{-\frac{n}{2}|z|^{\bt}}\phi(n^{1/\bt}z/A)&=
\sqrt{\frac{2}{\pi n^{1/\bt}}}(1-z/A)^{-1/4}(1+z/A)^{-1/4}\\
&\times\left\{
\cos\left(\pi n \int_A^z\psi(t)dt+\frac{1}{2}\arcsin(z/A)\right)((\Phi_{\bt,n})_{11}(z/A)+r_{1,\bt}(n,z))\right.\\
&\left.+\sin\left(\pi n \int_A^z\psi(t)dt-\frac{1}{2}\arcsin(z/A)\right)(i(\Phi_{\bt,n})_{12}(z/A)+r_{2,\bt}(n,z))
\right\},
\end{aligned}
\ee
where the $2\times 2$ matrix $\Phi_{\beta,n}(z)$ takes the form \[\Phi_{\bt,n}(z)=N(z)L((2n d_{\bt})^{1/\bt}iz)N(z)^{-1},\] where the matrix-valued function $L(w)$
is the solution to the RH problem \cite[(6.41)--(6.43)]{KM}.
The constant $d_{\bt}$ is given in \cite{KM} explicitly, as well as the matrix $N(z)$, which satisfies $N(z)=\mathcal O(1)$ as $z\to 0$. The error terms $r_{1,\bt}(n,z)$, $r_{2,\bt}(n,z)$ are bounded by $\mathcal O(n^{-1})+\mathcal O(n^{1+\delta-1/\beta})$ as $n\to\infty$ for any $\delta>0$, uniformly for $z$ in the first quarter of the $\varepsilon$-neighbourhood of $0$.

Let first $\bt > 1$. We recall the scaling \eqref{scaling-ge1}, denoting 
\[
x=\frac{u}{n\psi(0)},\qquad y=\frac{v}{n\psi(0)}.
\]
Now use the formulae \eqref{Pphi}--\eqref{1asymp} above and substitute $P_n(x)$, $P_n(y)$, $P_{n-1}(x)$, $P_{n-1}(y)$
into \eqref{kernel}. Recall that in our scaling, by \eqref{cond1},
\[
\int_A^x\psi(t)dt =\frac{1}{2}-\int_0^x\psi(t)dt=\frac{1}{2}-\frac{u}{n}+o(1/n),\qquad n\to\infty,
\]
and similarly for $y$, where the error term is uniform for $u,v$ in compact subsets of the real line.
We then obtain, with the same uniformity,
\[
\frac{1}{n\psi(0)}K_n\left(\frac{u}{n\psi(0)},\frac{v}{n\psi(0)}\right)\rightarrow K^{(\beta=2)}(u,v)=
\frac{\sin(\pi (u-v))}{\pi(u-v)},
\qquad n\to\infty.
\]

We easily obtain the convergence to the sine kernel  for $\bt=1$ as well, using the appropriate scaling in \eqref{scaling-ge1}. Given the different scaling, this sine kernel convergence may seem surprising at first sight. Note however that this is consistent with the analysis from \cite{BleherBothner}.

If $\bt<1$, one verifies, using in particular \eqref{2asymp}, the existence of the limiting kernel, which is expressed in terms of $L(w)$. 
By standard arguments, see e.g.\ \cite{DG}, one can show that the uniform on compacts convergence of the kernels implies convergence of the corresponding Fredholm determinants on $[-s,s]$, and in particular the existence of the Fredholm 
determinant $F^{\bt}(s)$.

\section{Asymptotics of $F^{(\beta)}(s)$. Proof of Theorems \ref{thm:Freudlargegap} and \ref{thm:const}}\label{sec:ProofofThm1}

We will now prove Theorems  \ref{thm:Freudlargegap}, \ref{thm:const}, relying on Propositions \ref{prop:H0} and \ref{prop:Hankel}. 
First we need small $\mu$ expansions for $a(\mu)$ and $\rho(\mu)$ defined in \eqref{defa} and \eqref{def_rho}.

\begin{lemma}\label{lemma:expaf}
As $\mu\to 0$,
\be\nonumber
a(\mu)=\begin{cases}
a_0+a_1\mu+\mathcal O(\mu^{(\beta+1)/2})+\mathcal O(\mu^2),& \beta>1,\\
4\pi^2-\frac{1}{2}\mu\log\mu+\mathcal O(\mu),& \beta=1,\\
a_0+\wt a_1\mu^{(\beta+1)/2}+\mathcal O(\mu),& \beta<1,
\end{cases}
\ee
where
\be\label{afcoeff}
a_0=(2B(\beta/2,1/2))^{2/\beta}=2^{2/\beta} A^2,\qquad a_1=\frac{1}{\beta-1},\qquad \wt a_1=\frac{\sqrt{a_0}}{2\pi(\beta+1)} B(\beta/2,1/2) \tan\frac{\pi\beta}{2}.
\ee
Furthermore, with this $a_0$, as $\mu\to 0$,
\be\label{fcoeff}
\rho(\mu)=\begin{cases}
\frac{2\beta}{\beta-1}\frac{1}{a_0}+\mathcal O(\mu^{(\beta-1)/2})+\mathcal O(\mu),& \beta>1,\\
-\frac{1}{4\pi^2}\log\mu+\mathcal O(1),& \beta=1,\\
\frac{\beta}{2\pi\sqrt{a_0}}B((1-\beta)/2,1/2){\mu^{(\beta-1)/2}}+\mathcal O(1),& \beta<1.
\end{cases}
\ee
\end{lemma}

\begin{proof}
First of all, making the change of variable $t=\mu+(a(\mu)-\mu) y$ in the integral in  \eqref{defa}, we see that for every $\mu>0,$ the left hand side there depends monotonically on $a,$ and hence equation \eqref{defa} indeed defines a unique $a=a(\mu)\in(\mu,+\infty).$

Next, 
to obtain the expansion of $a(\mu)$, we split the integral in \eqref{defa} into two parts: the first part contains the contribution of the interval $(\mu,\sqrt{\mu})$ and the second part is the remaining part of the integral 
over $(\sqrt{\mu},a(\mu))$.
We then apply the change of variable $t=\mu u$ in the first, and $t=a(\mu) u$ in the second, and expand the integrands in small parameters
using the expression
\[
B(\gamma,\delta)=\int_0^1 x^{\gamma-1}(1-x)^{\delta-1}dx.
\]
The expansion for $\rho(\mu)$ is found in a similar way.
\end{proof}

To prove Theorem \ref{thm:Freudlargegap}, by \eqref{Flimit} and \eqref{FnHankel}, it is sufficient to estimate the ratio $H_n(\lambda(n,s))/H_n(0)$,
where $\lambda(n,s)$ is as in \eqref{scaling-ge1}, \eqref{scaling-less1}. By \eqref{HH}, \eqref{asZ}, and  Proposition \ref{prop:Hankel}, we have, denoting (note the doubling of $n$)
\[
\mu(n,s)=2^{2/\beta}\lambda(2n,s)^2,
\]
that 
\be\label{logHH}
\begin{aligned}
\log\frac {H_{2n}(\lambda(2n,s))}{H_{2n}(0)}&=-\frac{4n^2}{\beta}\log 2+ \log\left[\wt H_{n}^{(\alpha=1/2)}(\mu(n,s))
\wt H_{n}^{(\alpha=-1/2)}(\mu(n,s))\right]-\log H_{2n}(0)\\
&=-\frac{4n^2}{\beta}\log 2+
2C_2(\mu(n,s))n^2+2n\log(2\pi)-\frac{1}{3}\log n+ C_0^{(\alpha=1/2)}(\mu(n,s))\\&\qquad +C_0^{(\alpha=-1/2)}(\mu(n,s))
-\left(\log\frac{A}{2}-\frac{3}{2\beta}\right)4n^2 -2n\log(2\pi)
+\frac{1}{12}\log(2n)\\&\qquad  -c(\beta)+
{\mathcal O\left(\frac{1}{n\sqrt{\mu(n,s)}\rho(\mu(n,s))}\right)+o(1),}\end{aligned}
\ee
{as $n\to\infty$,
provided that $\mu(n,s)$ converges to $0$ and $n\rho(\mu(n,s))\sqrt{\mu(n,s)}$ is sufficiently large.}
We already see that the terms of order $n$ cancel out, and that { in the limit $n\to\infty$ the error terms become $\mathcal{O}\(s^{-1}\)+o(1)$ when $\beta\geq 1$ and $\mathcal{O}\(s^{-\beta}\)+o(1)$ when $0<\beta\leq1,$ for $s$ sufficiently large.}
 We further consider different ranges of $\beta$ separately.
 \bigskip
 
\noindent
1) Let $\beta>1$. By \eqref{scaling-ge1}, recalling the definition of $a_0$ in \eqref{afcoeff}, we have
\[
\mu(n,s)=2^{2/\beta}\lambda(2n,s)^2=\frac{\pi^2 a_0}{4}\left(\frac{\beta-1}{\beta}\right)^2 \frac{s^2}{n^2}.
\]
As $n\to\infty$, for large $s$ fixed, we can use the expansions  of Lemma \ref{lemma:expaf} for $a(\mu(n,s))$ and $\rho(\mu(n,s))$, and substituting them into \eqref{C1234} we obtain
\[\begin{aligned}
&C_2(\mu(n,s))=\log\frac{a_0}{4}+\frac{a_1-1}{a_0}\mu(n,s)-\frac{3}{\beta}-
\frac{\mu(n,s)}{\beta a_0}\left(\frac{2\beta}{\beta-1}+\left[\frac{\beta}{\beta-1}\right]^2\right)+o(1/n^2),\\
&C_0(\mu(n,s))^{(\alpha=1/2)}+C_0(\mu(n,s))^{(\alpha=-1/2)}=\frac{1}{4}\log n - \frac{1}{4}\log s +\widehat c(\beta) +o(1),\qquad n\to\infty,
\end{aligned}
\]
where
\be\label{whc}
\widehat c(\beta)=4\zeta'(-1)-\frac{1}{12}\log\frac{\bt}{2}-\frac{1}{4}\log\pi.
\ee
Substituting these expressions into \eqref{logHH} and taking the limit $n\to\infty$, we prove \eqref{asF}
for $\beta>1$ with
\be\label{largeC}
C(\bt)=\widehat c(\beta)+\frac{1}{12}\log 2 -c(\bt).
\ee
On the other hand, recalling that $K^{(\beta)}(u,v)=K^{(\beta=2)}(u,v)$ for $\beta\ge 1$, and recalling the expansion  \eqref{beta2exp} for $\beta=2$,
we immediately have an alternative proof of \eqref{asF} for $\beta\ge 1$ with the explicit constant $C(\bt)=3\zeta'(-1)+\frac{1}{12}\log 2-\frac{1}{4}\log\pi$. Furthermore, substituting this value into \eqref{largeC} and using \eqref{whc}, we prove
Theorem \ref{thm:const} for $\bt>1$.
 \bigskip

\noindent
2) Let $\beta<1$. By \eqref{scaling-less1}, 
\[
\mu=\mu(n,s)=2^{2/\beta}\lambda(2n,s)^2=\left(\frac{4\pi}{\tan(\pi\beta/2)}\frac{1}{2n}\right)^{2/\beta} s^2.
\]
Note that $\mu^{(\beta+1)/2}n^2\to 0$, and also  $\mu n^2\to 0$ as $n\to\infty$. Moreover, since by \eqref{fcoeff} 
$\rho(\mu)=\mathcal O(\mu^{(\beta-1)/2})$, we have that $\mu \rho(\mu) n^2\to 0$. However, $\mu \rho(\mu)^2 n^2$ in $C_2n^2$ gives a nonzero contribution in the limit, namely,
\be\label{muf2}\begin{aligned}
\mu(n,s) \rho(\mu(n,s))^2 n^2 &=\frac{\beta^2}{4\pi^2 a_0}B((1-\beta)/2,1/2)^2\mu(n,s)^{\beta}n^2 (1+o(1))\\
&= \frac{\beta^2}{a_0} B(\bt/2,1/2)^2 s^{2\beta}(1+o(1)),
\end{aligned}
\ee
as $n\to\infty$, where we used the reflection formula for the $\Gamma$ function,
\[
\Gamma(x)\Gamma(1-x)=\frac{\pi}{\sin(\pi x)}.
\]
Using \eqref{muf2}, we obtain
\be\label{C2less1}
C_2(\mu(n,s))=\log\frac{a_0}{4}-\frac{3}{\beta}-
\frac{\beta}{4} B(\bt/2,1/2)^2 
\frac{s^{2\beta}}{n^2}+o(1/n^2).
\ee
Furthermore, one easily obtains
\be\label{C0less1}
\begin{aligned}
&C_0(\mu(n,s))^{(\alpha=1/2)}+C_0(\mu(n,s))^{(\alpha=-1/2)}
=-\frac{1}{3}\log\mu^{(\beta-1)/2}+\frac{1}{12}\log\mu^{(\beta-1)/2}-\frac{1}{8}\log\mu\\
&+4\zeta'(-1) -\frac{1}{3}\log \left(\frac{a_0}{4} \frac{\beta}{2\pi\sqrt{a_0}}B((1-\beta)/2,1/2)\right)+
\frac{1}{12}\log\left(\frac{\bt+1}{2}\wt a_1\right)+\frac{1}{4}\log\frac{\sqrt{a_0}}{4}+o(1)\\
&=\frac{1}{4}\log n - \frac{\beta}{4}\log s +
4\zeta'(-1)-\frac{1}{4}\log B(\beta/2,1/2)-\frac{1}{3}\log\bt+\frac{1}{12}\log 2+o(1).
\end{aligned}
\ee
Substituting  \eqref{C2less1} and \eqref{C0less1} into \eqref{logHH} and taking the limit $n\to\infty$ we prove the statement of Theorem \ref{thm:Freudlargegap} for $\beta<1$.
\bigskip

\noindent
3) Let $\beta=1$. 
By \eqref{scaling-ge1}, 
\[
\mu(n,s)=4\lambda(2n,s)^2=\left(\frac{\pi^2 s}{n\log(2n)}\right)^2.
\]
We have $\mu n^2\to 0$, $(\mu\log\mu)  n^2\to 0$, $\mu \rho(\mu) n^2\to 0$, but
\[
\mu \rho(\mu)^2 n^2=\frac{n^2}{16\pi^4}(\mu\log^2\mu )(1+o(1))=
\frac{s^2}{4}(1+o(1)),\qquad n\to\infty.
\]
Note that
\be\label{Co=1}
C_0(\mu)^{(\alpha=1/2)}+C_0(\mu)^{(\alpha=-1/2)}=\frac{1}{4}\log n - \frac{1}{4}\log s +
\widehat c(1)+o(1),\qquad n\to\infty,
\ee
where $\widehat c(\bt)$ is given by \eqref{whc}.
Using the expressions above, we obtain \eqref{asF} with undetermined $C(1)$.

On the other hand, we have an immediate alternative proof which also provides the value of $C(1)$, and hence the proof of 
Theorem \ref{thm:const} for $\bt=1$ as above, by the fact that $K^{(\bt=1)}(u,v)=K^{(\bt=2)}(u,v)$.

This completes the proof of Theorems \ref{thm:Freudlargegap} and \ref{thm:const}.

\section{Hankel determinant for the Freud weight. Proof of Proposition \ref{prop:H0}}
\label{section:FreudHankel}
Asymptotics of the leading coefficients of the orthonormal polynomials allow to find the asymptotics
of the corresponding Hankel determinant up to an undetermined multiplicative constant, cf. \cite{KMVV}.
Recall the expression \eqref{KMformula} for the leading coefficients of the polynomials $\phi_k(z)$.
The corresponding Hankel determinant $\widehat H_n$ with symbol $e^{-|A x|^\beta}$ on $\mathbb R$ is then
\be\nonumber
\widehat H_n=\prod_{k=0}^{n-1} \gamma_k^{-2}=\prod_{k=1}^{n-1} (\gamma_k^{(approx)})^{-2}\prod_{k=1}^{n-1} 
\left(\frac{\gamma_k}{\gamma_k^{(approx)}}\right)^{-2}\gamma_0^{-2}.
\ee
Note that the second product here converges as $n\to\infty$ by the error term estimate in \eqref{KMformula}, and the first one
is straightforward to estimate for large $n$ using the expansions:
\begin{align*}
&\sum_{k=1}^{n-1}k\log k=\frac{n^2}{2}\log n-\frac{n^2}{4}-\frac{n}{2}\log n+\frac{1}{12}\log n +\frac{1}{12}-\zeta'(-1)+\mathcal O(1/n),\\
&\sum_{k=1}^{n-1}\log k=\log n! -\log n=n\log n-n-\frac{1}{2}\log n +\frac{1}{2}\log(2\pi)+\mathcal O(1/n).
\end{align*}
We then obtain
\be\label{whH}
\log\widehat H_n=\frac{1}{\beta}n^2\log n -\left(\log 2 +\frac{3}{2\beta}\right)n^2 +n\log(2\pi) -\frac{1}{12}\log n +c(\beta)+o(1),\qquad n\to\infty,
\ee
for some constant $c(\beta)$.

To obtain the asymptotics for the determinant $H_n(0)$ with symbol $e^{-n|y|^{\beta}}$ and thus prove the proposition, we use the relation \eqref{eq:chigamma} between the appropriate leading coefficients, 
which implies that 
\[
H_n(0)=\prod_{k=0}^{n-1}\chi_k^{-2}=\widehat H_n \frac{A^n}{n^{n/\beta}}\prod_{k=0}^{n-1}\left(\frac{A}{n^{1/\beta}}\right)^{2k}.  
\]
Using this identity and \eqref{whH}, we prove the proposition.

\section{Differential identity}\label{section:3diffid}

In this section, we establish an identity for the logarithmic $\mu$-derivative of $\widetilde H_n(\mu)=\widetilde H_n^{(\alpha)}(\mu)$, given by \eqref{def:tildeH}, in terms of the orthonormal polynomials $p_n^{(\mu)}$ with respect to the weight $w_n(x):=x^\alpha e^{-n x^{\beta/2}}$ on $(\mu,\infty)$. Note that the polynomials $p_n$ depend not only on $\mu$ but also on $\alpha, \beta$, although we omit this in our notation. We denote $\kappa_n(\mu), \kappa_{n-1}(\mu)$ for the leading coefficients of $p_n^{(\mu)}, p_{n-1}^{(\mu)}$.

\begin{proposition}\label{prop_dif_ident}\textbf{(Differential identity)}
For any $\alpha>-1$, $\beta>0$, and $\mu>0$, we have the identity
\begin{equation}\label{DI1}\partial_\mu\ln \widetilde H_n(\mu) = -w_n(\mu)\frac{\kappa_{n-1}(\mu)}{\kappa_n(\mu)}\left((\partial_x p_n^{(\mu)}(x))p_{n-1}^{(\mu)}(x)-p_n^{(\mu)}(x)(\partial_x p_{n-1}^{(\mu)}(x)\right)\Big|_{x=\mu}.\end{equation}
Denoting
\begin{equation}\label{def:Y}Y(z)\equiv Y(z;\mu)=\begin{pmatrix}
\kappa_n(\mu)^{-1}p_n^{(\mu)}(z) & \dfrac{1}{2\pi\i\kappa_n(\mu)}\displaystyle
\int_{\mu}^{+\infty}\dfrac{p_n^{(\mu)}(u)w_n(u)\d u}{u-z}
\\\\
-2\pi\i\;\kappa_{n-1}(\mu) p_{n-1}^{(\mu)}(z) & -\kappa_{n-1}(\mu)\displaystyle\int_{\mu}^{+\infty}\dfrac{p_{n-1}^{(\mu)}(u)w_n(u) \d u}{u-z}
\end{pmatrix},\end{equation}
we can alternatively write this as
\begin{equation}\label{DI2}\partial_\mu\ln \widetilde H_n(\mu)=-\dfrac{w_n(\mu)}{2\pi\i}\(Y^{-1}(x;\mu)\partial_xY(x;\mu)\)_{21}\Big|_{x=\mu}.\end{equation}
\end{proposition}
\begin{proof}
The proof \clr{(cf. \cite{Kduke})} is based on the confluent form of the Christoffel-Darboux formula,
$$\sum\limits_{j=0}^{n-1}p_j^{(\mu)}(x)^2=\frac{\kappa_{n-1}}{\kappa_n}\((\partial_x p_n^{(\mu)}(x))p_{n-1}^{(\mu)}(x)-p_n^{(\mu)}(x)(\partial_x p_{n-1}^{(\mu)}(x)\),$$
and the relation
$$\widetilde H_n(\mu)=\prod\limits_{j=0}^{n-1}\kappa_j(\mu)^{-2}.$$
From the latter we obtain
$$\partial_\mu\ln \widetilde H_n(\mu) = -2\sum\limits_{j=0}^{n-1}\frac{\partial_\mu\kappa_j(\mu)}{\kappa_j(\mu)}.$$
From the orthonormality of $p_n^{(\mu)}(x)$ we obtain
$$\int_\mu^{+\infty}p_j^{(\mu)}(x)\(\partial_\mu p_j^{(\mu)}(x)\)\; w_n(x)\d x = \frac{\partial_\mu\kappa_j(\mu)}{\kappa_j(\mu)},$$
and substituting this into the previous formula, and continuing the calculations, we obtain further
$$\partial_\mu\ln \widetilde H_n(\mu) = -\int_{\mu}^{+\infty}\partial_\mu\(\sum\limits_{j=0}^{n-1}p_j^{(\mu)}(x)^2\)
w_n(x)\d x.$$
By the orthogonality of the polynomials, we have
$$0=
\partial_\mu
\underbrace{\left[\int_\mu^{+\infty}\(\sum\limits_{j=0}^{n-1}p_j^{(\mu)}(x)^2\)
w_n(x)\d x\right]}_{=n}
=
\int_{\mu}^{+\infty}\partial_\mu\(\sum\limits_{j=0}^{n-1}p_j^{(\mu)}(x)^2\)
w_n(x)\d x-
w_n(\mu)\sum\limits_{j=0}^{n-1}p_j^{(\mu)}(\mu)^2
,
$$
hence
\[\partial_\mu\ln \widetilde H_n(\mu) = - w_n(\mu)\sum\limits_{j=0}^{n-1}p_j^{(\mu)}(x)^2,
\]
and using the confluent form of the Christoffel-Darboux formula, we obtain the first identity \eqref{DI1}. The second identity \eqref{DI2} follows 
by using the fact that the determinant of $Y$ is identically equal to $1$ (this is a standard consequence of the RH conditions for $Y$ given in Section \ref{section:6RH1}),
by writing out the right hand side in terms of the entries $Y_{11}$ and $Y_{21}$ of $Y$, and by substituting their values given in \eqref{def:Y}.
\end{proof}

\section{The $g$-function and equilibrium measure}\label{sec:gfunction}

\subsection{Construction}
In this section we construct and study a $g$-function and an associated function $\varphi(z)$, which will play a fundamental role in the asymptotic RH analysis in the following sections.

It is well understood that the asymptotic analysis of the function $Y$ from \eqref{def:Y} involves an
equilibrium measure, characterized as the minimizer of a logarithmic energy functional, and coinciding with the limiting zero distribution of the orthogonal polynomials, see e.g.\ \cite{Deift}.
For the polynomials $p_n^{(\mu)}$ under consideration here, the relevant equilibrium problem is to minimize 
\[
\iint\ln\frac{1}{|x-y|}\d\nu(x)\d\nu(y)+\int x^{\beta/2}\d\nu(x)
\]
among all probability measures $\nu$ supported on a subset of $[\mu,+\infty).$ Assuming that $\nu$ is supported on a single interval $[\mu,a(\mu)]$, the associated Euler-Lagrange  conditions (see e.g. \cite{SaffTotik}) state that there exists a constant $\ell_\mu$ such that 
\begin{equation}\label{eq:EL}
\begin{split}
&2\int\ln|x-y|\d\nu(y)-x^{\beta/2}=2\ell_{\mu},\quad \mbox{ for }x\in[\mu,a(\mu)],
\\
&2\int\ln|x-y|\d\nu(y)-x^{\beta/2}<2\ell_{\nu},\quad \mbox{ for }x>a(\mu).
\end{split}
\end{equation}
If we can construct a probability measure $\nu$ satisfying these properties for some $\ell_\mu$ and for some $a(\mu)$, we can conclude that $\nu$ is the unique equilibrium measure minimizing the above energy functional (and hence that the one-interval assumption was correct).
Rather than constructing the measure $\nu$ directly, we will construct the associated $g$-function
\begin{equation}\label{def:glog}
g(z;\mu)=\int \log(z-x)d\nu(x),
\end{equation}
corresponding to the principal branch of the logarithm.
Note that the Euler-Lagrange conditions are equivalent to $g_+(x;\mu)+g_-(x;\mu)=x^{\beta/2}+2\ell_{\mu}$ for $x\in[\mu,a(\mu)]$ 
and 
$g_+(x;\mu)+g_-(x;\mu)<x^{\beta/2}+2\ell_{\mu}$ for $x>a(\mu).$ 
If these conditions are satisfied, the Stieltjes transform $g'(z;\mu)=\int \frac{d\nu(x)}{z-x}$ needs to satisfy the identity
$g_+'(x;\mu)+g_-'(x;\mu)=\frac{\beta}{2}x^{\frac{\beta}{2}-1}$ for $x\in [\mu,a(\mu)].$
The Sokhotsky-Plemelj formula and the fact that $g'$ tends to $0$ at infinity then imply that $g'$ takes the form
\be\label{gprime1}
g'(z;\mu) = -\frac{\beta}{4\pi}\(\frac{z-a(\mu)}{z-\mu}\)^{\frac12} \int\limits_{\mu}^{a(\mu)}
\sqrt{\frac{x-\mu}{a(\mu)-x}}
\frac{x^{\beta/2-1} \d x}{x-z},
\ee
where $\(\frac{z-a(\mu)}{z-\mu}\)^{\frac{1}{2}}$ 
is chosen analytic in $\mathbb C\setminus[\mu,a(\mu)]$ and tending to $1$ at infinity, and where $\sqrt{.}$ is the positive square root.
The Stieltjes inversion formula yields
\be\label{psi}
d\nu(x)=\psi(x;\mu)dx\ \mbox{ with }\ \psi(x;\mu) = 
\frac{\beta}{4\pi^2}
\sqrt{\frac{a(\mu)-x}{x-\mu}}
\dashint_{\mu}^{a(\mu)}\sqrt{\frac{t-\mu}{a(\mu)-t}}
\frac{t^{\beta/2-1}\d t}{t-x}
,
\quad x\in(\mu,a(\mu)).
\ee
By construction, the first of the Euler-Lagrange conditions, i.e.\ the equality in \eqref{eq:EL}, is satisfied. Moreover, if $a(\mu)>\mu$ is such that 
\eqref{defa} holds, we have $g'(z)\sim 1/z$ as $z\to\infty$, hence $\int d\nu(x)=1$. This also implies that $g_+-g_-=2\pi i$ on $(-\infty,\mu)$.
Since 
\[\dashint_{\mu}^{a(\mu)}\frac{1}{\sqrt{(t-\mu)(a(\mu)-t)}}
\frac{\d t}{t-x}=0,\]
which is easily seen by a contour deformation argument applied to an associated Cauchy integral, we can alternatively write $\psi$ as the non-singular integral
\be\label{psialter}\psi(x;\mu) = 
\frac{\beta}{4\pi^2}
\sqrt{\frac{a(\mu)-x}{x-\mu}}
\int_{\mu}^{a(\mu)}\frac{(t-\mu)t^{\beta/2-1}-(x-\mu)x^{\beta/2-1}}{(t-x)\sqrt{(t-\mu)(a(\mu)-t)}}
dt,\ee
which implies that $\psi$ is positive on $(\mu,a(\mu))$ and hence that $\nu$ is a probability measure.

In order to prove that $\nu$ is the equilibrium measure we are looking for, we still need to prove that the inequality in \eqref{eq:EL} is satisfied.
To see this, it is convenient to introduce 
\be\label{def:phi}
\varphi(z;\mu) = 
-\frac12 z^{\frac{\beta}{2}} + g(z;\mu) -\ell_{\mu},\qquad z\in\mathbb C\setminus(-\infty,a(\mu)],
\ee
where $z^{\beta/2}$ corresponds to arguments in $(-\pi,\pi)$.
From the properties of $g$ discussed above, we obtain directly that
\be\label{phippphim}
\varphi_+(x;\mu)-\varphi_-(x;\mu) = 2\pi\i \ \mbox{ for } x\in(0, \mu)
\quad\mbox{ and }\quad
\varphi_+(x;\mu)+\varphi_-(x;\mu) = 0 \ \mbox{ for } x\in(\mu, a(\mu)).\ee
From this we conclude that
\be\label{varphi_limit}2\varphi_+(x;\mu) = \varphi_+(x;\mu)-\varphi_-(x;\mu) = g_+(x;\mu)-g_-(x;\mu)=\clu{2\pi}\i\int_x^{a(\mu)}\psi(\xi;\mu)\d\xi \quad \mbox{ for } x\in(\mu, a(\mu)).
\ee
Substituting \eqref{psialter} and extending the resulting expression analytically, we get
\be\label{varphi}
\varphi(z;\mu) = -\frac{1}{2}
\int_{a(\mu)}^{z}
\(\frac{\xi-a(\mu)}{\xi-\mu}\)^{\frac12}
f(\xi;\mu)\d \xi,\quad z\in\mathbb{C}\setminus(-\infty, a(\mu)],
\ee
where $f$ is given by
\be\label{freg}
f(z;\mu) = \frac{\beta}{2\pi}
\int_{\mu}^{a(\mu)}
\frac{(t-\mu)t^{\frac{\beta}{2}-1} - (z-\mu)z^{\frac{\beta}{2}-1} }
{(t-z)\sqrt{(a(\mu)-t)(t-\mu)}}\d t,\quad z\in\mathbb{C}\setminus(-\infty,0],
\ee
with $z^{\frac{\beta}{2}-1}$ analytic in $\mathbb C\setminus(-\infty,0]$ and positive for $z>0$.
Observe that $f(\mu;\mu)$ is equal to the function $\rho(\mu)$ defined in \eqref{def_rho},
\be\label{f_rho}f(\mu;\mu) = \rho(\mu).\ee
From the positivity of the integrand in this expression, it is easy to deduce that $\varphi(x;\mu)<0$ for $x>a(\mu)$, which in view of \eqref{def:phi}, \eqref{def:glog} implies the Euler-Lagrange inequality in \eqref{eq:EL}. We can finally conclude that $\nu$ is indeed the equilibrium measure we were looking for.
Note that we can also write $f$ as a contour integral,
\be\label{f_def}
f(z;\mu)=
\frac{\beta}{4\pi\i}\int\limits_{\Gamma}\left(\frac{t-\mu}{t-a(\mu)}\right)^{1/2}\frac{t^{\beta/2-1}dt}{t-z},
\quad z\in\mathbb{C}\setminus(-\infty,0],
 \ee
where the square root is analytic for $t\in\mathbb C\setminus[\mu,a(\mu)]$ and tends to $1$ as $t\to\infty$, and where $\Gamma$ is a counter-clock-wise oriented loop around $z$ and around $[\mu,a(\mu)]$, but not intersecting $(-\infty,0]$.
This representation will be useful later in Section \ref{sec:integration}.

\begin{remark}
The double integral representation \eqref{varphi} is sufficient to establish the relevant properties of the function $\varphi$ (given in the rest of this Section), however, for numerical computations it is more convenient to operate with a single integral representation of $\varphi$ given in Appendix \ref{AppSIR}. 
\end{remark}
\begin{remark}
Though in what follows we will not need an expression for the constant $\ell_{\mu}$ from \eqref{eq:EL}, we give it for the sake of completeness: it equals
\[
\ell_\mu = \frac{-1}{2\pi}\int_{\mu}^{a_\mu}\frac{x^{\beta/2}\d x}{\sqrt{(a(\mu)-x)(x-\mu)}}
+
\ln\frac{a(\mu)-\mu}{4}.
\]
This can be seen from the large $z$ asymptotics of $g(z),$ the relation \eqref{def:phi} between the functions $g$ and $\varphi,$ and the representation \eqref{varphi_single}: the formula follows by taking the limit $z\to\infty$ in \eqref{varphi_single}.
\end{remark}

\subsection{Asymptotic estimates: small $\mu.$}

\begin{lemma}
Let $\mu_0>0$ be an arbitrary but fixed number.
For every $\beta>0$ there exists $c=c(\beta,\mu_0)>0$ such that for all $\mu\in(0,\mu_0]$
and for all $x\in[\mu,a(\mu)]$ the following holds:
\[
f(x;\mu)>c\begin{cases}
x^{\frac{\beta-1}{2}}, & \beta<1,
\\
1+|\log x|, & \beta=1,
\\
1,&\beta>1.
\end{cases}
\]
\end{lemma}

\begin{proof}
We distinguish between the cases $\beta<1,$ $\beta=1,$ $1<\beta<2$, and $\beta\geq2.$
For $0<\beta<1$ and $\beta\geq 2$, we use a representation of $f$ as an integral over $[\mu, a(\mu)],$ while for $\beta=1$ and $1<\beta<2$, we use a representation of $f$ as an integral over $[0,\infty).$

\medskip
\noindent\textbf{1. Integration over $[\mu, a(\mu)].$} \color{black}
\noindent Let us rewrite \eqref{freg} as
\begin{equation}\label{f_2}
f(z;\mu)=\frac{\beta}{2\pi}\int\limits_{\mu}^{a(\mu)}\frac{t^{\frac{\beta}{2}}-z^{\frac{\beta}{2}}\ -\ \mu(t^{\frac{\beta}{2}-1}-z^{\frac{\beta}{2}-1})}{(t-z)\sqrt{(a(\mu)-t)(t-\mu)}}dt,
\quad z\in\mathbb{C}\setminus(-\infty,0], 
\end{equation}
and substitute the identity
\[
t^{\frac{\beta}{2}-1}-z^{\frac{\beta}{2}-1} = \frac{1}{z}(t^{\frac{\beta}{2}}-z^{\frac{\beta}{2}}) \ - \ \frac{1}{z}(t-z)t^{\frac{\beta}{2}-1},
\]
to obtain
\[
f(z;\mu)=\(1-\frac{\mu}{z}\)\frac{\beta}{2\pi}\int\limits_{\mu}^{a(\mu)}\frac{(t^{\frac{\beta}{2}}-z^{\frac{\beta}{2}})\  d t}{(t-z)\sqrt{(a(\mu)-t)(t-\mu)}}
\ + \ 
\frac{\mu}{z}\frac{\beta}{2\pi}\int\limits_{\mu}^{a(\mu)}\frac{t^{\frac{\beta}{2}-1}\ d t}{\sqrt{(a(\mu)-t)(t-\mu)}},
\quad z\in\mathbb{C}\setminus(-\infty,0]. 
\]
Now we take $z=x\in(\mu,a(\mu))$, and we note that the first integral is a nonsingular one.

\medskip\noindent
\textbf{1.1. $0<\beta<1. $ } \color{black}
For $x\in[\mu, a(\mu)]$, we split the interval of integration in the first integral of the last expression into two parts: $[\mu, x]$ and $[x, a(\mu)],$ and then use the estimate $\frac{t^{\beta/2}-x^{\beta/2}}{t-x}\in \(\frac{\beta}{2}t^{\beta/2-1}, \frac{\beta}{2}x^{\beta/2-1}\)$ (note that the order of the endpoints depends on whether $\beta$ is bigger or smaller than $2$). 
We obtain
\[
\begin{split}
f(x;\mu)\geq & 
\frac{\mu}{x}\frac{\beta}{2\pi}\int\limits_{\mu}^{a(\mu)}\frac{t^{\frac{\beta}{2}-1}\ d t}{\sqrt{(a(\mu)-t)(t-\mu)}}
\ + \
\(1-\frac{\mu}{x}\)x^{\frac{\beta}{2}-1}\frac{\beta}{2\pi}\frac{\beta}{2}\int\limits_{\mu}^{x}\frac{ d t}{\sqrt{(a(\mu)-t)(t-\mu)}}
\\ & +  
\(1-\frac{\mu}{x}\)\frac{\beta}{2\pi}\ \frac{\beta}{2}\int\limits_{x}^{a(\mu)}\frac{t^{\frac{\beta}{2}-1}\  d t}{\sqrt{(a(\mu)-t)(t-\mu)}}=:I_1+I_2+I_3, 
\end{split}
\]
{where we already note that $I_1, I_2, I_3\geq 0$.
If $\mu\leq x \leq \mu(1+\delta)$ for some $\delta>0,$ we substitute $t=\mu+(a(\mu)-\mu)\mu \tau$ in the integral from $I_1,$ and we find that
\begin{multline*}f(x;\mu)\geq I_1=\frac{\mu^{\frac{\beta+1}{2}}}{x}\frac{\beta}{2\pi}\int_0^{1/\mu}\frac{(1+(a(\mu)-\mu) \tau)^{\frac{\beta}{2}-1}d\tau}{\sqrt{\tau(1-\mu\tau)}}\\
\geq \frac{\mu^{\frac{\beta+1}{2}}}{x}\frac{\beta}{2\pi}\int_0^{\frac{1}{\mu_0}}\frac{(1+(a(\mu)-\mu) \tau)^{\frac{\beta}{2}-1}d\tau}{\sqrt{\tau(1-\mu\tau)}}\geq c'\frac{\mu^{\frac{\beta+1}{2}}}{x}\geq cx^{\frac{\beta-1}{2}}\end{multline*}
for sufficiently small $c', c>0$.

If $\mu(1+\delta)\leq x \leq a(\mu)$ for sufficiently small $\delta>0,$ we have
\[f(x;\mu)\geq I_2\geq \(1-\frac{\mu}{x}\)x^{\frac{\beta}{2}-1}\frac{\beta}{2\pi}\frac{\beta}{2\sqrt{a(\mu)}}\int\limits_{\mu}^{x}\frac{ d t}{\sqrt{t-\mu}}\geq c\mu^{\frac{\beta-1}{2}}\geq cx^{\frac{\beta-1}{2}}
\]
for sufficiently small $c>0$.}

\medskip\noindent\textbf{1.2. $\beta\geq 2.$} \color{black}{Set $z=x\in[\mu,a(\mu)]$ and subsitute $t-\mu = (t-x) + (x-\mu)$ in \eqref{freg}. This gives\[
\begin{split}
f(x; \mu)
=
\frac{\beta}{2\pi}\int\limits_{\mu}^{a(\mu)}
\frac{t^{\frac{\beta}{2}-1}}
{\sqrt{(a(\mu)-t)(t-\mu)}}dt
+
\frac{(x-\mu)\beta}{2\pi}\int\limits_{\mu}^{a(\mu)}
\frac{\(t^{\frac{\beta}{2}-1}-x^{\frac{\beta}{2}-1}\)}
{(t-x)\sqrt{(a(\mu)-t)(t-\mu)}}dt.
\end{split}
\]
It is easy to see that this is an increasing function of $x>\mu$ for $\beta\geq 2$, so $f(x;\mu)\geq f(\mu;\mu),$ which proves the lemma since $f(\mu;\mu)$ has a positive limit as $\mu\to0$ and hence is uniformly separated from $0$ for $\mu\in(0,\mu_0].$
}

\medskip
\noindent 
\textbf{2. Integration over $[0,\infty).$} \color{black}
Deforming the contour of integration in  \eqref{f_def} to a large circle and both branches of the negative real line, and then sending the radius of the circle to infinity, we obtain an integral over the negative real half-line; changing the variable $t\mapsto -t$ in the obtained integral, we obtain the following representation:

\begin{equation}\label{f_halfline}
f(z;\mu)=
\frac{\beta}{2\pi} \, \sin\frac{\pi\beta}{2}\int\limits_{0}^{\infty}\sqrt{\frac{t+\mu}{t+a(\mu)}}\, \frac{t^{\frac{\beta}{2}-1}\, d t}{t+z},\quad z\in\mathbb{C}\setminus(-\infty,0],
\end{equation}
which is valid for $0<\beta<2.$ 

\medskip\noindent\textbf{2.1. $\beta=1.$} {We set $z=x\in[\mu,a(\mu)]$, split the integration interval in \eqref{f_halfline} into $[0,1]$ and $[1,\infty),$ and then introduce convenient cross-terms, to obtain
\begin{multline*}
f(x;\mu)
=
\frac{1}{2\pi\sqrt{a(\mu)}} \int\limits_{0}^{1}\sqrt{\frac{t+\mu}{t}}\, \frac{d t}{t+x}
-
\frac{1}{2\pi\sqrt{a(\mu)}} \int\limits_{0}^{1}\sqrt{\frac{t+\mu}{t+a(\mu)}}\, \frac{\sqrt{t\,} \, d t}{\(\sqrt{t+a(\mu)}+\sqrt{a(\mu)}\)(t+x)}
\\
 +
\frac{1}{2\pi} \int\limits_{1}^{\infty}\sqrt{\frac{t+\mu}{t(t+a(\mu))}}\, \frac{d t}{t+x},
\qquad 
x\in[\mu, a(\mu)].
\end{multline*}
The last two terms here are uniformly bounded for all $\mu\in(0,\mu_0]$ and $x\in[\mu, a(\mu)],$ while the integral in the first term evaluates explicitly to
\begin{multline*}
\int\limits_0^1\sqrt{\frac{t+\mu}{t}}\frac{d t}{t+x}
=
\(1-\sqrt{1-\frac{\mu}{x}}\)\log\frac{1}{\mu}
+
\sqrt{1-\frac{\mu}{x}} \log\frac{1+x}{x}
\\
-2\sqrt{1-\frac{\mu}{x}}\log\(\sqrt{1+\mu}+\sqrt{1-\frac{\mu}{x}}\)
+2\log\(\sqrt{1+\mu}+1\).
\end{multline*}
The second line here is uniformly bounded, while the first line is estimated from below by $C |\log x|$ for some $C$
(indeed, for $\mu\leq x \leq\mu(1+\delta)$ with sufficiently small $\delta>0$ the estimate follows from the first term, while for $\mu(1+\delta)\leq x\leq a(\mu)$ it follows from the second term of the first line). Finally, the estimate from below by a constant easily follows from \eqref{f_halfline}.
}

\medskip\noindent
\textbf{2.2. $1<\beta<2.$} {
The representation \eqref{f_halfline} implies that $f(x;\mu)$ decreases monotonically  as a function of $x,$
and hence $f(x;\mu)>f(a(\mu);\mu),$ and the latter is uniformly separated from $0$ for $\mu\in(0,\mu_0].$
}
\medskip\noindent

\end{proof}
\color{black}

\begin{lemma}\label{lemma:f}
Let $\mu_0>0$ be an arbitrary but fixed number.
There exists $\delta>0$ such that we have uniformly for $\mu\in(0,\mu_0]$, $\mu(1+\delta)<x<a(\mu)-\delta$, $0<y<\delta x$,
\[f(x\pm iy;\mu)-f(x;\mu)=
\begin{cases}\mathcal O(x^{\frac{\beta-3}{2}}y),&\beta<1,\\
\mathcal O\((1+|\log x|)x^{-1}y\),&\beta=1,\\ 
\mathcal O(x^{-1}y),&\beta>1.\end{cases}
\]
\end{lemma}
\begin{proof}
Let us focus on $z=x+iy$ with $y>0$, as the case $y<0$ is similar.
By \eqref{f_def}, we have
\[f(x+iy;\mu)-f(x;\mu)=\frac{\beta y}{4\pi}\int_{\Gamma}\left(\frac{t-\mu}{t-a(\mu)}\right)^{1/2}t^{\frac{\beta}{2}-1}\frac{dt}{(t-x)(t-x-iy)}.\]

Let us now choose $\Gamma$ as the circle of radius $r=a(\mu)-\mu/2$ around $a(\mu)$.
Then, for $t\in\Gamma$, if $\mu(1+\delta)<x<a-\delta$, $0<y<\delta u$, we have
\[|t-x-iy|\geq |t-x|-|y|\geq x-\frac{\mu}{2}-\delta x\geq (1-\delta)x-\frac{\mu}{2}\geq x\left(1-\delta -\frac{1}{2+2\delta}\right),\]
hence for $\delta>0$ sufficiently small
\[\frac{1}{|t-a(\mu)|^{1/2}|t-x-iy|}=\mathcal O(x^{-1}),\]
and we obtain
\[f(x+iy;\mu)-f(x;\mu)=\mathcal O\left(x^{-1}y\int_{\Gamma}\frac{|t-\mu|^{1/2}|t|^{\frac{\beta}{2}}}{|t| |t-x|} |d t|\right).\]
Parametrizing $t=a(\mu)-re^{i\theta}$ with $\theta\in[-\pi,\pi]$, note that the contribution of $|\theta|>\delta$ is $\mathcal O(1)$, and focus on the part $|\theta|\leq \delta$ for sufficiently small $\delta$. We can then estimate $|t|$ by the sum of absolute values of its real and imaginary part, using a Taylor expansion for small $\theta$,
\[|t|\leq  a(\mu)-(a(\mu)-\frac{\mu}{2})\cos\theta+(a(\mu)-\frac{\mu}{2})|\sin\theta| \leq \frac{\mu}{2}+C|\theta|\leq C(x+|\theta|),\]
for sufficiently large $C$,
and similarly \[|t-\mu|\leq \left|a(\mu)-\mu-(a(\mu)-\frac{\mu}{2})\cos\theta\right|+(a(\mu)-\frac{\mu}{2})|\sin\theta| \leq \frac{\mu}{2}+C|\theta|\leq C(x+|\theta|).\]
Moreover,
\[|t|=\sqrt{(a(\mu)-(a(\mu)-\frac{\mu}{2})\cos\theta)^2+(a(\mu)-\frac{\mu}{2})^2\sin^2\theta}\geq c(\mu+|\theta|),\]
and
\[|t-x|=\sqrt{(a(\mu)-x-(a-\frac{\mu}{2})\cos\theta)^2+(a(\mu)-\frac{\mu}{2})^2\sin^2\theta}\geq c(\mu+|\theta|),\]
for sufficiently small $c>0$.
Hence, $|t-\mu|^{1/2}/|t|^{1/2}$ is bounded, and
\[f(x+iy;\mu)-f(x;\mu)=\mathcal O(x^{-1}y)+\mathcal O\left(x^{-1}y\int_{-\delta}^\delta (x+|\theta|)^{\frac{\beta-2}{2}}(\mu+|\theta|)^{-1/2}d\theta\right).\]
Making the substitution $|\theta|=-x +(x-\mu)s$
in the latter integral, we obtain that the above is \[\mathcal O(x^{-1}y)+\mathcal O(x^{\frac{\beta-3}{2}}y)\] for $\beta\neq 1$, and 
\[\mathcal O(x^{-1}y \(|\log x|+1)\),\]
for $\beta=1$.
\end{proof}
Next, we need information about the sign of the real part of $\varphi$. See Figure \ref{Fig_varphi} for an illustration.

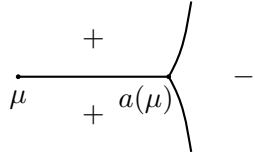
\begin{figure}[ht!]
\center
\begin{tikzpicture}
\draw[thick] (1,0) to (3,0);
\filldraw (1,0) circle (0.7pt);
\filldraw (3,0) circle (0.7pt);
\draw[thick] (3,0) [out = 60, in = -100] to (3.3,1);
\draw[thick] (3,0) [out = -60, in = 100] to (3.3, - 1);
\node at (4,0){$-$};
\node at (2,0.5){$+$};
\node at (2,-0.5){$+$};
\node at (1,-0.3){$\mu$};
\node at (2.7,-0.3){$a(\mu)$};
\end{tikzpicture}
\caption{Distribution of signs of $\Re \varphi(z)$ in a neighborhood of the segment $[\mu, a(\mu)].$}
\label{Fig_varphi}
\end{figure}

\begin{lemma}\label{lemma:estimatephi}
Let $\mu_0>0$ be an arbitrary but fixed number.
There exist $c, \delta>0$ such that for $\mu\in (0,\mu_0]$, $\mu(1+\delta)<x<a(\mu)-\delta$, $0<y<\delta x$,
\[\Re\varphi(x+iy)>c\sqrt{y}f(x;\mu).\]
\end{lemma}
\begin{proof}
By \eqref{varphi}, we have
\[\varphi(x+iy;\mu)-\varphi(x;\mu)=-\frac{1}{2}\int_x^{x+iy}f(\xi;\mu)\left(\frac{\xi-a(\mu)}{\xi-\mu}\right)^{1/2}d\xi.\]
It follows that
\begin{multline*}
\varphi(x+iy;\mu)-\varphi(x;\mu)=
-\frac{if(x;\mu)}{2}\sqrt{{a(\mu)-x}}\int_x^{x+iy}(\xi-\mu)^{-1/2}d\xi
\\
\qquad-\frac{f(x;\mu)}{2}\int_x^{x+iy}\left(\left(\frac{\xi-a(\mu)}{\xi-\mu}\right)^{1/2}-\left(\frac{\mu-a(\mu)}{\xi-\mu}\right)^{1/2}\right)d\xi
\\
-\frac{1}{2}\int_x^{x+iy}\left(f(\xi;\mu)-f(x;\mu)\right)\left(\frac{\xi-a(\mu)}{\xi-\mu}\right)^{1/2}d\xi.
\end{multline*}
\color{black}
The first term is
\[-i\left.\sqrt{a(\mu)-x}f(x;\mu)(\xi-\mu)^{1/2}\right|_{\xi=x}^{x+iy},\]
and has real part
\begin{multline*}\sqrt{a(\mu)-x}f(x;\mu)\left({(x-\mu)^2+y^{2}}\right)^{1/4}\cos\left(\frac{1}{2}\arctan\frac{y}{x-\mu}\right)
\\
\geq \frac{1}{2}\sqrt{a(\mu)-x}f(x;\mu)(x-\mu)^{1/2}\geq cf(x;\mu)u^{1/2}.
\end{multline*}
It remains to prove that the two other terms are smaller than, say, one third of the above expression.
The second term can be estimated as
\[\mathcal O\left(f(x;\mu)\int_x^{x+iy}(\xi-\mu)^{1/2}d\xi\right)=\mathcal O\left(f(x;\mu)|x-\mu+iy|^{1/2} y\right)=\mathcal O\left(f(x;\mu)x^{1/2} y\right),\]
and is indeed much smaller than the first one.
For the third term, we use the previous lemma and obtain for $\beta<1$,
\[\mathcal O\left(x^{\frac{\beta-3}{2}}\int_x^{x+iy}(\xi-x)\left(\frac{\xi-a(\mu)}{\xi-\mu}\right)^{1/2} d\xi\right)=\mathcal O\left(x^{\frac{\beta-3}{2}}x^{-1/2}y^2\right)=\mathcal O\left(f(x;\mu)x^{-3/2}y^2\right).\]
By choosing $\delta>0$ small enough, we can make this also smaller than one third of the first term. The estimates for $\beta\geq 1$ follow similarly.
\end{proof}

\subsection{Local changes of variables}

Next, we need estimates for the derivative of $f(z;\mu)$ near $z=\mu$ and $z=a(\mu)$.
\begin{proposition}
\label{prop_fderf}
For any $\mu_0>0$, there exists $r_0>0$ such that the following holds. 
\begin{enumerate}
\item 
Uniformly for $\mu\in(0,\mu_0]$ and for $|z-\mu|<r_0\mu$, we have
\begin{equation}\label{estderfmu}f(z;\mu)=\mathcal O\left(\rho(\mu)\right),\qquad
\partial_zf(z;\mu)=
\begin{cases}\mathcal O(1),&\beta>3,
\\
\mathcal O\left(\ln\frac{1}{\mu}\right),& \beta=3,
\\
\mathcal O\left(\mu^{\frac{\beta-3}{2}}\right),& \beta<3.
\end{cases} 
\end{equation}
\item
Uniformly for $\mu\in(0,\mu_0]$ and for $|z-a(\mu)|<r_0$, we have
\begin{equation}\label{estderfa}
\partial_zf(z;\mu)=\mathcal O(1).
\end{equation}
\end{enumerate}
\end{proposition}
\begin{proof}
For the proof of \eqref{estderfmu}, we proceed as in the proof of Lemma \ref{lemma:f}. Let us focus on the estimate for $\partial_zf(z;\mu)$, as the one for $f(z;\mu)$ is similar.
In \eqref{f_def}, we let $\Gamma$ be the circle of radius $r=a(\mu)-\mu/2$ around $a(\mu)$, and we differentiate to obtain the estimate
\[\partial_zf(z;\mu)=\mathcal O\left(\int_{\Gamma}\left|\frac{t-\mu}{t-a(\mu)}\right|^{1/2}\frac{|t|^{\beta/2}|dt|}{|t|.|t-z|^2}\right).
\]
Parametrizing $t=a(\mu)-re^{i\theta}$ with $\theta\in[-\pi,\pi]$, note that the contribution of $|\theta|>\delta$ is $\mathcal O(1)$, and focus on the part $|\theta|\leq \delta$ for sufficiently small $\delta$. 
Like in the proof of Lemma \ref{lemma:f}, we have
\[|t|=\mathcal O(\mu+|\theta|),
\quad |t-\mu|=\mathcal O(\mu +|\theta|),\quad |t|^{-1}=\mathcal O\left((\mu+|\theta|)^{-1}\right),\quad
|t-a(\mu)|^{-1}=\mathcal O(1).\]
For $|z-\mu|<r_0\mu$ with $r_0>0$ sufficiently small, we have moreover
\[|t-z|\geq |t-\mu|-|\mu-z|\geq |t-\mu|-r_0\mu \geq c(\mu+|\theta|),\]
for $c$ sufficiently small.
Hence,
\[\partial_zf(z;\mu)=\mathcal O\left(\int_{-\delta}^{\delta}
(\mu+|\theta|)^{\frac{\beta-5}{2}}d\theta
\right)+\mathcal O(1)=\mathcal O\left(\mu^{\frac{\beta-3}{2}}\right)+\mathcal O(1),\]
for $\beta\neq 3$. For $\beta=3$, we get
$\mathcal O\left(\log\frac{1}{\mu}\right)$. This proves \eqref{estderfmu}.
The proof of \eqref{estderfa}
is similar but simpler, and we omit it.
\end{proof}

Expanding formula \eqref{varphi} in the vicinities of the points $\mu, a(\mu),$ we obtain
\be\label{varphi_asymp_mu}\varphi(z;\mu)=\pm\pi i+\rho(\mu)\sqrt{a(\mu)-\mu}(\mu-z)^{1/2}\left(1+\mathcal{O}(\mu-z)\right),
\qquad \mbox{ as } z\to\mu\quad  \mbox{ with }\quad \pm\Im z>0,\ee
and 
\be\label{varphi_asymp_a}\varphi(z;\mu)=\frac{-f(a(\mu);\mu)}{3\sqrt{a(\mu)-\mu}}(z-a(\mu))^{3/2}\left(1+\mathcal O
(z-a(\mu))\right)
\qquad \mbox{ as }\quad z\to a(\mu),\ee
where, as usual, $\sqrt{.}$ denotes the {positive square root} and $(.)^{1/2},$ $(.)^{3/2}$ denote the principal branches of the roots, and where we used \eqref{f_rho}.

Formulas \eqref{varphi_asymp_mu}, \eqref{varphi_asymp_a} suggest to introduce local conformal changes of the variable $z$ in small neighborhoods of the points $\mu,$ $a(\mu),$ respectively, as follows:
\begin{equation}\label{phimuzeta}
\varphi(z;\mu)\mp\pi\i =:(-\zeta(z;\mu))^{1/2},\quad z\to \mu,
\qquad
\varphi(z;\mu) =:-\frac23\eta(z;\mu)^{3/2}, \quad z\to a(\mu),
\end{equation}
where the roots in $\eta(z;\mu)^{3/2},$ $(-\zeta(z;\mu))^{1/2}$ correspond to $\arg(.)\in(-\pi,\pi).$

In the following lemma, we gather some properties of $\zeta(z;\mu)$ and $\eta(z;\mu)$ that we will need later.

\begin{lemma}\label{lemma:eq}
For any $\mu_0>0$, there exists $r_0>0$ such that the following holds.
\begin{enumerate}[label = \arabic*.]
\item For any $\mu\in(0, \mu_0],$ the function $\zeta(z;\mu)$ is conformal in the disk $|z-\mu|<r_0\mu$, and
\begin{equation}\label{derzeta}\zeta(\mu; \mu)=0,\qquad 
\partial_z\zeta(z; \mu)_{z=\mu}= \rho(\mu)^2(a(\mu)-\mu)>0.
\end{equation}
Moreover, for any $0<r<r_0$, there exists $C>0$ such that we have
\begin{equation}\label{estzeta}
\frac{1}{C}<\frac{|\zeta(z;\mu)|}{\mu\, \rho(\mu)^2}<C\quad \mbox{ for all $\mu\in(0,\mu_0)$ and $|z-\mu|=\mu r$}.
\end{equation}
\item For any $\mu\in(0,\mu_0]$, the function $\eta(z;\mu)$ is conformal for $|z-a(\mu)|<r_0$. There exists $c>0$ (independent of $\mu$) such that
\begin{equation}\label{dereta}
\eta(a(\mu);\mu)=0,\quad 
\partial_z\eta(z;\mu)_{z=a(\mu)}
>c
>0,\qquad\mbox{for any $\mu\in(0,\mu_0]$.}
\end{equation}
Moreover, for any $0<r<r_0$, there exists $C>0$ such that we have
\begin{equation}\label{esteta}
\frac{1}{C}<|\eta(z;\mu)|<C\quad \mbox{ for all $\mu\in(0,\mu_0)$ and $|z-a(\mu)|=r$.}
\end{equation}
\end{enumerate}
\end{lemma}

\begin{proof}
With $\zeta$ and $\eta$ defined by \eqref{phimuzeta}, \eqref{derzeta} and \eqref{dereta} follow easily from \eqref{varphi_asymp_mu} and \eqref{varphi_asymp_a}.
Next, we write \eqref{varphi} (using \eqref{f_rho}) as
\begin{multline*}\varphi(z;\mu)\mp \pi i=-\frac{\rho(\mu)\sqrt{a(\mu)-\mu}}{2}\int_{\mu}^z\left(\frac{1}{\mu-\xi}\right)^{1/2}d\xi \\-\frac{1}{2}\int_{\mu}^z\left(\frac{1}{\mu-\xi}\right)^{1/2}\left((a(\mu)-\xi)^{1/2}f(\xi;\mu)-\sqrt{a(\mu)-\mu}f(\mu;\mu)\right)d\xi,\ \pm \Im z>0.
\end{multline*}
For $|\xi-\mu|<r_0\mu$ with $r_0>0$ small enough, we have
\begin{multline*}
\left|(a(\mu)-\xi)^{1/2}f(\xi;\mu)-\sqrt{a(\mu)-\mu}f(\mu;\mu)\right|\leq
f(\xi;\mu)\left|(a(\mu)-\xi)^{1/2}-\sqrt{a(\mu)-\mu}\right|\\
+\sqrt{a(\mu)-\mu}\left|f(\xi;\mu)-f(\mu;\mu)\right|
\end{multline*}
By \eqref{estderfmu}, the first term at the right is $\mathcal O\left(\rho(\mu)|\xi-\mu|\right)$,
 and the second is, also by \eqref{fcoeff}, 
$\mathcal O\left(\frac{\rho(\mu)|\xi-\mu|}{\mu}\right)$ for $\beta\neq 3$, and $\mathcal O\left(|\xi-\mu|\log\frac{1}{\mu}\right)$ for $\beta=3$.
Hence, it follows that
\[\varphi(z;\mu)\mp \pi i={\rho(\mu)\sqrt{a(\mu)-\mu}}(\mu-z)^{1/2}\left(1 +\mathcal O\left(\mu^{-1}|{\mu-z}|\right)\right),\qquad \pm \Im z>0,
\]
for $\beta\neq 3$, and similarly but with error term
$\mathcal O\left(\log\frac{1}{\mu}|{\mu-\xi}|\right)$ for $\beta=3$.
It then follows from \eqref{phimuzeta}
that
\[\zeta(z;\mu)={\rho(\mu)^2(a(\mu)-\mu)}(z-\mu)\left(1 +\mathcal O\left(\mu^{-1}|{\mu-z}|\right)\right)\]
for $|z-\mu|<r_0\mu$ with $r_0$ sufficiently small.
It follows that $\zeta(z;\mu)$ is conformal in $|z-\mu|<r_0\mu$ with sufficiently small $r_0>0$, and the estimate \eqref{estzeta} also follows easily.
The analogous results for $\eta$ are proved similarly but simpler, and we omit the details.
\end{proof}

\section{Standard asymptotic analysis of the RH problem: large $n$ asymptotics for {$\frac{M}{\rho(\mu)n}\leq \mu\leq \mu_0.$} }
\label{section:6RH1}

{In this section, we will obtain asymptotics for $Y$ defined in \eqref{def:Y} as $n\to\infty$, uniformly for $\frac{M}{\rho(\mu)n}\leq \mu\leq \mu_0.$, where $\mu_0>0$ is an arbitrary constant and $M>0$ is a sufficiently large constant.}
This asymptotic analysis is rather standard, and was done for instance in \cite {CG19} for $\mu>0$ fixed. The analysis in this section follows the same lines as \cite{CG19}, with minor adaptations in order to be able to cover also the case where $\mu$ depends on $n$ and tends to $0$ slow enough as $n\to\infty$. We repeat the analysis mainly for the sake of completeness, and to be able to write all relevant quantities explicitly using coherent notations.
There is only one unusual feature worth noticing in the RH analysis in this section: we do not normalize the global parametrix to the identity, but rather to a fixed matrix $\sigma_P$. This leads, compared to \cite{CG19}, to a modified definition of $P^{(\infty)}$ in \eqref{def:Pinfty} and of $R$ in \eqref{def:R}. In this section, those changes are convenient but not important; however, in our refinements of the RH analysis in the next sections, the modified normalization of $P^{(\infty)}$ will have important implications.

\subsection{RH problem for $Y$}
It is well-known \cite{FokasItsKitaev} that the function $Y$ defined by \eqref{def:Y} is the unique function satisfying the following RH conditions.

\subsubsection*{RH problem for $Y$}\label{RHP_Y}
\begin{enumerate}[label=(\alph*)]
\item $Y:\mathbb C\setminus[\mu,+\infty)\to \mathbb C^{2\times 2}$ is analytic,
\item $Y$ has continuous boundary values $Y_\pm(x)$ when $x\in(\mu,+\infty)$ is approached from above and below, and they are related by
\begin{equation}
\nonumber
Y_+(x)=Y_-(x)\begin{pmatrix}1&x^\alpha e^{-nx^{{\beta}/{2}}}\\0&1\end{pmatrix},\qquad x\in(\mu,+\infty),
\end{equation}
\item $Y$ has the asymptotics 
\begin{equation}
\nonumber
Y(z)=\left(I+\mathcal O(z^{-1})\right)z^{n\sigma_3}\qquad\mbox{as $z\to\infty$, with $\sigma_3=\begin{pmatrix}1&0\\0&-1\end{pmatrix}$},\end{equation}
\item as $z\to \mu$, $Y(z)=\mathcal O(\log (z-\mu))$, by which we mean that each entry of the matrix is of this order.
\end{enumerate}
We will now analyze $Y$ asymptotically as $n\to\infty$ together with $\mu\to 0$ by applying the Deift-Zhou nonlinear steepest descent method \cite{Deift, DKMVZ2, DKMVZ1}. This method consists of a series of explicit and invertible transformations $Y\mapsto T\mapsto S\mapsto R$, where we will in the end arrive at a RH problem for $R$ which we will be able to solve approximately. Inverting the transformations will then allow us to recover the asymptotics for $Y$.

\subsection{Normalization at infinity and on the interval $[\mu, a(\mu)]$: $Y\mapsto T$}
The first feature of the RH problem for $Y$ which we would like to improve is its asymptotics as $z\to\infty.$ Our first transformation will ensure that the RH solution {\color{black}tends to $I$ as $z\to\infty$.}
For this, we use the $g$-function constructed in Section \ref{sec:gfunction}, and the associated quantity $\ell_\mu$.
 Namely, 
we define
\begin{equation}
\label{def:T}
T(z)=\delta(\infty)^{-\sigma_3}e^{-n\ell_\mu\sigma_3}Y(z)e^{-ng(z;\mu)\sigma_3}e^{n\ell_\mu\sigma_3}\delta(z)^{\sigma_3}.
\end{equation}
where the function $\delta(z)$ is defined in the following way:
\begin{equation}\label{delta}
\delta(z)=\left(\frac{\sqrt{a(\mu)}(z-\mu)^{1/2}+\sqrt{\mu}(z-a(\mu))^{1/2}}{(z-\mu)^{1/2}+(z-a(\mu))^{1/2}}\right)^\alpha.
\end{equation}
The function $\delta$ satisfies the jump and normalisation properties
\begin{equation}\label{deltainf}\delta_+(x)\delta_-(x)=x^\alpha \qquad\mbox{for }\mu<x<a(\mu),\qquad \delta(\infty)=\lim_{z\to\infty}\delta(z)=\left(\frac{\sqrt{a(\mu)}+\sqrt{\mu}}{2}\right)^\alpha,
\end{equation}
and the function $\delta$ is analytic for $z\in\mathbb C\setminus[\mu,a(\mu)]$.
For further usage we denote
\begin{equation}\label{deltahat}
\widehat\delta(z) = \delta(z)z^{-\alpha/2},
\end{equation}
where the principal branch of $z^\alpha$ is chosen.
Observe that $\widehat\delta(\mu) = \widehat\delta(a(\mu)) = 1,$
and that $\widehat\delta(z)$ is uniformly bounded and bounded away from $0$ for $z$ in the disks $|z-\mu|\leq \frac{\min\(\mu, a(\mu)-\mu\)}{2},$ $|z-a(\mu)|\leq \frac{a(\mu)-\mu}{2},$ { uniformly for all $\mu\in(0, +\infty).$}

Later on, we will need expansions for $\widehat\delta$ and $\widehat\delta^{-1}$ for $z$ near $\mu$. For this, we
introduce the local change of variable
\begin{equation}\label{eq:changevarw}
w=\frac{\mu-z}{a(\mu)-\mu},
\qquad
z=\mu - (a(\mu)-\mu)w,
\quad 
z-a(\mu) = -(a(\mu)-\mu)(1+w),
\end{equation}
such that 
\[
\widehat\delta(z) =
\(
\frac{\sqrt{\mu}\sqrt{1+w}+\sqrt{a(\mu)}\sqrt{w}}{\sqrt{(\mu-(a(\mu)-\mu)w)}\(\sqrt{1+w}+\sqrt{w}\)}
\)^{\alpha}.\]
We need to expand $\widehat\delta$ 
as $w\to 0$ up to terms of order $o(w)$. The direct computation of the $\mathcal{O}(w)$ term in this expansion is rather involved, but we can circumvent that computation by using the relation $\widehat\delta_+(x)\widehat\delta_-(x)=1$ for $\mu<x<a(\mu)$.
Writing 
\[
\begin{split}
\widehat\delta(z) = 1+c_1\sqrt{w}+c_2w+\mathcal{O}(w^{3/2}),
\end{split}\]
and multiplying these two expressions, we find 
$c_2=\frac12c_1^2.$
The coefficient $c_1$ can be computed rather easily, and we obtain
\begin{equation}\label{delta_s}
\begin{split}
\widehat\delta(z) = 1 + \alpha\(\frac{\sqrt{a(\mu)}}
{\sqrt{\mu}}-1\)\sqrt{w}+\frac12\alpha^2\(\frac{\sqrt{a(\mu)}}
{\sqrt{\mu}}-1\)^2w+\mathcal{O}(w^{3/2}),
\\
\widehat\delta(z)^{-1} = 1 - \alpha\(\frac{\sqrt{a(\mu)}}
{\sqrt{\mu}}-1\)\sqrt{w}+\frac12\alpha^2\(\frac{\sqrt{a(\mu)}}
{\sqrt{\mu}}-1\)^2w+\mathcal{O}(w^{3/2}).
\end{split}
\end{equation}

Using the function $\varphi$ from \eqref{def:phi}, we can show that $T$ solves the following RH problem.
\subsubsection*{RH problem for $T$}
\begin{enumerate}[label=(\alph*)]
\item $T:\mathbb C\setminus[\mu,+\infty)\to \mathbb C^{2\times 2}$ is analytic,
\item $T$ satisfies the jump relations
\[
\begin{split}
&T_+(x)=T_-(x)\begin{pmatrix}\dfrac{\delta_+(x;\mu)}{\delta_-(x;\mu)}e^{-n(\varphi_+(x;\mu)-\varphi_-(x;\mu))}&\overbrace{\frac{x^\alpha
\e^{n(\varphi_+(x;\mu)+\varphi_-(x;\mu))}}{\delta_+(x;\mu)\delta_-(x;\mu)}}^{=1}\\\\0&\dfrac{\delta_-(x;\mu)}{\delta_+(x;\mu)}e^{n(\varphi_+(x;\mu)-\varphi_-(x;\mu))}\end{pmatrix}, & x\in(\mu,a(\mu)),
\\
&T_+(x)=T_-(x)\begin{pmatrix}1&{\color{black}\dfrac{e^{2n \varphi(x;\mu)}}{\widehat\delta(x;\mu)^2}} \\ 0 & 1 \end{pmatrix}, & x\in(a(\mu),+\infty),
\end{split}
\]
\item $T$ has the asymptotics 
\begin{equation}\nonumber T(z)=I+\mathcal O(z^{-1})\qquad\mbox{as $z\to\infty$}, \end{equation}
\item as $z\to \mu$, $T(z)=\mathcal O(\log (z-\mu))$.
\end{enumerate}
Note that, by \eqref{phippphim}, the $(1,2)$ entry in the jump matrix equals $1$ on $(\mu,a(\mu))$.

\subsection{Opening of lenses: $T\mapsto S$}
Whereas the asymptotics at infinity of $T$ are more convenient than the ones for $Y$, the jump matrices are at first sight more complicated because they are oscillatory on the interval $[\mu,a(\mu)].$ The idea now is to transform the jumps in such a way that the oscillating exponentials will be moved from the contour $[\mu,a(\mu)]$ to contours in the complex plane on which they are rapidly decaying. 

To this end, we use the factorization
\[\begin{pmatrix}\frac{\delta_+(x;\mu)}{\delta_-(x;\mu)}e^{-2n\varphi_+(x;\mu)}&1\\0&
\frac{\delta_-(x;\mu)}{\delta_+(x;\mu)}e^{-2n\varphi_-(x;\mu)}\end{pmatrix}
\hskip-1mm
=
\hskip-1mm
\begin{pmatrix}1&0\\\frac{\delta_-(x;\mu)^2}{x^{\alpha}}
e^{-2n\varphi_-(x;\mu)}&1\end{pmatrix}
\hskip-2mm
\begin{pmatrix}0&1\\-1&0\end{pmatrix}
\hskip-2mm
\begin{pmatrix}1&0\\
\frac{\delta_+(x;\mu)^2}{ x^{\alpha}}e^{-2n\varphi_+(x;\mu)}&1\end{pmatrix}\hskip-1mm,
\]
for $\mu<x<a(\mu)$. Next, we define lenses $\Sigma_\pm$ connecting $\mu$ with $a(\mu)$ in the upper and lower parts of the complex plane, as shown in Figure \ref{fig:lens}.
The precise shape of the lenses will not be important, but they need to satisfy some further conditions. Let us take $\delta>0$ sufficiently small and define disks $U_\mu$ of radius $\delta\mu$
 around $\mu$, and $U_a$ of radius $\delta$ around $a(\mu).$ Then we require that $\Sigma_-=\overline{\Sigma_+}$ and that 
$\Sigma_+\setminus\left(U_\mu\cup U_a\right)$ lies in the region ${\Im z\geq \frac{\delta}{2}\mu}$, i.e. the upper part of the lens outside the disks lies above the line connecting $0$ with $a(\mu)+i\frac{\delta}{2}$.
Let us now define
\begin{equation}\label{def:S}
S(z)=\begin{cases}T(z),&\mbox{ for $z$ outside the lenses,}\\
T(z)\begin{pmatrix}1&0\\-\widehat\delta(z;\mu)^2e^{-2n\varphi(z;\mu)}&1\end{pmatrix},&\mbox{ for $z$ in the upper part of the lens,}\\
T(z)\begin{pmatrix}1&0\\\widehat\delta(z;\mu)^2e^{-2n\varphi(z;\mu)}&1\end{pmatrix},&\mbox{ for $z$ in the lower part of the lens,}\\
\end{cases}
\end{equation}
where we recall the definition \eqref{deltahat} of $\widehat\delta$, and the lens-shaped region is as illustrated in Figure \ref{fig:lens}.

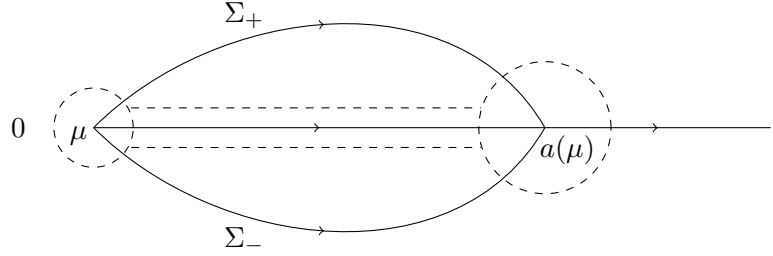
\begin{figure}
\center{\center
\begin{tikzpicture}
[decoration={markings, mark=at position 0.5 with{\arrow{>};}}]

\draw [postaction=decorate](1,0) to [out=45, in =120] (7,0);
\draw [postaction=decorate](1,0) to [out=-45, in =-120] (7,0);
\draw [postaction=decorate](1,0) -- (7,0);
\draw [postaction=decorate](7,0) -- (10,0);

\node at (3,1.5) {$\Sigma_+$};
\node at (3,-1.5) {$\Sigma_-$};

\node at (0,0) {$0$};
\node at (0.8,-0.1) {$\mu$};
\node at (7.3,-0.3) {$a(\mu)$};

\draw[dashed] (1,0) circle(15pt);
\draw[dashed] (7,0) circle(25pt);
\draw[dashed] (1.5 ,7.5pt) -- (6.15,7.5pt);
\draw[dashed] (1.5,-7.5pt) -- (6.15,-7.5pt);

\end{tikzpicture}}
\caption{Contour for the RH problem for the function $S(z).$
{ The dashed lines are the circles $\partial U_{\mu}, \partial U_{a}$ and the lines $|\Im z|=\frac{\delta\mu}{2}.$}}
\label{fig:lens}
\end{figure}

\subsubsection*{RH problem for $S$}\label{sect:RH:S}
\begin{enumerate}[label=(\alph*)]
\item $S:\mathbb C\setminus[\mu,+\infty)\to \mathbb C^{2\times 2}$ is analytic,
\item $S$ satisfies the jump relations
\begin{align}\nonumber
&S_+(z)=S_-(z)\begin{pmatrix}1&0\\\widehat\delta(z;\mu)^2
e^{-2n\varphi(z;\mu)}&1\end{pmatrix},& z\in\Sigma_+\cup\Sigma_-,\\
\nonumber
&S_+(x)=S_-(x)\begin{pmatrix}0&1\\-1&0\end{pmatrix},& \mu<x<a(\mu),\\
&S_+(x)=S_-(x)\begin{pmatrix}1&\widehat\delta(x;\mu)^{-2}e^{2n\varphi(x;\mu)}\\0&1\end{pmatrix},& x\in(a(\mu),+\infty),
\nonumber
\end{align}
\item $S$ has the asymptotics 
\begin{equation}\nonumber
S(z)=I+\mathcal O(z^{-1})\qquad\mbox{as $z\to\infty$}, 
\end{equation}
\item as $z\to \mu$, $S(z)=\mathcal O(\log (z-\mu))$.
\end{enumerate}

\subsection{Construction of the global parametrix $P^{(\infty)}$}
The RH problem for $S$ has constant jump on the interval $(\mu,a(\mu))$, and exponentially small jumps on the other parts of the contour {\color{black}as long as one is not too close to the endpoints $\mu$ and $a(\mu)$. We will show} that the solution $S$ outside small neighbourhoods of the points $\mu, a(\mu)$ is close to the solution of the RH problem obtained formally from the one for $S$ if we let $n\to\infty$, as long as $\mu$ does not go to $0$ too rapidly. 
In other words, we then obtain the following RH problem, which, as already mentioned earlier, we normalize in an unusual way, {which helps us to obtain simpler matching conditions on circles around the points $\mu,a(\mu)$.}

\subsubsection*{RH problem for $P^{(\infty)}$}
\begin{enumerate}[label=(\alph*)]
\item $P^{(\infty)}:\mathbb C\setminus{\color{black}[\mu,a(\mu)]}\to \mathbb C^{2\times 2}$ is analytic,
\item $P^{(\infty)}$ satisfies the jump relation
\begin{equation}\nonumber
P^{(\infty)}_+(x)=P^{(\infty)}_-(x)\begin{pmatrix}0&1\\-1&0\end{pmatrix},\qquad \mu<x<a(\mu),\end{equation}
\item 
$\lim\limits_{z\to\infty}P^{(\infty)} = \sigma_P,$
where
\begin{equation}\label{sigmaPN}
\sP = \frac{1}{\sqrt{2}}\begin{pmatrix}1 & \i \\ \i & 1\end{pmatrix}.
\end{equation}
\end{enumerate}

Note that the asymptotics at $\infty$ of $P^{(\infty)}$ is $\sigma_P$ rather than the identity matrix, and hence we expect $P^{(\infty)}$ to approximate $\sigma_P S$ rather than $S.$
There are many solutions to this RH problem, depending on how one wants to normalize it at the points $\mu$ and $a(\mu)$. 
The solution that will turn out to be a good approximation to $\sigma_P S$ for large $\mu$ is defined as 
\begin{equation}\label{def:Pinfty}
P^{(\infty)}(z):=\begin{pmatrix}\gamma(z)^{-1} & 0 \\ 0 & \gamma(z)\end{pmatrix} 
\sP,
\quad \mbox{ where }
\gamma(z)=\left(\frac{z-a(\mu)}{z-\mu}\right)^{1/4},
\end{equation}
where $\gamma$ is analytic in $\mathbb C\setminus[\mu,a(\mu)]$, and tends to $1$ at infinity.

\subsection{Construction of the local parametrix $P^{(\mu)}$ near $\mu$}
As already mentioned, we expect (and will show) that $P^{(\infty)}$ is close to $\sigma_P S$ outside some small neighborhoods of $\mu, a(\mu)$. However, we can not expect this also inside these neighborhoods, since $P^{(\infty)}$ blows up at $\mu$ and $a(\mu)$, while $\sigma_P S$ remains bounded there. In order to find the behavior of $S$ when $z$ is close to $\mu, a(\mu),$ we will construct functions satisfying the exact jump relations locally (referred to as parametrices) in neighbourhoods of $z=\mu, z=a(\mu).$

We start with the local parametrix at the point $z=\mu.$
Let us denote, as before, $U_\mu$ for the disk of radius 
$\delta\mu$ around $\mu$, for some sufficiently small $\delta>0$ independent of $\mu$. For $\mu\leq \mu_0$ and $\delta>0$ small, $a(\mu)$ is not contained in this disk. The aim of this section is to construct a function satisfying the following conditions.

\subsubsection*{RH problem for $P^{(\mu)}$}
\begin{enumerate}[label=(\alph*)]
\item $P^{(\mu)}$ is analytic in ${\color{black}U_\mu}\setminus \left(\Sigma_+\cup\Sigma_-\cup[\mu,a(\mu)]\right)$,
\item $P^{(\mu)}(z)$ satisfies exactly the jump relations as $S$ for $z\in U_\mu\cap\left(\Sigma_+\cup\Sigma_-\cup[\mu,a(\mu)]\right)$,
\item $P^{(\mu)}$ matches with $P^{(\infty)}$ on the boundary of the disk, in the sense that
\begin{equation}\label{matching}
P^{(\mu)}(z)\left(P^{(\infty)}(z)\right)^{-1}\to I
\end{equation}
uniformly for $z\in\partial U_\mu$, $\frac{M}{\rho(\mu)n}\leq \mu\leq \mu_0$ in the limit where $n\to\infty$.
\end{enumerate}
{\color{black}
The limitation of the present RH analysis lies in this condition (c): we will only be able to achieve \eqref{matching} if  $n\rho(\mu)\mu$ is sufficiently large.
}
If $\mu\to 0$ more rapidly, we will have to refine our RH analysis later on.

Recall the analytic for $z\in U_\mu\setminus [\mu, a(\mu)]$ function $\varphi(z;\mu),$ defined in \eqref{def:phi}, {and recall the function $\zeta(z;\mu)$ defined in \eqref{phimuzeta}, whose properties are described in Lemma \ref{lemma:eq}.}

We will construct $P^{(\mu)}$ in terms of the unique solution to the following Bessel model RH problem, for which we construct the explicit solution below. 

\subsubsection*{RH problem for $\Psi_B$}
\begin{enumerate}[label=(\alph*)]
\item $\Psi_B$ is analytic for $\zeta\in\mathbb{C}\setminus\Sigma_{B},$ where $\Sigma_{B}=[0,+\infty)\cup[0,e^{\i\theta}\infty)\cup[0,e^{-\i\theta}\infty),$ and $0<\theta<\pi,$ with each of the rays oriented away from the origin,
\item $\Psi_B$ satisfies the jumps $\Psi_{B,+}(\zeta) = \Psi_{B,-}(\zeta) J_{\Psi_B}(\zeta),$ where
\[
\begin{split}J_{\Psi_B}(\zeta) =\begin{cases} \begin{pmatrix}1 & 0 \\ 1 & 1\end{pmatrix}, &\zeta\in [0, e^{\i\theta}\infty),
\\
\begin{pmatrix}1 & 0 \\ 1 & 1\end{pmatrix}, &\zeta\in [0, e^{-\i\theta}\infty),
\\
\begin{pmatrix}0 & 1 \\ -1 & 0 \end{pmatrix}, &\zeta\in[0, \infty),
\end{cases}\end{split}\]
\item $\Psi_B$ has the following asymptotics as $\zeta\to\infty,$ uniformly w.r.t. $\arg\zeta\in(-\pi,\pi):$ for any integer $M\geq 1,$
\[
\begin{split}
\Psi_B(\zeta)&=(-\zeta)^{-\sigma_3/4}\sP^{-2}
\, \mathcal{E}_B(\zeta)\,
\sP\cdot
e^{(-\zeta)^{1/2}\sigma_3},
\end{split}\]
\end{enumerate}
where $\sP$ is as in \eqref{sigmaPN} and  $\mathcal{E}_B(\zeta)$ has the following asymptotics as $\zeta\to\infty$, which is uniform w.r.t. $\arg\zeta\in[-\pi,\pi]:$ for any integer $M\geq 1,$
\begin{equation}\label{EB}
\begin{split}
\mathcal{E}_B(\zeta)
&=
\begin{pmatrix}
\sum\limits_{k=0}^{M}\frac{-(4k+1)!!(4k-3)!!}{(2k)!8^{2k}(-\zeta)^k}
&
\i\sum\limits_{k=1}^{M}\frac{(4k-1)!!(4k-5)!!}{(2k-1)!8^{2k-1}(-\zeta)^{k-1/2}}
\\
\i\sum\limits_{k=1}^{M}\frac{(4k-3)!!^2}{(2k-1)!8^{2k-1}(-\zeta)^{k-1/2}}
&
\sum\limits_{k=0}^{M}\frac{(4k-1)!!^2}{(2k)!8^{2k}(-\zeta)^k}
\end{pmatrix}
+\mathcal{O}(\zeta^{-M-1/2}).
\end{split}\end{equation}
In what follows we will need an explicit expression for the sub-leading term of $\mathcal{E}_B,$
\begin{equation}\label{EB2}
\mathcal{E}_B(\zeta)
=
I+
\begin{pmatrix} 0 & \frac{3\i}{8} \\ \frac{\i}{8} & 0
\end{pmatrix}\frac{1}{\sqrt{-\zeta}}+\mathcal{O}(\zeta^{-1}),\quad \zeta\to\infty.
\end{equation}

\textbf{Construction of a parametrix out of modified Bessel functions.}
We refer the reader interested in a more detailed and pedagogical construction of similar Bessel parametrices to \cite{KMVV} (see also \cite{Vanlessen}).
We define  \[
\widehat\Psi(\zeta)=\begin{pmatrix}\sqrt{\pi}I_{0}(\sqrt{-\zeta}) & \frac{-\i}{\sqrt{\pi}}K_0(\sqrt{-\zeta})
\\
-\i\sqrt{\pi} \sqrt{-\zeta} I_{0}'(\sqrt{-\zeta}) & \frac{-\sqrt{-\zeta}}{\sqrt{\pi}}K_0'(\sqrt{-\zeta})\end{pmatrix},
\]
and then we set
\begin{equation}\label{Psi_B}\Psi_B(\zeta)=\begin{cases}\widehat\Psi(\zeta),& \arg\zeta\in(\theta,\pi]\cup[-\pi,-\theta),
\\
\widehat\Psi(\zeta)\begin{pmatrix}1 & 0 \\ -1 & 1\end{pmatrix},& \arg\zeta\in(0, \theta),
\\
\widehat\Psi(\zeta)\begin{pmatrix}1 & 0 \\ 1 & 1\end{pmatrix},& \arg\zeta\in(-\theta,0).
\end{cases}
\end{equation}

\noindent 
Next, we let
\begin{equation}\nonumber
P^{(\mu)}(z)=E^{(\mu)}(z)\Psi_B(n^2\zeta(z; \mu))e^{-n(\varphi(z;\mu)\mp\pi\i)\sigma_3}\widehat\delta(z;\mu)^{\sigma_3},
\quad \pm\Im z>0.
\end{equation}
where
\begin{equation}
\label{def:Es}
\nonumber
\clu{
E^{(\mu)}(z)=P^{(\infty)}(z)\widehat\delta(z;\mu)^{-\sigma_3}\sP(-n^2\zeta(z; \mu))^{\sigma_3/4},}
\end{equation}
such that $E^{(\mu)}$ is analytic in $U_\mu$ and such that {the matching condition becomes (by Lemma \ref{lemma:eq}), uniformly for $0<\mu\leq \mu_0$ and $n\rho(\mu)\mu\geq M$ as $n\to\infty$,}
\begin{equation}\label{reffor}
\begin{split}P^{(\mu)}(z)\left(P^{(\infty)}(z)\right)^{-1}&=
\gamma(z)^{-\sigma_3}\sP
\widehat\delta(z;\mu)^{-\sigma_3}\sP^{-1}
\mathcal{E}_B(n^2\zeta(z;\mu))
\sP\widehat\delta(z;\mu)^{\sigma_3}\sP^{-1}
\gamma(z)^{\sigma_3}
\\&=
\gamma(z)^{-\sigma_3}
\left(I+\mathcal O(n^{-1}\rho(\mu)^{-1}\mu^{-1/2})\right)
\gamma(z)^{\sigma_3},
\end{split}\end{equation}
for $z\in\partial U_\mu $.
Here we used that $\widehat\delta(z;\mu)$ is uniformly bounded from $0$ and $\infty$ in $z\in U_\mu .$
Henceforth,
since $\gamma(z)$ is $\mathcal O(\mu^{-1/4})$ on $\partial U_\mu,$
estimate \eqref{reffor} becomes
\begin{equation}\label{jump_Cs}
P^{(\mu)}(z)\left(P^{(\infty)}(z)\right)^{-1}=
I+\mathcal O(n^{-1}\rho(\mu)^{-1}\mu^{-1}).
\end{equation}

\subsection{Construction of the local parametrix $P^{(a)}$ near $a(\mu)$}\label{sectPhiAi}
We now fix, as before, a disk $U_a$ with sufficiently small radius {$\delta>0$ around $a(\mu)$,} such that $U_a$ does not intersect with $U_\mu.$
The aim of this section is to construct a function satisfying the following conditions.
\subsubsection*{RH problem for $P^{(a)}$}
\begin{enumerate}[label=(\alph*)]
\item $P^{(a)}$ is analytic in ${\color{black}U_a}\setminus \left(\Sigma_+\cup\Sigma_-\cup\mathbb R\right)$,
\item $P^{(a)}(z)$ satisfies exactly the same jump relations as $S$ for $z\in U_\mu\cap\left(\Sigma_+\cup\Sigma_-\cup\mathbb R\right)$,
\item $P^{(a)}$ matches with $P^{(\infty)}$ on the boundary of the disk, in the sense that
\begin{equation}\nonumber
P^{(a)}(z)\left(P^{(\infty)}(z)\right)^{-1}=I+\mathcal O(n^{-1})
\end{equation}
uniformly for $z\in\partial U_a$ { and $0<\mu\leq \mu_0$ }in the limit where $n\to\infty.$ 
\end{enumerate}
Recall the function $\varphi(z)$ defined in \eqref{def:phi}, which is analytic for {\color{black}$z\in U_a\setminus[\mu, a(\mu)],$}
and recall the conformal in $U_a$ map $\eta(.)=\eta(.;\mu)$ defined in the second of the formulas \eqref{phimuzeta},
\begin{equation}\label{phi_mu}\nonumber
\varphi(z) = - \frac23\eta(z)^{3/2},\quad z\in U_a.
\end{equation}
We will construct $P^{(a)}$ in terms of the Airy function in the following way.

\medskip

\noindent \textbf{Construction of a parametrix in terms  of Airy functions.}
Define 
\[
\Phi_{{\rm Ai}}(\eta) =
\begin{cases}
\begin{pmatrix}
v_0 & iv_1 \\ -iv_0' & v_1'
\end{pmatrix},&\arg\eta\in(0,2\pi/3),\\
\begin{pmatrix}
v_0 & iv_{-1} \\ -iv_0' & v_{-1}'
\end{pmatrix},&\arg\eta\in(-2\pi/3,0),
\\
\begin{pmatrix}
-\i v_{-1} & iv_1 \\ - v_{-1}' & v_1'
\end{pmatrix},&\arg\eta\in(2\pi/3,\pi),
\\
\begin{pmatrix}
\i v_1 & iv_{-1} \\ v_1' & v_{-1}'
\end{pmatrix},&\arg\eta\in(-\pi, -2\pi/3),
\end{cases}
\]
where
\[v_0(\eta) = \sqrt{2\pi\;} {\rm Ai}(\eta),\quad v_1(\eta) = \sqrt{2\pi\;}\e^{-\pi\i/6} {\rm Ai}(\eta\e^{-2\pi\i/3}),\quad 
v_{-1}(\eta) = \sqrt{2\pi\;}\e^{\pi\i/6} {\rm Ai}(\eta\e^{2\pi\i/3}).\]
Then, ${\rm \Phi}_{\rm Ai}$ solves  the following model RH problem.
\subsubsection*{RH problem for $\Phi_{{\rm Ai}}$}
\begin{enumerate}[label=(\alph*)]
\item $\Phi_{{\rm Ai}}$ is analytic for $\eta\in\mathbb{C}\setminus\Sigma_{{\rm Ai}},$ where $\Sigma_{{\rm Ai}}=(-\infty,+\infty)\cup e^{2\pi\i/3}(+\infty, 0)\cup e^{-2\pi\i/3}(+\infty, 0)$
{ is an oriented contour, with the orientation of the (half-)lines from the first mentioned point to the second one.}
\item $\Phi_{{\rm Ai}}$ satisfies the jumps $\Phi_{{\rm Ai},+}(\eta) = \Phi_{{\rm Ai},-}(\eta) J_{\Phi_{{\rm Ai}}}(\eta),$ where
\[
\begin{split}J_{\Phi_{{\rm Ai}}}(\eta) =\begin{cases} \begin{pmatrix}1 & 0 \\ 1 & 1\end{pmatrix}, &\eta\in e^{2\pi\i/3}(+\infty, 0)\cup e^{-2\pi\i/3}(+\infty, 0),
\\
\begin{pmatrix}1 & 1 \\ 0 & 1\end{pmatrix}, &\eta\in (0,+\infty),
\\
\begin{pmatrix}0 & 1 \\ -1 & 0 \end{pmatrix}, &\eta\in(-\infty,0),
\end{cases}\end{split}.\]
\item $\Phi_{{\rm Ai}}$ has the following asymptotics as $\eta\to\infty,$ uniformly w.r.t. $\arg\eta\in(-\pi,\pi):$
\clu{\begin{equation}\nonumber\label{eq:AiryRHPc}
\Phi_{{\rm Ai}}(\eta)=\eta^{-\sigma_3/4}
\mathcal{E}_{{\rm Ai}}(\eta)
\sP
\e^{-\frac23\,\eta^{3/2}\sigma_3},
\end{equation}}
\noindent where $\sP$ is defined in \eqref{sigmaPN}, and
$\mathcal{E}_{{\rm Ai}}$ has the following asymptotic expansion  as $\eta\to\infty,$ which is uniform w.r.t. $\arg\eta\in[-\pi,\pi]:$ for any integer $M\geq 1,$ we have
\begin{multline}\label{EAi}
\mathcal{E}_{{\rm Ai}}(\eta)=
\begin{bmatrix}
1+\sum\limits_{m=1}^{M}\frac{(12m-1)\cdot(12m-5)!!}{2^{8m}3^{4m-1}(2m)!(4m-3)!!\eta^{3m}}
&
\i\sum\limits_{m=1}^{M}
\frac{(12m-7)\cdot (12m-11)!!}{2^{8m-4}3^{4m-3}(2m-1)!(4m-5)!!\eta^{3m-3/2}} 
\\
\i\sum\limits_{m=1}^{M}\frac{(12m-5)\cdot(12m-11)!!}{2^{8m-4}3^{4m-3}(2m-1)!(4m-5)!!\eta^{3m-3/2}}
&
1+\sum\limits_{m=1}^M
\frac{-(12m+1)\cdot (12m-5)!!}{2^{8m}3^{4m-1}(4m-3)!!(2m)!\eta^{3m}}
\end{bmatrix}\\+\mathcal O(\eta^{-3M-3/2}).
\end{multline}
\end{enumerate}

\noindent
The form of the expansion \eqref{EAi} is not important at this point, but it will become crucial later on that
\begin{equation}\label{EAi2}
\mathcal{E}_{\rm Ai}(\eta)=I+\begin{pmatrix}
0 & \frac{5\i}{48} \\ \frac{7\i}{48} & 0
\end{pmatrix}\frac{1}{\eta^{3/2}}+\mathcal{O}(\eta^{-3}),\quad \eta\to\infty.
\end{equation}

\medskip

We define
\begin{equation}\nonumber
P^{(a)}(z)=E^{(a)}(z)\Phi_{\rm Ai}(n^{2/3}\eta(z;\mu))e^{-n\varphi(z;\mu)\sigma_3}\widehat\delta(z;\mu)^{\sigma_3}.
\end{equation}
If we set
\begin{equation}\nonumber
E^{(a)}(z)=P^{(\infty)}(z)\widehat\delta(z;\mu)^{-\sigma_3} \sP^{-1}\(n^{1/6}\eta(z;\mu)^{1/4}\)^{\sigma_3},
\end{equation}
then $E^{(a)}$ is analytic in $U_a$ and 
the matching condition becomes by Lemma \ref{lemma:eq}, as $n\to\infty,$ for $z\in\partial U_a$,
\begin{equation}\label{eq:Airymatching}
\begin{split}
P^{(a)}(z)\left(P^{(\infty)}(z)\right)^{-1}=&
\gamma(z)^{-\sigma_3}\sP \widehat\delta(z;\mu)^{-\sigma_3}\,\sP^{-1}\mathcal{E}_{\rm Ai}(n^{2/3}\eta(z;\mu))\sP\,
\widehat\delta(z;\mu)^{\sigma_3}\sP^{-1} \gamma(z)^{\sigma_3}
\\
=&\gamma(z)^{-\sigma_3} \left(I+\mathcal O(n^{-1})\right)\gamma(z)^{\sigma_3}.
\end{split}\end{equation}
Hence
\begin{equation}\label{jump_Ca}
P^{(a)}(z)\left(P^{(\infty)}(z)\right)^{-1}=I+\mathcal O(n^{-1}),
\end{equation}
uniformly {for $z$ on $\partial U_a$ and $0<\mu\leq \mu_0$ as $n\to\infty.$}

\subsection{Final transformation $S\mapsto R$}\label{sect_Final_pre}
Define
\begin{equation}\label{def:R}
R(z)=\begin{cases}
\sP S(z)P^{(\mu)}(z)^{-1}&\mbox{for }z\in U_\mu,\\
\sP S(z)P^{(a)}(z)^{-1}&\mbox{for }z\in U_a,\\
\sP S(z)P^{(\infty)}(z)^{-1}&\mbox{elsewhere},
\end{cases}
\end{equation}
where $\sP$ is defined in \eqref{sigmaPN}.
The function $R$ solves the following RH problem:
\subsubsection*{RH problem for $R$}
\begin{enumerate}[label=(\alph*)]
\item $R$ is analytic in $\mathbb C\setminus \left(\partial U_\mu\cup  \partial U_a\cup\widetilde\Sigma_+\cup\widetilde\Sigma_-\cup(a(\mu)+\delta ,+\infty)\right)$,
where $\widetilde\Sigma_\pm=\Sigma_\pm\setminus~(U_\mu\cup U_a)$, and $\delta $ is the radius of the disk $U_a$,
\item $R$ satisfies the jump relations $R_+(z)=R_-(z)J_R(z),$ where 
\[\begin{split}
&J_{R}(z) = P^{(\infty)}(z)\begin{pmatrix}1&0\\
\widehat\delta(z;\mu)^2 e^{-2n\varphi(z;\mu)}&1\end{pmatrix}
P^{(\infty)}(z)^{-1},  &z\in\widetilde\Sigma_\pm,\\\\
&J_R(x) = P^{(\infty)}(x)\begin{pmatrix}1&\widehat\delta(x;\mu)^{-2}e^{2n\varphi(x;\mu)}
\\0&1\end{pmatrix}P^{(\infty)}(x)^{-1}, &x\in(a(\mu)+\delta ,+\infty),
\\\\
&J_R(z)=P^{(\infty)}(z)P^{(a)}(z)^{-1},
 &z\in \partial U_a,
\\\\
&J_R(z)=P^{(\infty)}(z)P^{(\mu)}(z)^{-1}
,&z\in \partial U_\mu,
\end{split}
\]
\item $R$ has the asymptotics 
\begin{equation}\nonumber R(z)=I+\mathcal O(z^{-1}),\qquad \mbox{ as $z\to\infty$.} \end{equation}
\end{enumerate}
{
\begin{lemma}\label{lemma:JR1}
For any $\mu_0>0$ there exist $M,c>0$, such that as $n\to\infty$, we have uniformly for
$0<\mu\leq\mu_0$, $\mu\rho(\mu)n\geq M$, and for $z$ on the jump contour for $R$ that
\begin{align*}
&J_R(z)-I=\mathcal O(n^{-1}), &\mbox{for $z\in\partial U_a$},\\
&J_R(z)-I=\mathcal O(n^{-1}\rho(\mu)^{-1}\mu^{-1}), &\mbox{for $z\in\partial U_\mu$},\\
&J_R(z)-I=\mathcal O(e^{-cn\rho(\mu)\mu}), &\mbox{for $z$ on $\widetilde\Sigma_\pm$,}\\
&J_R(z)-I=\mathcal O(e^{-cnz^\beta}), &\mbox{for $z>a(\mu)+\delta$.}
\end{align*}
\end{lemma}

\begin{proof}
For $z$ on $\partial U_a$ and $\partial U_\mu $, the required estimates for $J_R(z)-I$ follow from \eqref{jump_Cs} and \eqref{jump_Ca}. 

For $z=u+iv$ on $\widetilde\Sigma_+$,
it follows from Lemma \ref{lemma:estimatephi} and from the shape of the lenses (recall that we imposed $v\geq \frac{\delta}{2}u$ on $\widetilde\Sigma_+$) that for some $c', c''>0$,
\[|e^{-2n\varphi(z)}|<e^{-2c'n\sqrt{v}\rho(\mu)}\leq e^{-2c''n\rho(\mu)\sqrt{u}}\leq e^{-2c''n\rho(\mu)\sqrt{\mu}},\]
which tends to $0$ rapidly in the limit under consideration.
It follows that
\[J_R(z)-I=P^{(\infty)}(z)\mathcal O\left(e^{-2c''n\rho(\mu){\sqrt{\mu}}}\  \widehat\delta(z)^{2}\right)P^{(\infty)}(z)^{-1}.\]
We also recall that $\widehat \delta(z)^{2}$  is uniformly bounded on $\widetilde\Sigma_+$ and that we have the uniform bound $P^{(\infty)}(z)=\mathcal O(|z-\mu|^{-1/4})$ for $z\in\widetilde\Sigma_+$, such that
\[J_R(z)-I=\mathcal O\left(\mu^{-1/2}e^{-2c''n\rho(\mu)\sqrt{\mu}} \right)=\mathcal O\left(n\rho(\mu)\sqrt{\mu}e^{-2c''n\rho(\mu)\sqrt{\mu}} \right)=\mathcal O\left(e^{-cn\rho(\mu)\mu} \right)\] 
as $n\to\infty$ with $0<\mu\leq\mu_0$, $\mu\rho(\mu)n\geq M$.
The analogous estimate on $\widetilde\Sigma_-$ is obtained by using the symmetry of the jump matrices.
For $z$ on $(a(\mu)+\delta,+\infty)$, 
we observe from \eqref{freg} that $f(z;\mu)$ is positive, hence by \eqref{varphi}, 
\[\varphi(z;\mu)\leq \varphi(a(\mu)+\delta)<-C\] for some $C>0$.
It follows that
\[J_R(z)-I=P^{(\infty)}(z)\mathcal O\left(e^{-2Cn}\widehat\delta(z)^{-2}\right)P^{(\infty)}(z)^{-1}=\mathcal O\left(e^{-2Cn} \right).\]
Moreover, for $z$ large, we can improve this estimate: by \eqref{def:phi} and \eqref{def:glog}, $\varphi(z;\mu)<-\frac{1}{4}z^\beta$ for $z$ large enough, and this yields the required estimate.
\end{proof}

}
It follows that the $L^2$-, $L^1$-, and $L^\infty$- norms of the entries of $J_R-I$ on the jump contour are of order $\mathcal O(n^{-1}\rho(\mu)^{-1}\mu^{-1})$. 
Using standard small-norm theory for RH problems, see e.g.\  \cite{Deift, DKMVZ2, DKMVZ1}, we can conclude from this that
\[R(z)=I+\mathcal O(n^{-1}\rho(\mu)^{-1}\mu^{-1}),\]
as $n\to\infty$, uniformly for $0<\mu\leq \mu_0$ and 
$n\rho(\mu)\mu\geq M$ sufficiently large, and uniformly in $z$, provided that we take $M>0$ sufficiently large such that Lemma \ref{lemma:JR1} holds.
Estimates for $R'$ can also be obtained, and we can use this to derive asymptotics for the logarithmic derivative $\partial_\mu\ln H_n(\mu)$, uniformly in $\mu$. In the next sections, we will explain how to refine the above RH analysis in such a way that we can extend the uniformity region in $\mu$.

\section{First refinement: large $n$ asymptotics for $\left(\frac{M}{n}\right)^2\leq \mu \leq \mu_0$}\label{section:7RH2}

In this section, we refine the RH analysis from the previous section in order to obtain large $n$ asymptotics valid uniformly for $\left(\frac{M}{n}\right)^2\leq \mu \leq \mu_0$
(which is, \clr{in general}, a larger region than the region $\frac{M}{\rho(\mu)n}\leq \mu \leq \mu_0$ from the previous section, but for $\beta$ small still smaller than the region $\left(\frac{M}{n\rho(\mu)}\right)^2\leq \mu\leq \mu_0$ which we will ultimately need; recall \eqref{fcoeff} and \eqref{f_rho}), for an arbitrary constant $\mu_0>0$ and a sufficiently large constant $M>0$.
To achieve this, we modify the choice of normalization point. This refinement is similar in spirit to the one used in \cite{DIKairy}, although in that work the authors used a conformal mapping instead of a different normalization point.
We will modify the constructions of $S, P^{(\infty)}, P^{(\mu)}, P^{(a)}, R$, and we will denote the corresponding objects here as $\widehat S, \widehat P^{(\infty)}, \widehat P^{(\mu)}, \widehat P^{(a)}, \widehat R$.
Note that this refinement does not allow to improve immediately the asymptotics for the function $Y$ from Section \ref{RHP_Y} itself, but it allows to obtain immediately an improved asymptotics for the ratio $Y^{-1}\partial_{\mu}Y,$ which is what we need in order to obtain $\partial_{\mu}\log H_n(\mu)$ by using Proposition \ref{prop_dif_ident}.

\subsection{Refinement of the RH analysis}\label{sect_refin1}
\paragraph{Normalization of $S$.}
Instead of normalizing $S$ at infinity such that $S(z)\to I$ as $z\to\infty$, we will now normalize $S$ at the point $z=-\mu.$ 
Let us define for this purpose
\begin{equation}\label{def:hatS}
\widehat S(z)=S(-\mu)^{-1}S(z).
\end{equation}
This will turn out to be a convenient choice later on, {because it will allow us to keep the size of the jump matrices near the points $\mu$ and $a(\mu)$ in balance}. 
{The precise location of the normalization point is not that important, as long as it is of the order $-\mu$ for small $\mu$.}

\paragraph{Normalization of the global parametrix.}
Next, we will define a modified global parametrix $\widehat P^{(\infty)}$ which is also normalized at the point $z=-\mu$. 
To this end, we let
\begin{equation}\label{def:Pinftyhat}\nonumber
\widehat P^{(\infty)}(z)
:=\widehat\gamma(z)^{-\sigma_3}
\sP,
\end{equation}
where
\begin{equation}\label{def:gammahat}
\widehat\gamma(z)=\frac{\gamma(z)}{\gamma(-\mu)}= \left(\frac{2\mu}{a(\mu)+\mu}\cdot\frac{z-a(\mu)}{z-\mu}\right)^{1/4}.
\end{equation}
Writing as before, with $\delta>0$ sufficiently small, $U_\mu $ for the disk of radius $\delta\mu$ around $\mu$, and $U_a$ for the disk of radius $\delta$ around $a(\mu)$, we have that $\widehat\gamma(z)$ and $\widehat P^{(\infty)}(z)$ are uniformly bounded for $z\in \partial U_\mu $, but on the other hand $\widehat\gamma(z)^{-1}$ and $\widehat P^{(\infty)}(z)$ blow up like $\mathcal O(\mu^{-1/4})$ for $z\in\partial U_a$, as $\mu\to 0$.

\medskip

Then, $\widehat P^{(\infty)}$ is a solution to the RH problem for $P^{(\infty)}$, but with condition (c) replaced by
\begin{itemize}
\item[(c)] 
\clu{$\widehat P^{(\infty)}(-\mu)=\sP $.}
\end{itemize}

\paragraph{Modified local parametrix near $\mu$.}

\noindent 
We define $\widehat P^{(\mu)}$ in the same way as $P^{(\mu)}$, but with $P^{(\infty)}$ replaced by $\widehat P^{(\infty)}$: namely, let
\begin{equation}\label{def:hatPs}
\widehat P^{(\mu)}(z)=\widehat E^{(\mu)}(z)\Psi_B(n^2\zeta(z; \mu))e^{-n(\varphi(z;\mu)\mp\pi\i)\sigma_3}\widehat\delta(z;\mu)^{\sigma_3}
,
\quad \pm\Im z>0,
\end{equation}
with
\begin{equation}\label{EhatB}
\widehat E^{(\mu)}(z)=\widehat P^{(\infty)}(z)
\widehat\delta(z;\mu)^{-\sigma_3}
\sP(-n^2\zeta(z; \mu))^{\sigma_3/4},
\end{equation}
where we recall that $\widehat\delta$ is given in \eqref{deltahat}.
Similarly as in \eqref{reffor}, we then have the matching condition
\begin{equation}\label{reffor2}\nonumber
\begin{split}
\widehat P^{(\mu)}(z)\left(\widehat P^{(\infty)}(z)\right)^{-1}&=
\widehat\gamma(z)^{-\sigma_3}\sP
\widehat\delta(z;\mu)^{-\sigma_3}\sP^{-1}
\mathcal{E}_B(n^2\zeta(z;\mu))
\sP\widehat\delta(z;\mu)^{\sigma_3}\sP^{-1}
\widehat\gamma(z)^{\sigma_3}
\\&=
\widehat\gamma(z)^{-\sigma_3}
\left(I+\mathcal O(n^{-1}\rho(\mu)^{-1}\mu^{-1/2})\right)
\widehat\gamma(z)^{\sigma_3},
\end{split}
\end{equation}
for $z\in\partial U_\mu $.
Since $\widehat\gamma(z)$ is bounded on $\partial U_\mu $, we get
\begin{equation}\label{jump_Cs2}
\widehat P^{(\mu)}(z)\left(\widehat P^{(\infty)}(z)\right)^{-1}=I+\mathcal O(n^{-1}\rho(\mu)^{-1}\mu^{-1/2}),
\end{equation}
as {$n\to\infty$, uniformly for $n\rho(\mu)\sqrt{\mu}\geq M$ and $0<\mu\leq\mu_0$. This improves the matching condition \eqref{jump_Cs}.}

\paragraph{Modified local parametrix near $a(\mu)$.}
We modify the local parametrix in $U_a$ by setting
\begin{equation}\nonumber
\widehat P^{(a)}(z)=\widehat E^{(a)}(z)\Phi_{\rm Ai}(n^{2/3}\eta(z;\mu))e^{-n\varphi(z;\mu)\sigma_3}\widehat\delta(z;\mu)^{\sigma_3},
\end{equation}
and
\begin{equation}\label{Eahat}\nonumber
\widehat E^{(a)}(z)=\widehat P^{(\infty)}(z)\widehat\delta(z;\mu)^{-\sigma_3} \sP^{-1}\(\frac{-3n}{2}\varphi(z)\)^{\sigma_3/6}.
\end{equation}
This results in the matching condition, similar to \eqref{eq:Airymatching},
\begin{equation}\label{eq:Airymatching2}
\begin{split}
\widehat P^{(a)}(z)\left(\widehat P^{(\infty)}(z)\right)^{-1}=&
\widehat\gamma(z)^{-\sigma_3}\sP \widehat\delta(z;\mu)^{-\sigma_3}\,\sP^{-1}\mathcal{E}_{Ai}(n^{2/3}\eta(z;\mu))\sP\,
\widehat\delta(z;\mu)^{\sigma_3}\sP^{-1} \widehat\gamma(z)^{\sigma_3}=
\\
=&\widehat\gamma(z)^{-\sigma_3} \left(I+\mathcal O(n^{-1})\right)\widehat\gamma(z)^{\sigma_3}=I+\mathcal O(n^{-1}\mu^{-1/2}),
\end{split}\end{equation}
for $z\in\partial U_a$, as $n\to\infty$ in such a way that $n\sqrt{\mu}\to\infty$.
Here we used the fact that $\widehat \gamma(z)$ is of the order $\mu^{1/4}$ due to our modified definition of $\widehat\gamma(z).$

Note that the matching between the global parametrix and the local parametrix near $a$ {has now become} worse, because of our refinement. This is the price we have to pay for the improvement of the matching condition between the global parametrix and the local parametrix near $\mu$.

\paragraph{Modification of $R$.}
Define
\begin{equation}\label{def:hatR}
\widehat R(z)=\begin{cases}
\sP\widehat S(z)\widehat P^{(\mu)}(z)^{-1}&\mbox{for }z\in U_\mu,\\
\sP\widehat S(z)\widehat P^{(a)}(z)^{-1}&\mbox{for }z\in U_a,\\
\sP\widehat S(z)\widehat P^{(\infty)}(z)^{-1}&\mbox{elsewhere}.
\end{cases}
\end{equation}
Then $\widehat R$ solves the RH problem:
\subsubsection*{RH problem for $\widehat R$}\label{sect_RHPRhat}
\begin{enumerate}[label=(\alph*)]
\item $\widehat R$ is analytic in $\mathbb C\setminus \left(\partial U_\mu\cup \partial U_a\cup\widetilde\Sigma_+\cup\widetilde\Sigma_-\cup(a(\mu)+\delta ,+\infty)\right)$,
\item $\widehat R(z)$ satisfies the jump relations $\widehat R_+(z)=\widehat R_-(z)J_{\widehat R}(z),$ where 
\[\begin{split}
&J_{\widehat R}(z) = \widehat P^{(\infty)}(z)\begin{pmatrix}1&0\\\widehat\delta(z;\mu)^2e^{-2n\varphi(z;\mu)}&1\end{pmatrix}
\widehat P^{(\infty)}(z)^{-1}, &z\in\widetilde\Sigma_\pm,\\
&J_{\widehat R}(x) = \widehat P^{(\infty)}(x)\begin{pmatrix}1&\widehat\delta(x;\mu)^{-2} e^{2n\varphi(x;\mu)}
\\0&1\end{pmatrix}\widehat P^{(\infty)}(x)^{-1},&x\in(a(\mu)+\delta ,+\infty),
\\
&J_{\widehat R}(z)=\widehat P^{(\infty)}(z)\widehat P^{(a)}(z)^{-1},
&z\in \partial U_a,
\\
&J_{\widehat R}(z)=\widehat P^{(\infty)}(z)\widehat P^{(\mu)}(z)^{-1},
  &z\in \partial U_\mu,
\end{split}
\]
\item $\widehat R$ has the asymptotics 
$\widehat R(-\mu)=I$.
\end{enumerate}
{
We have the following analogue of Lemma \ref{lemma:JR1}.
\begin{lemma}\label{lemma:JR2}
For any $\mu_0>0$ there exist $M,c>0$, such that we have uniformly for $\left(\frac{M}{n}\right)^2<\mu\leq\mu_0$,
 and for $z$ on the jump contour for $\widehat R$ that
\begin{align*}
&J_{\widehat R}(z)-I=\mathcal O(n^{-1}\mu^{-1/2}), &\mbox{for $z\in\partial U_a$},\\
&J_{\widehat R}(z)-I=\mathcal O(n^{-1}\rho(\mu)^{-1}\mu^{-1/2}), &\mbox{for $z\in\partial U_\mu$},\\
&J_{\widehat R}(z)-I=\mathcal O(e^{-cn\rho(\mu)\mu}), &\mbox{for $z$ on $\widetilde\Sigma_\pm$,}\\
&J_{\widehat R}(x)-I=\mathcal O(e^{-cnx^\beta}), &\mbox{for $x>a(\mu)+\delta$.}
\end{align*}
\end{lemma}
\begin{proof}
The proof is similar to that of Lemma \ref{lemma:JR1}, except for the behavior of $\widehat P^{(\infty)}$ that is different from the one of $P^{(\infty)}(z)$:
$\widehat P^{(\infty)}(z)$ remains uniformly bounded for $z$ on the jump contour close to $\mu$; for $z$ on the jump contour near $a(\mu)$, we have
$\widehat P^{(\infty)}(z)=\mathcal O(\mu^{-1/4})$. Implementing these changes in the proof of Lemma \ref{lemma:JR1}, we obtain the stated result. 
\end{proof}
}
It is straightforward to prove that this implies the following result, which we will utilize later.
\begin{corollary}
\label{corollary:JRhat}
Let $\|J_{\widehat R}-I\|_j$, $j=1,2,\infty$, denote the maximum of the $L^j(\Sigma_{\widehat R})$-norms of the entries of $J_{\widehat R}-I$.
For any $\mu_0>0$ there exist $M,c>0$, such that we have uniformly for $\left(\frac{M}{n}\right)^2<\mu\leq\mu_0$,
\[\|J_{\widehat R}-I\|_j=\mathcal O(n^{-1}\mu^{-1/2})\qquad \mbox{for $j=1,2,\infty$},\] 
and moreover 
\begin{equation}\label{integral}
\left\|\frac{J_{\widehat R}(\xi)-I}{\xi-\mu}\right\|_1=\mathcal O(n^{-1}\mu^{-1/2}),
\end{equation}
 where $\left\|\frac{J_{\widehat R}(\xi)-I}{\xi-\mu}\right\|_1$ is the maximum of the $L^1(\Sigma_{\widehat R})$-norms of the entries of 
 $\frac{J_{\widehat R}(\xi)-I}{\xi-\mu}$.
\end{corollary}
Later on, we will use the explicit form of the jump matrices on $\partial U_{a},$ $\partial U_{\mu}$,
\begin{equation}\label{JumpRhatExplicit}
\begin{split}
J_{\widehat R}(z)&=
\widehat\gamma(z)^{-\sigma_3}
\sP\widehat\delta(z;\mu)^{-\sigma_3}\sP^{-1}
\(\mathcal{E}_{Ai}(n^{\frac23}\eta(z;\mu))\)^{-1}
\sP\widehat\delta(z;\mu)^{\sigma_3}\sP^{-1}
\widehat\gamma(z)^{\sigma_3},\quad & z\in\partial U_a,
\\
J_{\widehat R}(z)&=
\widehat\gamma(z)^{-\sigma_3}
\sP
\widehat\delta(z;\mu)^{-\sigma_3}
\sP^{-1}
\(\mathcal{E}_{B}(n^2\zeta(z;\mu))\)^{-1}
\sP
\widehat\delta(z;\mu)^{\sigma_3}
\sP^{-1}
\widehat\gamma(z)^{\sigma_3},\quad & z\in\partial U_\mu.
\end{split}
\end{equation}

The conclusion of this refined RH analysis is that we improved the jump matrix of $\widehat R$ on $\partial U_\mu $, but in exchange the jump matrix of $\widehat R$ on $\partial U_a$ now behaves less good. For $\beta>1$, this trade off is convenient, and it allows us to obtain asymptotics for $\widehat R$, which will allow us next to compute asymptotics for $\partial_\mu\ln H_n(\mu)$. For $\beta\leq 1$, we will need one more refinement of the RH analysis.

To obtain asymptotics for $\widehat R$, $\widehat R'$, we follow the usual techniques of small-norm theory for RH problems, but adapted to the situation of a RH problems normalized at the point $z=-\mu$ instead of $z=\infty$.

\begin{proposition}\label{prop:asRhat}
As $n\to\infty$, we have uniformly $0<\mu\leq \mu_0$ and $n\sqrt{\mu}\geq M$ that
\[\widehat R(\mu)=I+\mathcal O(n^{-1}\mu^{-1/2}),\qquad \widehat R'(\mu)=\mathcal O(n^{-1}\mu^{-3/2}).\]
\end{proposition}
\begin{proof}
Consider the M\"obius transformation
\[w(z)=\frac{z-\mu}{z+\mu},\qquad z(w)=-\mu\frac{w+1}{w-1},\]
which maps the normalization point $z=-\mu$ to $w=\infty$. The function $\widehat R(z(w))$ is then analytic in $w$ except on the image under $w(.)$ of the jump contour for $\widehat R$, and on this image, it satisfies the jump relations
\[\widehat R(z(w))_+=\widehat R(z(w))_-J_{\widehat R}(z(w)).\] Moreover, $\widehat R(z(w))$ tends to $I$ as $w\to \infty$. We are then in a situation similar to that in \cite{DIKairy}. Using the arguments of \cite[Lemma 1]{DIKairy}, one can then use small norm theory for RH problems, Lemma \ref{lemma:JR2}, and Corollary \ref{corollary:JRhat} to conclude that 
\begin{equation}\label{eq:estimatefromDIK}\widehat R(z(w))=I+\mathcal O\left(n^{-1}\mu^{-1/2}(|w|+1)^{-1}\right)\quad \mbox{or}\quad \widehat R(z)=I+\mathcal O\left(n^{-1}(z+\mu)\right),\end{equation}
uniformly for $z(w)$ off the jump contour for $\widehat R$. In particular, this holds for $w=0$, $z=\mu$.

\medskip

For the estimate of the derivative $\widehat R'(\mu)$, let us return to the variable $z$ instead of $w$. 
From the RH conditions for $\widehat R$, we know that $\widehat R_+-\widehat R_-=\widehat R_-(J_{\widehat R}-I)$ on the jump contour $\Sigma_{\widehat R}$, and 
also that $\widehat R(-\mu)=I$.
It follows that the unique RH solution satisfies the following equation,
\begin{equation}\label{Rhat}
\widehat R(z) = I + \frac{z+\mu}{2\pi\i}
\int\limits_{\Sigma_{\widehat R}}\frac{\widehat R_-(\xi)(J_{\widehat R}(\xi)-I)(\xi)\ \d\xi}{(\xi+\mu)(\xi-z)}.
\end{equation}

\noindent 
Differentiating, we obtain
\begin{equation}\label{Rhatx}\widehat R'(z) = \frac{1}{2\pi\i}\int\limits_{\Sigma_{\widehat R}}\frac{\widehat R_-(\xi)(J_{\widehat R}(\xi)-I)\ \d\xi}{(\xi-z)^2},
\qquad
z\in\mathbb{C}\setminus\Sigma_{\widehat R}.
\end{equation}
For $z=\mu$, this implies
\[\widehat R'(\mu)=\frac{1}{2\pi\i}\int\limits_{\Sigma_{\widehat R}}\frac{(\widehat R_-(\xi) -I)(J_{\widehat R}(\xi)-I)\ \d\xi}{(\xi-\mu)^2}
+\frac{1}{2\pi\i}\int\limits_{\Sigma_{\widehat R}}\frac{(J_{\widehat R}(\xi)-I)\ \d\xi}{(\xi-\mu)^2}.
\]
Using the fact that $\frac{1}{\xi-\mu}=\mathcal O(\mu^{-1})$ and \eqref{integral} in the second term, and \eqref{eq:estimatefromDIK}, $\frac{\xi+\mu}{\xi-\mu}=\mathcal O(1)$, and \eqref{integral} in the first term, we obtain the required estimate.
\end{proof}

\section{Second refinement: asymptotics in the regime $\left(\frac{M}{\rho(\mu)n}\right)^2\leq \mu\leq \mu_0$}\label{section:8RH3}
In this section, we will further improve the RH analysis, in order to obtain uniform estimates as $n\to\infty$, uniformly for $\left(\frac{M}{\rho(\mu)n}\right)^2\leq \mu\leq \mu_0$, where $\mu_0$ is an arbitrary constant and $M$ a sufficiently large constant.

Observe that the jump $J_{\widehat R}$ for $\widehat R$ is $I+\mathcal{O}(\frac{1}{n \rho(\mu)\sqrt{\mu}})$ on $\partial U_\mu $ as $n\rho(\mu)\sqrt{\mu}\to\infty$, and is hence close to the identity. On $\partial U_a,$ different entries of the jump matrix $J_{\widehat R}$ have different order: we have
\begin{equation}\label{Jestimate}
J_{\widehat R}(z) = \widehat\gamma(z)^{-\sigma_3}
\(
I+\mathcal{O}\(\frac{1}{n}\)
\)
\widehat\gamma(z)^{\sigma_3}
=\begin{pmatrix}
1+\mathcal{O}\(\frac{1}{n }\)
&
\mathcal{O}\(\frac{1}{\sqrt{\mu}\,n}\)
\\
\mathcal{O}\(\frac{\sqrt{\mu}}{n}\)
&
1+\mathcal{O}\(\frac{1}{n}\)
\end{pmatrix},
\end{equation}
as $n\to\infty$, $\mu\to 0$, so it is only the $(1,2)$ entry of the jump matrix that prevents it from converging to the identity.
We will transform $\widehat R$ to $\widetilde R$ in such a way that the $(1,2)$ entry of the jump matrix for $\widetilde R$ becomes 
$\mathcal{O}\(\frac{1}{\sqrt{\mu}\, n^{m+1}}\)$ for sufficiently large integer $m$, and hence (by \eqref{fcoeff}) is also of the order $\mathcal{O}\(\frac{1}{n\rho(\mu)\sqrt{\mu}}\).$

To this end, 
define the nilpotent matrix $\sigma_+=\begin{pmatrix}0&1\\0&0\end{pmatrix}$, and define
\begin{equation}\label{def:Rtilde}
\widetilde R(z) = \widehat R(z)\left(I-F(z;\mu)\sigma_+\right),\end{equation}
where $F(z;\mu)$ is a scalar function, analytic in $z\in\mathbb{C}\setminus \partial U_{a}$ and such that $F(-\mu;\mu)=0$, which is to be determined.
We will search for $F$ in the form
\begin{equation}\label{def:F}F(z;\mu)=\sum\limits_{j=1}^{m}\frac{F_{j}(z;\mu)}{n^j\sqrt{\mu}},\end{equation}
where  $m\in\mathbb N$ and $F_{j}(z;\mu), j=1,\ldots,m,$ are uniformly bounded functions that we will construct below.
Using the nilpotency of $\sigma_+$, it is straightforward to verify that the function $\widetilde R$ solves the following RH problem (see Figure \ref{fig:R}):
\subsubsection*{RH problem for $\widetilde R$}
\begin{enumerate}[label=(\alph*)]
\item $\widetilde R$ is analytic in $\mathbb C\setminus \left(\partial U_\mu\cup \partial U_a\cup\widetilde\Sigma_+\cup\widetilde\Sigma_-\cup(a(\mu)+\delta ,+\infty)\right)$,
\item $\widetilde R(z)$ satisfies the jump relations $\widetilde R_{+}(z)=\widetilde R_{-}(z)J_{\widetilde R}(z),$ where 
\begin{equation}\label{def:Jtilde}J_{\widetilde R}(z)= \left(I+F_-(z)\sigma_+\right) J_{\widehat R}(z) \left(I-F_+(z)\sigma_+\right) ,\end{equation} 
\item $\widetilde R(-\mu)=I$.
\end{enumerate}
{
We will construct $F_j$ to be analytic in $\mathbb C\setminus\partial U_a$ and such that it is uniformly bounded in $z$ for $\mu$ sufficiently small and $n$ sufficiently large. Moreover, we need it to be small near $0$: $F_j(z)=\mathcal O(z+\mu)$ uniformly for $|z|<\delta$ for some small $\delta>0$, and thus in particular on $\mathcal O(\partial U_\mu)$.
In this way, $F(z;\mu)$ is uniformly small on $\partial U_\mu$, and we can observe that the jump matrices $J_{\widetilde R}$ for $\widetilde R$ will be, like the ones for $\widehat R$, $I+\mathcal O\left(\frac{1}{n\sqrt{\mu}\rho(\mu)}\right)$ as $n\rho(\mu)\sqrt{\mu}$ is large, except possibly on $\partial U_a$.
}

To determine the coefficients $F_j(z)$, we need to take a closer look at the structure of the jump matrices $J_{\widehat R}$ and $J_{\widetilde R}$ on $\partial U_a$.
By \eqref{JumpRhatExplicit} and \eqref{EAi},
we can write $J_{\widehat R}$ for $z$ on the circle $\partial U_a$ in the form
\begin{equation}
\label{JRapprox}
J_{\widehat R}(z;\mu)=
\begin{pmatrix}
1+A(z;\mu) & \frac{B(z;\mu)}{\sqrt{\mu}}
\\
\sqrt{\mu} C(z;\mu) & 1+D(z;\mu) 
\end{pmatrix},
\end{equation}
where $A,B,C,D$ are functions of $z$ that depend also on $n$ and $\mu$, and that admit the following asymptotic expansions for large~$n:$
\begin{equation}\label{ABCD}
A(z;\mu)\sim\sum_{j=1}^{\infty}\frac{A_j(z;\mu)}{n^j},
\ 
B(z;\mu)\sim\sum_{j=1}^{\infty}\frac{B_j(z;\mu)}{n^j},
\
C(z;\mu)\sim\sum_{j=1}^{\infty}\frac{C_j(z;\mu)}{n^j},
\
D(z;\mu)\sim\sum_{j=1}^{\infty}\frac{D_j(z;\mu)}{n^j}.
\end{equation}
Moreover, each of the coefficients
$A_j(z;\mu), B_j(z;\mu), C_j(z;\mu), D_j(z;\mu)$ has a finite limit as $\mu\to 0$.

Substituting \eqref{JRapprox} in \eqref{def:Jtilde}, we obtain
\begin{equation}
J_{\widetilde R}=
\begin{pmatrix}
1+A+F_-C\sqrt{\mu}&F_--F_++\frac{B}{\sqrt{\mu}}-F_+A+DF_- -F_-F_+C\sqrt{\mu}\\
C\sqrt{\mu}&1+D-CF_+\sqrt{\mu}
\end{pmatrix}.
\end{equation}
Now, substituting the asymptotic expansions \eqref{def:F} and \eqref{ABCD},
we obtain 
\begin{equation}
J_{\widetilde R, 12}(z;\mu)\sim\frac{1}{\sqrt{\mu}}\sum_{j=1}^\infty \frac{F_{j,-}(z;\mu)-F_{j,+}(z;\mu)+B_j(z;\mu)+J_j(z;A,C,D,F_1,\ldots, F_{j-1})}{n^j},
\end{equation}
where $J_j(z;A,B,C,D,F_1,\ldots, F_{j-1})$ depends on 
the quantities 
\[A_k(z;\mu), \ C_k(z;\mu),\ D_k(z;\mu), \ \quad \mbox{for $k=1,\ldots, j$, and on }F_k(z;\mu)\quad\mbox{for $k=1,\ldots, j-1$,}\]
but the crucial observation is that it does not depend on $F_k(z)$ for $k\geq j$.
We have in particular
\begin{align*}
&J_1=0,\\
&J_2=-F_{1,+}A_1+D_1F_{1,-},\\
&J_3=-F_{1,+}A_2-F_{2,+}A_1+D_1F_{2,-}+D_2F_{1,-}-F_{1,-}F_{1,+}C_1.
\end{align*}
Hence, we can construct $F_1(z;\mu), F_2(z;\mu), \ldots$ iteratively by requiring that
\begin{equation}
F_{j,+}(z;\mu)-F_{j,-}(z;\mu)=B_j(z;\mu)+J_j(z;\mu),\qquad F_{j}(-\mu;\mu)=0.
\end{equation}
This yields, by the Sokhotsky-Plemelj formula,
\begin{equation}\label{eq:FSP}
F_j(z;\mu)=\frac{z+\mu}{2\pi i}\int_{\partial U_a}\frac{\left(B_j(\xi;\mu)+J_j(\xi;\mu)\right)d\xi}{(\xi-z)(\xi+\mu)}.
\end{equation}
By construction, the structure of the jump matrix $J_{\widetilde R}(z)$ on the disk $\partial U_{a}$ is as follows:
\begin{equation}\label{JRrefined}
J_{\widetilde R}(z)=
\begin{pmatrix}
1+\mathcal{O}\(\frac{1}{n}\) & \mathcal{O}\(\frac{1}{\sqrt{\mu}\,n^{m+1}}\)
\\
\mathcal{O}\(\frac{\sqrt{\mu}}{n}\) & 1+\mathcal{O}\(\frac{1}{n}\) 
\end{pmatrix}.\end{equation}
Moreover, each function $F_j(z;\mu)$ is analytic in $z\in\mathbb{C}\setminus\partial U_a,$ continuous up to the boundary, and satisfies the normalization condition $F_j(-\mu;\mu)=0.$ Moreover, the functions $F(z;\mu), F_j(z;\mu)$, their derivatives and boundary values are uniformly bounded for $n$ sufficiently large and $\mu$ sufficiently small.
From \eqref{eq:FSP} and \eqref{def:F}, we easily obtain the uniform estimate on $\partial U_\mu$ 
\begin{equation}\label{eq:Festimatemu}F(z;\mu)=\mathcal O(\sqrt{\mu}/n),\qquad z\in\partial U_\mu,
\end{equation}
which will be needed later.

Let us now take $m$ an integer larger than $\frac{1}{\beta}-1$ for $\beta<1,$ $m\geq 1$ for $\beta=1,$ and $m\geq 0$ for $\beta>1.$  
It is then straightforward to verify by \eqref{fcoeff} that  $J_{\widetilde R}(z)=I+\mathcal{O}\(\(n\sqrt{\mu}\rho(\mu)\)^{-1}\)$ as $n\to\infty$, uniformly for $z\in\partial U_a$ and for $\left(\frac{M}{\rho(\mu)n}\right)^2\leq \mu\leq \mu_0$.

In what follows we will need an explicit expression for the $z$-derivative of $F_1(z;\mu)$ at the point $z=-\mu$.

\begin{lemma}\label{lem_F1}
We have
\begin{multline*}
\partial_z F_1(z;\mu)\Big|_{z=\mu} 
\\
= \frac{\i\sqrt{a(\mu)+\mu}}{4(a(\mu)-\mu)^{2} f(a(\mu);\mu)\sqrt{2}}
\(1-2\alpha^2 \frac{(\sqrt{a(\mu)}-\sqrt{\mu})^2}{a(\mu)} +\frac{{\partial_z}f(z;\mu)_{z=a(\mu)}\,(a(\mu)-\mu)}{2f(a(\mu);\mu)}\).
\end{multline*}
\end{lemma}
\begin{proof}
By \eqref{JumpRhatExplicit} and \eqref{EAi}, we have
\[J_{\widehat R}(z)=I
+
\begin{pmatrix}
\frac{1}{16}\(\widehat\delta(z)^2-\frac{1}{\widehat\delta(z)^2}\)
&
\frac{-\i}{48\widehat\gamma(z)^2}
\(3\widehat\delta(z)^2+\frac{3}{\widehat\delta(z)^2}-1\)
\\
\frac{-\i \widehat\gamma(z)^2}{48}
\(3\widehat\delta(z)^2+\frac{3}{\widehat\delta(z)^2}+1\)
&
\frac{-1}{16}\(\widehat\delta(z)^2-\frac{1}{\widehat\delta(z)^2}\)
\end{pmatrix}\frac{1}{n\eta(z;\mu)^{3/2}}
+
\mathcal{O}\(\frac{1}{n^2\eta(z;\mu)^3}\)
\]
uniformly for $z\in\partial U_a$, as $n\to\infty$ with $\mu\leq \mu_0$ and $n\rho(\mu)\sqrt{\mu}$ sufficiently large.
It follows from \eqref{JRapprox}
that
\[B_1(z;\mu)=
\frac{-\i\sqrt{\mu}}{48 \eta(z;\mu)^{3/2} \widehat\gamma(z)^2}
\(3\widehat\delta(z)^2+\frac{3}{\widehat\delta(z)^2}-1\).\]

Using \eqref{eq:FSP}, we obtain
{\begin{equation}\label{intermComp2}
\begin{split}
&
F_1(z;\mu)
=
\frac{z+\mu}{2\pi\i}\int\limits_{\partial U_a}
\frac{-\i\(3\widehat\delta(\xi)^2+3\widehat\delta(\xi)^{-2}-1\)}
{48\widehat\gamma(\xi)^2\,\eta(\xi; \mu)^{3/2}}
\cdot
\frac{\sqrt{\mu}\ \d\xi}{(\xi-z)(\xi+\mu)},
\\
&
\partial_zF_1(z;\mu)
=
\frac{1}{2\pi\i}\int\limits_{\partial U_a}
\frac{-\i\(3\widehat\delta(\xi)^2+3\widehat\delta(\xi)^{-2}-1\)}
{48\widehat\gamma(\xi)^2\,\eta(\xi; \mu)^{3/2}}
\cdot
\frac{\sqrt{\mu}\ \d\xi}{(\xi-z)^2}.
\end{split}
\end{equation}}
These integrals can be computed by computing the residues of the integrands. To perform the computations, we set $\xi = a(\mu) + (a(\mu)-\mu)v,$
and find, using
\eqref{phimuzeta}, \eqref{delta}, \eqref{deltahat}, \eqref{def:gammahat} that
\begin{equation}
\label{details_a}
\begin{split}
&\widehat\gamma(\xi)=\left(\frac{2\mu v}{(a(\mu)+\mu)(1+v)}\right)^{1/4},
\\
&\widehat\delta(\xi)^{\pm 2}
=
1 \mp 2\alpha\,\frac{\sqrt{a(\mu)}-\sqrt{\mu}}{\sqrt{a(\mu)}}\sqrt{v}
+
2 \alpha^2 \, \frac{(\sqrt{a(\mu)}-\sqrt{\mu})^2}{a(\mu)} v
+\mathcal{O}(v^{3/2}),\ v\to0,
\\
&
\eta(\xi; \mu)^{3/2}
=
\frac12 f(a(\mu); \mu)\,(a(\mu)-\mu)\sqrt{v}
\(1+\frac35 \(\frac{-1}{2}+\frac{\partial_zf(z;\mu)_{z=a(\mu)}(a(\mu)-\mu)}{f(a(\mu); \mu)}\)v+\mathcal{O}(v^{2})\), v\to 0,
\end{split}
\end{equation}
and
\begin{multline}\label{intermComp}
\frac{-\i\(3\widehat\delta(\xi)^2+3\widehat\delta(\xi)^{-2}-1\)}
{48\,\widehat\gamma(\xi)^2\,\eta(\xi;\mu)^{3/2}}
=
\frac{-5\i\sqrt{a(\mu)+\mu}}{24f(a(\mu);\mu)(a(\mu)-\mu)\,\sqrt{2\mu}}\cdot\frac{1}{v^2}
\\
\cdot
\(
1+\frac45v\(1+3\alpha^2
\(\sqrt{\frac{\mu}{a(\mu)}}-1\)^2 -\frac{3\partial_zf(z;\mu)_{z=a(\mu)}\,(a(\mu)-\mu)}{4f(a(\mu);\mu)}\)
+\mathcal{O}(v^2)
\)
\end{multline}
as $v\to0.$ Combining this with
\[
\frac{\d\xi}{(\xi-z)(\xi + \mu)}
=
\frac{(a(\mu)-\mu)\, \d v}{(a(\mu)-z)(a(\mu)+\mu)}
\(
1-
\(\frac{a(\mu)-\mu}{a(\mu)-z}+\frac{a(\mu)-\mu}{a(\mu)+\mu}\)v
+\mathcal{O}(v^2)
\),\quad v\to0,
\]
and substituting all this in \eqref{intermComp2},
we obtain after computing residues
an 
explicit expression for $F_1(z;\mu)$ and its $z-$derivative for $z$ outside of $U_a,$
\begin{multline*}
F_1(z;\mu) = 
\frac{-\i(z+\mu)\sqrt{a(\mu)+\mu}}{6\sqrt{2} f(a(\mu);\mu)(a(\mu)-z)(a(\mu)+\mu)}
\\
\cdot
\left[
1+
3\alpha^2\(\sqrt{\frac{\mu}{a(\mu)}}-1\)^2
-\frac{3\partial_zf(z;\mu)_{z=a(\mu)}\,(a(\mu)-\mu)}{4f(a(\mu); \mu)}
-\frac{5(a(\mu)-\mu)}{4(a(\mu)-z)}
-\frac{5(a(\mu)-\mu)}{4(a(\mu)+\mu)}
\right],
\end{multline*}
\begin{multline*}
\partial_zF_1(z;\mu)
=
\frac{-\i (z+\mu)\sqrt{a(\mu)+\mu}}
{6\sqrt{2}(a(\mu)-z)^2 f(a(\mu); \mu) }
\\
\cdot
\left[
1+
3\alpha^2\(\sqrt{\frac{\mu}{a(\mu)}}-1\)^2
-\frac{3\partial_zf(z;\mu)_{z=a(\mu)}\,(a(\mu)-\mu)}{4f(a(\mu); \mu)}
-\frac{5(a(\mu)-\mu)}{2(a(\mu)-z)}
\right].
\end{multline*}
Substituting $z=\mu$ in the latter, we obtain the statement of the lemma.
\end{proof}

\begin{remark}\label{rem_normalization}
{Though we do not need it in the present paper, it is natural to ask whether the above analysis allows one to obtain the asymptotics for the original function $Y$ from RH problem \ref{RHP_Y} itself, and not only for its logarithmic derivative $Y^{-1}\partial_z Y.$
So far we did not obtain the asymptotics of the function $S$ from the  RH problem of Section \ref{sect:RH:S}, but instead we studied the function $\widehat S$ from \eqref{def:hatS}, related to $S$ via the unknown quantity $S(-\mu).$
We can indeed obtain asymptotics for $S$, and thus for $Y$, using the expression $S(z)=\widehat S(\infty)^{-1} \widehat S(z),$ where the asymptotics for both factors on the right are known.}
\end{remark}

{

\begin{proposition}\label{prop:asRtilde}
As $n\to\infty$, uniformly for $\left(\frac{M}{\rho(\mu)n}\right)^2\leq \mu\leq \mu_0$, we have
\[\widetilde R(\mu)=I+\mathcal O(n^{-1}\rho(\mu)^{-1}\mu^{-1/2}),\qquad \widetilde R'(\mu)=
\mathcal O(n^{-1}\rho(\mu)^{-1}\mu^{-3/2}),\]
and
\begin{equation}\label{eq:R12prime}\widetilde R'_{12}(\mu)=
\frac{-i\sqrt{a(\mu)}}{8n\rho(\mu)\sqrt{2\mu}}\left(4\alpha^2 \mu^{-1}+{\frac{\partial_zf(z;\mu)_{z=\mu}}{\rho(\mu)}}+\mathcal O(\mu^{-1/2})\right).
\end{equation}
\end{proposition}

\begin{proof}
The estimates
\[\widetilde R(\mu)=I+\mathcal O(n^{-1}\rho(\mu)^{-1}\mu^{-1/2}),\qquad \widetilde R'(\mu)=\mathcal O(n^{-1}\rho(\mu)^{-1}\mu^{-3/2})\]
are proved similarly as in Proposition \ref{prop:asRhat}, and we omit the details.
Also similarly as in Proposition \ref{prop:asRhat}, one shows that
$R_-(z)-I=\mathcal O(n^{-1}\rho(\mu)^{-1}\mu^{-3/2}(z+\mu))$.
It remains to compute the leading order term in $\widetilde R'(\mu)$ explicitly.
For this, we use the expression
\begin{equation}\label{Rtildex}\widetilde R'(z) = \frac{1}{2\pi\i}\int\limits_{\Sigma_{\widetilde R}}\frac{\widetilde R_-(\xi)(J_{\widetilde R}(\xi)-I)\ \d\xi}{(\xi-z)^2},
\qquad
z\in\mathbb{C}\setminus\Sigma_R.
\end{equation}
For $z=\mu$, the contribution of $\Sigma_R\setminus\partial U_\mu$ to this integral is $\mathcal O(\frac{1}{n^2{\mu^{3/2}}\rho(\mu)^2})$, so we have
\begin{equation}\label{Rtildex2}\widetilde R'(\mu) = \frac{1}{2\pi\i}\int\limits_{\partial U_\mu}\frac{(\widetilde R_-(\xi)-I)(J_{\widetilde R}(\xi)-I)\ \d\xi}{(\xi-\mu)^2}+\frac{1}{2\pi\i}\int\limits_{\partial U_\mu}\frac{(J_{\widetilde R}(\xi)-I)\ \d\xi}{(\xi-\mu)^2}+\mathcal O\left(\frac{1}{n^2{\mu}^{3/2}\rho(\mu)^2}\right).
\end{equation}
Next, we use the above estimate for $R_-(z)-I$, $\frac{\xi+\mu}{\xi-\mu}=\mathcal O(1)$, and an estimate similar to \eqref{integral} for $J_{\widetilde R}-I$ to conclude that the first term in this expression is also 
$\mathcal O\left(n^{-2}\rho(\mu)^{-2}\mu^{-3/2}\right)$, so 
we have
\begin{multline}\label{Rtildex3}\widetilde R'(\mu) = \frac{1}{2\pi\i}\int\limits_{\partial U_\mu}\frac{(J_{\widetilde R}(\xi)-I)\ \d\xi}{(\xi-\mu)^2}+\mathcal O\left(\frac{1}{n^2{\mu}^{3/2}\rho(\mu)^2}\right)\\= \frac{1}{2\pi\i}\int\limits_{\partial U_\mu}\frac{(J_{\widehat R}(\xi)-I)\ \d\xi}{(\xi-\mu)^2}+\mathcal O\left(\frac{1}{n^2{\mu}^{3/2}\rho(\mu)^2}\right),\end{multline}
where we used \eqref{eq:Festimatemu} in the second step.
\medskip

We will now use a residue computation to evaluate the $(1,2)$-entry of the remaining integral at the right.
By \eqref{JumpRhatExplicit} and \eqref{EB2}, we have
\[J_{\widehat R}(z)-I=
\frac{1}{n\sqrt{-\zeta(z;\mu)}}\widehat\gamma(z)^{-\sigma_3}\sP\widehat\delta(z)^{-\sigma_3}\sP^{-1}
\begin{pmatrix} 0 &\frac{3\i}{8}\\\frac{\i}{8}&0 \end{pmatrix}
\sP\widehat\delta(z)^{\sigma_3}\sP^{-1} \widehat\gamma(z)^{\sigma_3}+\mathcal O(n^{-2}\zeta(z;\mu)^{-2})\]
for $z\in\partial U_\mu$, and
straightforward computations yield
\begin{equation}\label{A}
J_{\widehat R,12}(z) 
=
\frac{-i}
{8n\sqrt{-\zeta(z)}\widehat\gamma(z)^2}
\(\widehat\delta(z)^2+\frac{1}{\widehat\delta(z)^2}+1\)+\mathcal O(n^{-2}\zeta(z;\mu)^{-2}),
\end{equation}

Recall the change of variable $z = \mu-(a(\mu)-\mu)w$ from \eqref{eq:changevarw}.
By \eqref{varphi},
\begin{equation}\label{expansionzeta}
\sqrt{-\zeta(z;\mu)}
=\rho(\mu)\,(a(\mu)-\mu)\,\sqrt{w}
\(1+\(\frac16-\frac{\partial_zf(z;\mu)_{z=\mu}\,(a(\mu)-\mu)}{3\,f(\mu;\mu)}\)w+\mathcal{O}(w^2)\)\ \mbox{ as }
w\to0.\end{equation}
By \eqref{delta_s}, we obtain
\[
\widehat\delta(z)^2+\widehat\delta(z)^{-2} = 2 +4\alpha^2\(\frac{\sqrt{a(\mu)}}{\sqrt{\mu}}-1\)^2w+\mathcal{O}(w^{2}),\qquad w\to0.
\]
Substituting all this and \eqref{def:gammahat} in \eqref{A} and using Lemma \ref{lemma:eq}, we get in particular that
\[J_{\widetilde R, 12}(z)=
\frac{-i\sqrt{a(\mu)}}{8n\rho(\mu)\sqrt{2\mu}}\left(h_0(\mu)+h_1(\mu)w+\mathcal O(w^2)\right),\]
uniformly in $\mu,n$ as $w\to 0$, where
\[h_0(\mu)=3+\mathcal O(\mu),\qquad h_1(\mu)=4\alpha^2 \mu^{-1}+\frac{\partial_zf(z;\mu)_{z=\mu}}{f(\mu;\mu)}+\mathcal O(\mu^{-1/2}),\qquad \mbox{uniformly in $\mu, n$}.\]
Substituting this in
\eqref{Rtildex3}, we obtain the result by computing the residue of the integrand, taking into account the change of variable $z\mapsto w$.

\end{proof}

\section{Asymptotic form of the differential identity}\label{section:9proofs}

\color{black}
In this section we compute the asymptotics of $\partial_{\mu}\ln H_n(\mu).$ 
To make the exposition uniform regardless of whether $\beta>1$ or $0<\beta\leq 1,$ we take $m>\max(0, \frac{1}{\beta}-1)$ (and hence, $m\geq1$), in the construction of the function $F$ in Section \ref{section:8RH3} (even though for $\beta>1$ we could also proceed with $m=0$).

Our objective in this section will be to prove the following result.
\begin{proposition}\label{theor_main}
Let $\beta>0, $ $\alpha\in\mathbb{R}$ be fixed parameters.
For any $\mu_0>0$, there exists $M>0$ such that we have uniformly in $\left(\frac{M}{\rho(\mu)n}\right)^2\leq\mu\leq\mu_0$ as $n\to\infty$ that
\[
\begin{split}
\partial_\mu\ln \widetilde H_n(\mu)&=D_2(\mu) n^2 + D_1(\mu) n+D_0(\mu)+r(\mu;n),
\end{split}\]
where
\[
\begin{split}
&D_2(\mu) = -\frac14 (a(\mu)-\mu) \rho(\mu)^2,
\qquad 
D_1(\mu) = 
\frac{\alpha}{2}
\rho(\mu)\(\frac{\sqrt{a(\mu)}}{\sqrt{\mu}}-1\),
\end{split}
\]
\[
\begin{split}
&
{D_0(\mu)
=
-\frac{1}{8}
\(2\alpha^2\mu^{-1}+\frac{{\partial_x f(x;\mu)_{x=\mu}}}{2\rho(\mu)}
\)}
,
\end{split}
\]
and
\[
\color{black}r(\mu;n) = \mathcal{O}\(\frac{1}{n \mu^{3/2} \,\rho(\mu)}\)+\mathcal O(\mu^{-1/2}).
\color{black}
\]
\end{proposition}

{
\begin{remark}\label{D_0:full}
With more effort, we can prove a slightly stronger version of the above result: if we replace $
D_0(\mu)$ by
\[\begin{split}
&
D_0^{\mathrm{full}}(\mu)
=
\frac{1}{8(a(\mu)-\mu)}
\(1-2\alpha^2\(\sqrt{\frac{a(\mu)}{\mu}}-1\)^2-\frac{(a(\mu)-\mu)\,{\partial_x}f(x;\mu)_{x=\mu}}{2f(\mu;\mu)}
\)
\\
&
\hskip1.3cm
-\frac18\frac{f(\mu;\mu)}{f(a(\mu) ;\mu)\,(a(\mu)-\mu)}
\(
1-2\alpha^2\(1-\sqrt\frac{\mu}{a(\mu)}\)^2+\frac{(a(\mu) -\mu)\,{\partial_x}f(x;\mu)_{x=a(\mu)}}{2f(a(\mu);\mu)}
\),
\end{split}
\]
then the error terms in $r(\mu;n)$ are improved to $\mathcal{O}\(\frac{1}{n\mu^{3/2}\rho(\mu)}\).$ This allows us to prove the estimate referred to in Remark \ref{remark:strongerestimate1}.
\end{remark}
}

Recall the expression \eqref{DI2} for $\partial_\mu  \ln H_n(\mu)$ in terms of the RH solution $Y$, and the transformations $Y\mapsto T\mapsto S\mapsto \widehat S\mapsto \widehat R\mapsto \widetilde R$ in the RH analysis, defined in \eqref{def:T}, \eqref{def:S}, \eqref{def:hatS}, \eqref{def:hatR}, and \eqref{def:Rtilde}. 
Inverting these transformations, we can express $Y$ in terms of $\widetilde R$ for $z\to\mu_-$ as follows:
\begin{equation}\label{Y_final}
Y(z) = \underbrace{\delta(\infty)^{\sigma_3}\e^{n\ell_{\mu}\sigma_3}S(-\mu;\mu)\sP^{-1}}_{\textrm{matrix independent of $z$}}
\cdot
\widetilde R(z)(1+F(z)\sigma_+)\widehat P^{(\mu)}(z)\e^{ng(z;\mu)\sigma_3}
\underbrace{\e^{-n\ell_{\mu}\sigma_3}\delta(z)^{-\sigma_3}}_{\textrm{diagonal matrix}}.
\end{equation}
Substituting this expression in \eqref{DI2} and 
 deducing from \eqref{deltainf}, \eqref{def:phi}, \eqref{phippphim}
 that $\delta(\mu_-) = \mu^{\alpha/2}$, $g(\mu_-\pm\i0)=\frac12\mu^{\beta/2}+\ell_{\mu}\pm\pi\i,$ 
 we obtain the identity

\begin{equation}\label{Hankel_R} \partial_\mu\ln \widetilde H_n(\mu)
=
\dfrac{-1}{2\pi\i}\(\( \widetilde R(x;\mu)Q(x)  \widehat P^{(\mu)}(x;\mu)\)^{-1}\cdot\partial_x\( \widetilde R(x;\mu)Q(x) \widehat P^{(\mu)}(x;\mu)\)\)_{21}\Big|_{x=\mu_-},
\end{equation}
where $\widehat P^{(\mu)}$ is given by \eqref{def:Rtilde} and \eqref{def:hatPs},
 and $Q$ is given by $Q(z)=I+F(z)\sigma_+$.

Next, by \eqref{def:hatPs} and \eqref{EhatB}, we have 
\[\begin{split} 
\widetilde R(x;\mu) Q(x;\mu)\widehat P^{(\mu)}(x;\mu) =&  \widetilde R(x;\mu)Q(x;\mu)\widehat E^{(\mu)}(x)\Psi_B(n^2\zeta(x;\mu)) \e^{-n(\varphi(x)\mp\pi\i)\sigma_3}\widehat\delta(x)^{\sigma_3}.
\end{split}
\]

By \eqref{phippphim} and \eqref{deltahat}, the diagonal matrix at the right of $\Psi_B$ is equal to $I$ at $z=\mu$. Substituting this into \eqref{Hankel_R}, we get
\begin{equation} 
\begin{split}
&\partial_\mu\ln \widetilde H_n(\mu)
=\\
&=
\dfrac{-1}{2\pi\i}\(\( \widetilde R(x;\mu)Q(x)\widehat E^{(\mu)}(x)\Psi_B(n^2\zeta(x;\mu))\)^{-1}\cdot\partial_x\( \widetilde R(x;\mu) Q(x)\widehat E^{(\mu)}(x)\Psi_B(n^2\zeta(x;\mu))\)\)_{21}\Big|_{x=\mu_-}.
\end{split}\label{exprlogder}
\end{equation}
The expression between parentheses above decomposes into 
\begin{equation}\label{sum_dec}
\Psi^{-1}\partial_x\Psi \ +\ 
 \Psi^{-1} E^{-1}\partial_x E\cdot \Psi \ +\  
\Psi^{-1} E^{-1}  (Q^{-1}\partial_x  Q )  E \Psi \ +\ 
\Psi^{-1} E^{-1} Q^{-1}  \widetilde R^{-1}\partial_x  
\widetilde R \cdot  Q E \Psi,
\end{equation}
where we abbreviated
\[\Psi:=\Psi_B(n^2\zeta(x;\mu)),\qquad E:=\widehat E^{(\mu)}(x),\qquad Q:=Q(x;\mu).\]
We will compute the first two terms in the above expression in Lemma's \ref{propA1} and \ref{propB1}, and estimate the third term in Lemma \ref{propABQ}. The fourth term will be computed in Lemma \ref{lem_Rtilde}.

\begin{lemma}\label{propA1}
We have
\[\frac{-1}{2\pi\i}\left(\Psi^{-1}\partial_x\Psi
\right)_{21}\Big|_{x=\mu_-}
=D_2(\mu)n^2
\quad\mbox{
and }\quad
\Psi
\begin{pmatrix}1\\0\end{pmatrix}\Big|_{x=\mu_-} = \sqrt{\pi}
\begin{pmatrix}1\\0\end{pmatrix}.
\]
\end{lemma}
\begin{proof}
We have 
\[\Psi^{-1}\partial_x\Psi = 
\Psi_B(n^2\zeta(x;\mu))^{-1}\partial_x\Psi_B(n^2\zeta(x;\mu))
=
n^2\partial_x\zeta(x;\mu)\Psi_B(n^2\zeta(x;\mu))^{-1}\Psi_B'(n^2\zeta(x;\mu)).
\]
Furthermore, since $\Psi$  has determinant $1$, only the first column contributes to the $(2,1)$ entry of $\Psi_B^{-1}\Psi'_B.$
Recall from \eqref{Psi_B} that the first column of $\Psi_B$ can be expanded as $\zeta\to 0_+$ in the following way,
\[\Psi_B(\zeta )\begin{pmatrix}1\\0\end{pmatrix} = 
\sqrt{\pi}
\begin{pmatrix}
I_0(\sqrt{-\zeta }) \\-\i\sqrt{-\zeta }I_0'(\sqrt{-\zeta })
\end{pmatrix}
=\sqrt{\pi}
\begin{pmatrix}
1-\frac{\zeta }{4}+\mathcal{O}(\zeta^2)
\\
\frac{\i}{2}\zeta +\mathcal{O}(\zeta^2)
\end{pmatrix},
\]
where $\mathcal{O}$-terms can be differentiated.
Hence 
\[
\left(
\Psi_B(0_+)^{-1}\Psi_B'(0_+)
\right)_{21}
=
\begin{pmatrix}0&1\end{pmatrix}\Psi_B(0_+)^{-1}\Psi_B'(0_+)\begin{pmatrix}1\\0\end{pmatrix}
= 
\pi\begin{pmatrix}
0&1
\end{pmatrix}
\begin{pmatrix}
-\frac{\i}{4}\\\frac{\i}{2}
\end{pmatrix}
=
\frac{\pi\i}{2}.
\]
Furthermore, recall from 
\eqref{phimuzeta} and \eqref{derzeta}
 that $$\zeta(x;\mu) = \rho(\mu)^2(a(\mu)-\mu)(x-\mu)(1+\mathcal{O}(x-\mu))\qquad \mbox{as $x\to\mu$},$$
 where the $\mathcal{O}$-term can be differentiated, and thus 
 $$\partial_z\zeta(x;\mu)_{x=\mu}= \rho(\mu)^2(a(\mu)-\mu),$$
By the chain rule, this easily yields the statement of the proposition.
\end{proof}
\noindent We proceed with the second term in \eqref{sum_dec}.
\begin{lemma}\label{propB1}
We have
\[\frac{-1}{2\pi\i}\left(\Psi^{-1} E^{-1}(\partial_x E) \Psi\right)_{21}\Big|_{x=\mu_-}
=D_1(\mu)n\]
and
\[
E\Psi\begin{pmatrix}1\\0\end{pmatrix}\Big|_{x=\mu_-}
=
\i\sqrt{\pi \rho(\mu)(a(\mu)-\mu)n}\,
\left(\frac{2\mu}{a(\mu)+\mu}\right)^{1/4}\begin{pmatrix}
0 \\ 1
\end{pmatrix}.
\]
\end{lemma}
\begin{proof}
By the second part of Proposition \ref{propA1}, we obtain
\[\left(\Psi^{-1}E^{-1}(\partial_x E) \Psi\right)_{21}\Big|_{x=\mu_-}
=
\begin{pmatrix}0&1
\end{pmatrix}
\Psi^{-1}E^{-1}(\partial_x E) \Psi\begin{pmatrix}1\\0
\end{pmatrix}
\Big|_{x=\mu_-}
=
\pi\begin{pmatrix}0&1
\end{pmatrix}
E^{-1}\partial_x E\begin{pmatrix}1\\0
\end{pmatrix}\Big|_{x=\mu_-},
\]
and 
\[
E\Psi\begin{pmatrix}1\\0\end{pmatrix}\Big|_{x=\mu_-}
=\sqrt{\pi}E\begin{pmatrix}1\\0\end{pmatrix}\Big|_{x=\mu_-},
\]
and it remains to compute  $E$. Since $\det E=1,$ the $(2,1)$ entry of $E^{-1}\partial_x E$ involves only the first column of $E.$
By \eqref{EhatB}, we have 
\[
E
=
\widehat\gamma(z)^{-\sigma_3}\sP\widehat\delta(z)^{-\sigma_3}\sP
(-n^2{\zeta(z;\mu)})^{\sigma_3/4},
\]
and careful computations lead us to an expression for the first column, 
\begin{equation}\label{insidepropB1}
E\begin{pmatrix}1 \\ 0\end{pmatrix}
=
\frac{(-n^2{\zeta(z;\mu)})^{1/4}}{2}
\begin{pmatrix}
-
\frac{1}{\widehat\gamma(z)}(\widehat\delta(z)-\frac{1}{\widehat\delta(z)})
\\
\i\widehat\gamma(z)(\widehat\delta(z)+\frac{1}{\widehat\delta(z)})
\end{pmatrix}.
\end{equation}
To compute the limits of the expression above and its derivative when $z$ approaches $\mu_-$, we 
use \eqref{delta_s}, \eqref{def:gammahat}, and \eqref{expansionzeta} (with
$w=\frac{\mu-z}{a(\mu)-\mu}$). Substituting all these ingredients in \eqref{insidepropB1}, we obtain
\[
E\begin{pmatrix}1\\0\end{pmatrix}
=
\sqrt{nf(\mu;\mu)(a(\mu)-\mu)}
\begin{pmatrix}
-\left(\frac{a(\mu)+\mu}{2\mu}\right)^{1/4}\alpha\(\sqrt{\frac{a(\mu)}{\mu}}-1\) w + \mathcal{O}(w^2)
\\
\i\left(\frac{2\mu}{a(\mu)+\mu}\right)^{1/4}\(1
\hskip-1mm
+
\hskip-1mm
\(
\hskip-1mm
\frac{7}{24}-\frac{\partial_zf(z;\mu)_{z=\mu}\,(a(\mu)-\mu)}{12f(\mu;\mu)}
\hskip-1mm
+
\hskip-1mm
\frac{\alpha^2}{2}\frac{(\sqrt{a(\mu)}-\sqrt{\mu})^2}{\mu}\)w \right. 
\\
\hskip8.8cm
+ \mathcal{O}(w^2)
\Big)
\end{pmatrix},
\]
where the $\mathcal{O}$-terms can be differentiated.
Thus 
\[
\begin{split}
&
E\begin{pmatrix}1\\0
\end{pmatrix}\Big|_{w=0_+}
=
\i\sqrt{n f(\mu;\mu)(a(\mu)-\mu)}\left(\frac{2\mu}{a(\mu)+\mu}\right)^{1/4}\begin{pmatrix}0\\1
\end{pmatrix},
\\
&
\partial_wE\begin{pmatrix}1\\0
\end{pmatrix}\Big|_{w=+0}
=\sqrt{n f(\mu;\mu)(a(\mu)-\mu)}
\begin{pmatrix}
-\left(\frac{a(\mu)+\mu}{2\mu}\right)^{1/4}\alpha\(\sqrt{\frac{a(\mu)}{\mu}}-1\)
\\
\i\left(\frac{2\mu}{a(\mu)+\mu}\right)^{\hskip-0.5mm1/4}
\hskip-1.5mm \(\frac{7}{24}-\frac{\partial_zf(z;\mu)_{z=\mu}\,(a(\mu)-\mu)}{12f(\mu;\mu)}
+\frac{\alpha^2}{2}\frac{(\sqrt{a(\mu)}-\sqrt{\mu})^2}{\mu}\)
\end{pmatrix}\hskip-0.5mm,
\end{split}
\]
and hence 
\[\begin{split}
&\left(E^{-1}\partial_wE
\right)_{21}\Big|_{w=0_+}
=
\begin{pmatrix}0&1\end{pmatrix}
E^{-1}\partial_wE
\begin{pmatrix}1\\0\end{pmatrix}\Big|_{w=+0}
=\i n f(\mu;\mu)(a(\mu)-\mu)\left(\frac{2\mu}{a(\mu)+\mu}\right)^{1/4}
\\
&
\qquad\hskip3cm\cdot
\begin{pmatrix}-1&0\end{pmatrix}
\begin{pmatrix}
-\left(\frac{a(\mu)+\mu}{2\mu}\right)^{1/4}\alpha\(\sqrt{\frac{a(\mu)}{\mu}}-1\)
\\
\i\left(\frac{2\mu}{a(\mu)+\mu}\right)^{1/4} \(\frac{7}{24}-\frac{\partial_zf(z;\mu)_{z=\mu}\,(a(\mu)-\mu)}{12f(\mu;\mu)}
+\frac{\alpha^2}{2}\frac{(\sqrt{a(\mu)}-\sqrt{\mu})^2}{\mu}\)
\end{pmatrix}
\\
&=
\i n  f(\mu;\mu)(a(\mu)-\mu)\alpha\(\sqrt{\frac{a(\mu)}{\mu}}-1\).
\end{split}
\]
To finish the proof it is sufficient to observe that 
$\partial_z = \frac{-1}{a(\mu)-\mu}\partial_w.$
\end{proof}
Next, we deal with the third term in \eqref{sum_dec}.
\begin{lemma}\label{propABQ}
As $n\to\infty$, we have uniformly for $\left(\frac{M}{\rho(\mu)n}\right)^2\leq \mu\leq \mu_0$ that 
\[-\frac{1}{2\pi\i}\left(\Psi^{-1}E^{-1}Q^{-1}(\partial_xQ) E
\Psi\right)_{21}\Big|_{x=\mu_-}
=\mathcal O(\rho(\mu)).
\]
{Moreover, 
\[QE\Psi\begin{pmatrix}1\\0\end{pmatrix}\Big|_{x=\mu_-} = i\sqrt{\pi\rho(\mu)(a(\mu)-\mu)n}\sqrt[4]{\frac{2\mu}{a(\mu)+\mu}}\begin{pmatrix}F(\mu)\\1\end{pmatrix}.\]}
\end{lemma}
\begin{proof}
By Proposition \ref{propB1}, we have
\[\begin{split}
-\frac{1}{2\pi\i}\left(\Psi^{-1}E^{-1} Q^{-1}(\partial_x Q) E\Psi\right)_{21}\Big|_{x=\mu_-}
\hskip-1mm
&=
-\frac{1}{2\pi\i}
\begin{pmatrix}
0 & 1
\end{pmatrix}\Psi^{-1}E^{-1} Q^{-1}(\partial_x Q) E\Psi
\begin{pmatrix}
1 \\ 0
\end{pmatrix}\Big|_{x=\mu_-}
\\
&=
\frac{\i \rho(\mu)(a(\mu)-\mu)n\sqrt{\frac{2\mu}{a(\mu)+\mu}}}{2}
\left(Q^{-1}\partial_x Q
\right)_{12}\Big|_{x=\mu_-},
\end{split}
\]
which allows us to concentrate just on the factor involving $Q.$
Note that the $(2,1)$ element of the matrix on the left is expressed in terms of $(1,2)$ element of the matrix on the right. Since $Q$ is upper-triangular, we have
\[
\left(
Q^{-1}\partial_xQ\right)_{12}\Big|_{x=\mu_-}
=\partial_z F(z;\mu)\Big|_{z=\mu_-}.\]
Hence,
\[-\frac{1}{2\pi\i}\left(\Psi^{-1}E^{-1}Q^{-1}
(\partial_xQ) E\Psi\right)_{21}\Big|_{x=\mu_-}
=
\frac{\i}{2} \rho(\mu)(a(\mu)-\mu)n\sqrt{\frac{2\mu}{a(\mu)+\mu}}\, \partial_z F(z;\mu)_{z=\mu}.\]
By \eqref{def:F} and the fact that the derivatives of $F_j$ are bounded, we have
$\partial_z F(z;\mu)_{z=\mu}=\mathcal O\left(\frac{1}{n\sqrt{\mu}}\right),$
{and the first result follows. The second statement is obtained similarly.}
\end{proof}

\color{black}

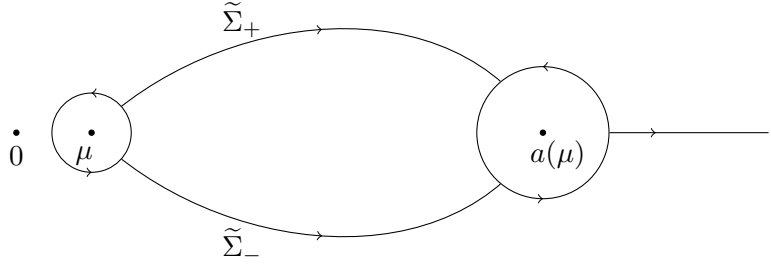
\begin{figure}[ht!]
\center{\center
\begin{tikzpicture}

\draw [decoration={markings, mark = at position 0.5 with {\arrow{>}}}, postaction=decorate](1,0) to [out=45, in =120] (7,0);
\draw [decoration={markings, mark = at position 0.5 with {\arrow{>}}}, postaction=decorate](1,0) to [out=-45, in =-120] (7,0);
\draw [decoration = {markings, mark = at position 0.5 with {\arrow{>}}}, postaction=decorate](7,0) -- (10,0);

\node at (3,1.5) {$\widetilde\Sigma_+$};
\node at (3,-1.5) {$\widetilde\Sigma_-$};

\filldraw[white] (7,0) circle (25pt);
\draw[decoration = {markings, mark = at position 0.25 with {\arrow{>}}}, decoration = {markings, mark = at position 0.75 with {\arrow{>}}}, 
postaction={decorate}] (7,0) circle (25pt);

\filldraw[white] (1,0) circle (15pt);
\draw[decoration = {markings, mark = at position 0.25 with {\arrow{>}}}, decoration = {markings, mark = at position 0.75 with {\arrow{>}}}, 
postaction={decorate}] (1,0) circle (15pt);

\node at (0,-0.3) {$0$};
\node at (0.9,-0.3) {$\mu$};
\node at (7.2,-0.3) {$a(\mu)$};

\filldraw (0,0) circle (1pt);
\filldraw (7,0) circle (1pt);
\filldraw (1,0) circle (1pt);

\end{tikzpicture}}
\caption{Contour for the RH problem for the function $\widetilde R(z).$}
\label{fig:R}
\end{figure}

Here we treat the fourth term in \eqref{sum_dec}.
\begin{lemma}\label{lem_Rtilde}
As $n\to\infty$, we have uniformly for $\left(\frac{M}{\rho(\mu)n}\right)^2\leq \mu\leq \mu_0$ that 
\[-\frac{1}{2\pi\i}
\left(\Psi^{-1}E^{-1}Q^{-1}
\widetilde R^{-1}
(\partial_x\widetilde R)
QE\Psi\right)_{21}\Big|_{x=\mu_-}
=D_0(\mu)
+\mathcal O(\rho(\mu))+\mathcal O(\mu^{-1/2})+\mathcal O\left(\frac{1}{n\mu^{3/2}\rho(\mu)}\right). 
\]
\end{lemma}
\begin{proof}
Using {Proposition \ref{propABQ}}, we obtain
\begin{multline}\label{Rtildeaux}
-\frac{1}{2\pi\i}\left(\Psi^{-1}E^{-1}Q^{-1} \widetilde R^{-1}\partial_x \widetilde R \cdot QE\Psi\right)_{21}\Big|_{x=\mu_-}
\\=
\frac{-\i}{2} \rho(\mu)(a(\mu)-\mu)n\sqrt{\frac{2\mu}{a(\mu)+\mu}}
\begin{pmatrix}
1&-F(\mu)
\end{pmatrix}
\widetilde R^{-1}\partial_x \widetilde R
\begin{pmatrix}
F(\mu) \\ 1
\end{pmatrix}
\Big|_{x=\mu_-},
\end{multline}
and
moreover by \eqref{eq:Festimatemu} and Proposition \ref{prop:asRtilde},
\begin{multline}\label{Rtildeaux2}
-\frac{1}{2\pi\i}\left(\Psi^{-1}E^{-1}Q^{-1} \widetilde R^{-1}(\partial_x \widetilde R) QE\Psi\right)_{21}\Big|_{x=\mu_-}
\\
=
\frac{-\i}{\sqrt{2}} \rho(\mu)\sqrt{a(\mu)}n\sqrt{\mu}
\partial_x \widetilde R_{12}(\mu)
+\mathcal O\left(\frac{1}{n\rho{\mu}^{3/2}}\right)
+\mathcal O\left(\frac{1}{\sqrt{\mu}}\right)+\mathcal O\left(1\right)
.
\end{multline}
We now substitute \eqref{eq:R12prime}
to obtain the result.
\end{proof}

In order to prove Proposition \ref{theor_main}, it is now sufficient to substitute the results from Lemma's \ref{propA1}--\ref{lem_Rtilde} in \eqref{exprlogder}--\eqref{sum_dec} and to observe that $\rho(\mu)=\mathcal O(\mu^{-1/2})$.

\section{Integration of the differential identity. Proof of Proposition \ref{prop:Hankel}}\label{sec:integration}

{In this section, we first derive asymptotics  
for $\widetilde H_n(\mu)$ as $n\to\infty$ for a fixed $\mu>0$; for this, we use results from \cite{CG19}. Then, we integrate the asymptotic form of the differential relation obtained in the previous section. 
In order to do that, we need some integral identities, which we collect first.}

\subsection{Some integrals}
\begin{lemma}\label{lemma:I}
Let
\be\label{defI}
I_{-j}(\mu)=\int_{\mu}^{a(\mu)} t^{\bt/2-j}\frac{dt}{\sqrt{(t-\mu)(a(\mu)-t)}},\qquad J(\mu)= \int_{\mu}^{a(\mu)} t^{\bt/2-2}\sqrt{\frac{t-\mu}{a(\mu)-t}}dt, 
\ee
where $a(\mu)$ is the unique solution of \eqref{defa}. (In particular, $\rho(\mu)=\frac{\bt}{2\pi}I_{-1}$, where $\rho$ is defined in \eqref{def_rho}).
Then
\begin{align}
&\mathrm{(a)}\qquad
I_0(\mu)=\mu I_{-1}(\mu)+ \frac{4\pi}{\bt},\label{lemma:a}\\
&\mathrm{(b)}\qquad
I_1(\mu)=\frac{4\pi(\beta+1)(a(\mu)+\mu)}{\beta(\beta+2)}
+
\frac{2\pi\(a(\mu)+\mu+\beta\mu\)\mu\rho(\mu)}{\beta(\beta+2)}\label{lemma:newb}\\
&\mathrm{(c)}\qquad
I_{-2}(\mu)=\frac{1}{\mu}(I_{-1}(\mu)-J(\mu)),\label{lemma:b}\\
&\mathrm{(d)}\qquad
J(\mu)=\frac{1}{(\bt/2-1) a(\mu)}\left( 2\pi-\frac{a(\mu)-\mu}{2} I_{-1}(\mu)\right),\label{lemma:c}\\
&\mathrm{(e)}\qquad
I_{-1}(\mu)=\frac{4\pi a'(\mu)}{a(\mu)-\mu a'(\mu)},\label{lemma:d}\\
&\mathrm{(f)}\qquad
f(a(\mu);\mu)=\frac{1}{a'(\mu)}\rho(\mu).\label{lemma:e}
\end{align}
\end{lemma}

\begin{proof}
(a) By the definition \eqref{defa} of $a(\mu)$, we have
\[
1=\frac{\bt}{4\pi}\int_{\mu}^{a(\mu)} t^{\bt/2-1}\frac{t-\mu}{\sqrt{(t-\mu)(a(\mu)-t)}}dt=\frac{\bt}{4\pi}(I_0(\mu)-\mu I_{-1}(\mu))
\]
and \eqref{lemma:a} follows.
\medskip

\noindent
(b) follows by integrating $\frac{d}{d x}\(x^{\beta/2}(x-a(\mu))^{\frac12}(x-\mu)^{\frac12}\)$
over a contour $\Gamma$ going around the interval $[\mu,a(\mu)]$ and not intersecting $(-\infty,0]$.

\noindent
(c) follows from the observation
\[
I_{-2}(\mu)-\frac{1}{\mu}I_{-1}(\mu)=\int_{\mu}^{a(\mu)} t^{\bt/2-1}\frac{dt}{\sqrt{(t-\mu)(a(\mu)-t)}}\left(\frac{1}{t}-\frac{1}{\mu}\right)=-\frac{J(\mu)}{\mu}.
\]
\medskip

\noindent
(d) First note that by \eqref{defa},
\begin{multline}\label{intint}
\frac{\beta}{4\pi}J(\mu)-\frac{1}{a(\mu)}=
\frac{\beta}{4\pi}\int_{\mu}^{a(\mu)} t^{\bt/2-1}\sqrt{\frac{t-\mu}{a(\mu)-t}}\left(\frac{1}{t}-\frac{1}{a(\mu)}\right)dt\\=
\frac{\beta}{4\pi a(\mu)}\int_{\mu}^{a(\mu)} t^{\bt/2-2}\sqrt{(t-\mu)(a(\mu)-t)}dt.
\end{multline}
Integrating by parts and observing that by (a)
\[
\int_{\mu}^{a(\mu)} t^{\bt/2-1}\sqrt{\frac{a(\mu)-t}{t-\mu}}dt=a(\mu) I_{-1}(\mu)-I_0(\mu)=(a(\mu)-\mu)I_{-1}(\mu)-\frac{4\pi}{\bt},
\]
we can write the r.h.s. of \eqref{intint} as
\[
-\frac{\bt (a(\mu)-\mu)}{8\pi(\beta/2-1)a(\mu)}I_{-1}(\mu)+\frac{1}{(\beta/2-1)a(\mu)},
\]
and \eqref{lemma:c} follows.
\medskip

\noindent
(e) In \eqref{defa} make the change of the integration variable $t=a(\mu) u$. Then
\[
1=\frac{\bt}{4\pi} a(\mu)^{\bt/2} \int_{\mu/a(\mu)}^1 u^{\bt/2-1} \sqrt{\frac{u-\mu/a(\mu)}{1-u}}du.
\]
In this form, the expression is easy to differentiate w.r.t. $\mu$, and we obtain \eqref{lemma:d}.
\medskip

\noindent
(f) We now write \eqref{defa} in the form
\[
{
1=\frac{\bt}{8\pi\i}\int_{\Gamma} t^{\bt/2-1}\left(\frac{t-\mu}{t-a(\mu)}\right)^{1/2}dt,
}
\]
where $\Gamma$ is a counter-clock-wise oriented closed loop around the interval $[\mu,a(\mu)]$ and not intersecting $(-\infty,0]$.
Differentiating this identity with respect to $\mu$, we obtain
\[
{0=-\frac{1}{2} 2I_{-1}(\mu)+\frac{1}{2\i}a'(\mu) \int_{\Gamma} t^{\bt/2-1}\left(\frac{t-\mu}{t-a(\mu)}\right)^{1/2}\frac{dt}{t-a(\mu)},
}
\]
which gives \eqref{lemma:e} since
\[
\rho(\mu)=\frac{\bt}{2\pi}I_{-1}(\mu),\qquad {
f(a(\mu);\mu)=\frac{\bt}{4\pi\i}\int_{\Gamma} t^{\bt/2-1}\left(\frac{t-\mu}{t-a(\mu)}\right)^{1/2}\frac{dt}{t-a(\mu)}.
}
\]
\end{proof}

\subsection{Large $n$ asympotics for $\widetilde H_n(\mu)$ for fixed $\mu>0$}\label{sect_fixed_mu}
A particular case of Theorem 1.2 from \cite{CG19} provides the large $n$ asymptotics for a Hankel determinant with a weight $e^{W(\xi)}e^{-nV(\xi)}$, which is such that the corresponding equilibrium measure is supported on the interval $\xi\in[-1,1]$ and has the form $d\mu_{CG}=\sqrt{\frac{1-\xi}{1+\xi}}f_{CG}(\xi),$ where $f_{CG}(\xi)>0$ for $\xi\in[-1,1].$
Applying this result in the case
\[ V(\xi)=\(\frac{a(\mu)-\mu}{2}\xi+\frac{a(\mu)+\mu}{2}\)^{\beta/2},
\qquad W(\xi)=\alpha\ln\(\frac{a(\mu)-\mu}{2}\xi+\frac{a(\mu)+\mu}{2}\),\]
after the  rescaling 
$x=\frac{a(\mu)-\mu}{2}\xi+\frac{a(\mu)+\mu}{2}$, 
yields
\begin{equation}\label{eq:CG}
\begin{split}
\ln\widetilde{H}_n(\mu)
=
\widehat{C}_2(\mu) n^2 + \widehat{C}_1(\mu) n - \frac16 \ln n + \widehat{C}_0(\mu) 
+ \mathcal{O}\(\frac{\ln n}{n}\),
\end{split}
\end{equation}
where the error term is uniform in $\varepsilon\leq \mu\leq \mu_0$ as $n\to\infty$, for any $\varepsilon,\mu_0>0$, and
\[
\begin{split}
&\widehat C_2(\mu) = \ln\frac{a(\mu)-\mu}{4}-\frac32-\frac{1}{\pi(a(\mu)-\mu)}
\int\limits_{\mu}^{a(\mu)}\(x^{\frac{\beta}{2}}-\frac{4(x-\mu)}{a(\mu)-\mu}\)\sqrt{\frac{a(\mu)-x}{x-\mu}}dx
\\
&\hskip4.5cm -\frac{1}{4\pi}\int\limits_{\mu}^{a(\mu)}\(x^{\frac{\beta}{2}}-\frac{4(x-\mu)}{a(\mu)-\mu}\)\sqrt{\frac{a(\mu)-x}{x-\mu}}\ f(x;\mu)dx,
\end{split}
\]
\[\begin{split}
&
\widehat C_1(\mu) = \ln(2\pi)+\frac{\alpha}{2\pi}\int\limits_{\mu}^{a(\mu)}\ln x \sqrt{\frac{a(\mu)-x}{x-\mu}}f(x;\mu)dx,
\end{split}\]
\[\begin{split}
& \widehat C_0(\mu) = 2\zeta'(-1)-\frac18\ln\(\frac{a(\mu)-\mu}{4} \rho(\mu) \) - \frac{1}{24}\ln\(\frac{a(\mu)-\mu}{4}f(a(\mu);\mu)\)
+\\
&\qquad\qquad\quad+
\frac{\alpha^2}{4\pi^2}
\int\limits_{\mu}^{a(\mu)}\frac{\ln x}{\sqrt{(x-\mu)(a(\mu)-x)}}\
\dashint_{\mu}^{a(\mu)}\frac{\sqrt{(y-\mu)(a(\mu)-y)}\, dy}{y(x-y)}dx,
\end{split}
\]
where $\dashint$ is the principal value integral.
Next we will prove that $\widehat C_j(\mu), j=0,1,2,$ coincide with $C_j(\mu)$ given in \eqref{C1234}.

\subsubsection*{$\bullet\ \widehat C_2(\mu) = C_2(\mu).$}
Using \eqref{f_def} and a contour deformation argument, we can rewrite $\int\limits_{\mu}^{a(\mu)}x^{\frac{\beta}{2}}\sqrt{\frac{a(\mu)-x}{x-\mu}}\, f(x;\mu)dx$ as a double complex contour integral, over {$\Gamma_1\times\Gamma_2$ in the integration variables $x$ and $t$, where the contours $\Gamma_1$ and $\Gamma_2$ encircle the segment $[\mu, a(\mu)]$ in the counter-clockwise direction and do not intersect each other and the half-line $(-\infty,0].$}
By exchanging the role of $x$ and $t$ in this representation, we obtain a second representation. Taking the average of these two representations, we can factor out $t-x$ in the numerator of the integrand, and this factor cancels out with the denominator. The double integral then decouples into products of single integrals, and in this way we obtain the identity
\begin{multline*}
-\frac{1}{4\pi}\int\limits_{\mu}^{a(\mu)}x^{\frac{\beta}{2}}\sqrt{\frac{a(\mu)-x}{x-\mu}}\, f(x;\mu)dx
=
\frac{\beta}{16\pi^2}\(-I_0^2(\mu)+2\mu I_0(\mu) I_1(\mu)-\mu a(\mu) I_{-1}^2(\mu)\)
\\=
\frac{-1}{\beta}-\frac{1}{4\beta}\mu(a(\mu)-\mu)\rho(\mu)^2.
\end{multline*}

Next, to compute 
\begin{multline*}
\frac{1}{\pi(a(\mu)-\mu)}\int\limits_{\mu}^{a(\mu)}\sqrt{(x-\mu)(a(\mu)-x)}\ f(x;\mu)dx\\
=
\frac{\beta}{2\pi(a(\mu)-\mu)}
\(
I_1(\mu)-\frac{a(\mu)+3\mu}{2}I_0(\mu)+\frac{\mu(a(\mu)+\mu)}{2}I_{-1}(\mu)
\)
=
\frac{\beta(a(\mu)-\mu)-4\mu}{(\beta+2)(a(\mu)-\mu)}+\frac{\mu\rho(\mu)}{\beta+2},
\end{multline*}
we again write the integral as a double contour integral, and then change the order of integration.

Furthermore, 
\[
\frac{-1}{\pi(a(\mu)-\mu)}\int\limits_{\mu}^{a(\mu)}x^{\beta/2}\sqrt{\frac{a(\mu)-x}{x-\mu}}dx
=
\frac{1}{\pi}\(I_1(\mu)-a(\mu)I_0(\mu)\)
=
\frac{-4(a(\mu)-\mu)+4\beta\mu}{(a(\mu)-\mu)\beta(\beta+2)}-\frac{2(\beta+1)\mu\rho(\mu)}{\beta(\beta+2)},
\]
and combining all the above pieces we obtain $\widehat C_2(\mu)=C_2(\mu).$

\subsubsection*{$\bullet\ \widehat C_1(\mu) = C_1(\mu).$}
We substitute \eqref{f_def} in the definition of $\widehat C_1(\mu)$, write the latter as a double integral, and change the order of integration.
Next, by deforming the contour $\Gamma$ to the circle of a large radius $R$ plus the interval $[-R,0]$, and sending $R$ to infinity, we obtain that
for $t$ inside $\Gamma,$
\begin{multline*}
\frac{1}{2\pi i}\int\limits_{\Gamma}\(\frac{x-a(\mu)}{x-\mu}\)^{\frac12}\frac{\ln x \ dx}{x-t}
\\=
2\ln\frac{\sqrt{\mu}+\sqrt{a(\mu)}}{2}
+2\(\frac{t-a(\mu)}{t-\mu}\)^{\frac12}
\ln\frac{t^{\frac12}\((t-a(\mu))^{\frac12}+(t-\mu)^{\frac12}\)}
{\sqrt{a(\mu)}(t-\mu)^{\frac12}+\sqrt{\mu}(t-a(\mu))^{\frac12}}.
\end{multline*}
Long but straighforward computations then lead to $\widehat C_1(\mu) = C_1(\mu).$

\subsubsection*{$\bullet\ \widehat C_0(\mu) = C_0(\mu).$}

To simplify the double integral in the definition of $\widehat C_0(\mu)$, we use the relations
\[
\dashint_{\mu}^{a(\mu)}\frac{\sqrt{(y-\mu)(a(\mu)-y)}\, dy}{\pi\, y(x-y)}=1-\frac{\sqrt{\mu\, a(\mu)}}{x},
\quad
\int\limits_{\mu}^{a(\mu)}\frac{\(x-\sqrt{\mu a(\mu)}\)\ln x\, dx}{\pi\,x\sqrt{(x-\mu)(a(\mu)-x)}}
=
2\ln\frac{\(\sqrt{\mu}+\sqrt{a(\mu)}\)^2}{4\sqrt{\mu a(\mu)}},
\]
and then the property follows from relation \eqref{lemma:e} of Lemma \ref{lemma:I}.

\medskip

The expansion \eqref{eq:CG} and the above identities together yield Proposition \ref{prop:Hankel} for $\varepsilon<\mu\leq \mu_0$, for any $\varepsilon,\mu_0>0$. It now remains to prove Proposition \ref{prop:Hankel} as $\mu\to 0$ and $n\to\infty$ simultaneously.

\subsection{Integration of differential identity.}

We can now integrate the result of Proposition \ref{theor_main}. Let $\mu_0>0$ be a fixed number. For sufficiently large  $M>0$, we take $\mu_1$ such that $\left(\frac{M}{\rho(\mu) n}\right)^2\leq\mu\leq \mu_1\leq \mu_0$.
Then, we integrate the result from Proposition \ref{theor_main} between $\mu_1$ and $\mu$, where the crucial observation is that the  integral of the error term is small enough: by the asymptotics \eqref{fcoeff} for $\rho(\mu)$, we have as $n\to\infty$, uniformly in $\mu$ and $\mu_1$,
\[\int_{\mu_1}^{\mu}r(\tilde\mu;n)d\tilde\mu=\begin{cases}\mathcal O\left(\int_{\mu_1}^{\mu}\frac{1}{n\tilde\mu^{3/2}}d\tilde\mu\right)+\mathcal O(\mu_1^{1/2})=\mathcal O\left(\frac{1}{n\sqrt{\mu}}\right)+\mathcal O\left(\mu_1^{1/2}\right),&\mbox{if $\beta>1,$}\\
\mathcal O\left(\int_{\mu_1}^{\mu}\frac{1}{n\tilde\mu^{3/2}|\(1+|\log\tilde\mu|\)}d\tilde\mu\right)+\mathcal O(\mu_1^{1/2})=\mathcal O\left(\frac{1}{n\sqrt{\mu}\(1+|\log\mu|\)}\right)+\mathcal O\left(\mu_1^{1/2}\right),&\mbox{if $\beta=1,$}\\
\mathcal O\left(\int_{\mu_1}^{\mu}\frac{1}{n\tilde\mu^{\beta/2-1}}d\tilde\mu\right)+\mathcal O(\mu_1^{1/2})=\mathcal O\left(\frac{1}{n{\mu}^{\beta/2}}\right)+\mathcal O\left(\mu_1^{1/2}\right),&\mbox{if $\beta<1,$}
\end{cases}
\]
which is $\mathcal O\left(\frac{1}{n\sqrt{\mu}\rho(\mu)}\right)+\mathcal O\left(\mu_1^{1/2}\right)$ in either of the three cases. Observe that the proportionality constants in the $\mathcal O$ terms do not depend on the choice of $\mu_1$.

As a consequence, we proved that there exists $M>0$ such that  as $n\to\infty$, 
we have uniformly for $\mu, \mu_1$ such that {$\left(\frac{M}{\rho(\mu) n}\right)^2\leq\mu\leq \mu_1\leq \mu_0$ } that 
\[
\ln \widetilde H_n(\mu) = \ln \widetilde H_n(\mu_1)
+n^2\int\limits_{\mu_1}^{\mu} D_2(\tilde \mu)\d\tilde\mu
+n\int\limits_{\mu_1}^{\mu} D_1(\tilde \mu)\d\tilde\mu
+\int\limits_{\mu_1}^{\mu} D_0(\tilde \mu)\d\tilde\mu
+\mathcal{O}\(\frac{1}{n\sqrt{\mu}\,\rho(\mu)}\)+\mathcal O\left(\mu_1^{1/2}\right),
\]
as $n\to\infty$.

For $\mu_1>0$ fixed, we have according to Section \ref{sect_fixed_mu},
\be\nonumber
\log \wt H_n(\mu_1)=C_2(\mu_1) n^2 +C_1(\mu_1) n -\frac{1}{6}\log n + C_0(\mu_1) + \widetilde r(\mu_1;n),\qquad n\to\infty,
\ee
where $|\widetilde r(\mu_1;n)|\leq c(\mu_1)\frac{\log n}{n}$ for some constant $c(\mu_1)$ which may depend on $\mu_1$.
Here $C_2(\mu_1), $ $C_1(\mu_1),$ $C_0(\mu_1)$ are as in Proposition \ref{prop:Hankel}.
Substituting this, we obtain
\begin{multline}\label{eq:logHnas}
\ln \widetilde H_n(\mu) = n^2\left( C_2(\mu_1)+\int\limits_{\mu_1}^{\mu} D_2(\tilde \mu)\d\tilde\mu\right)
+n\left(C_1(\mu_1)+\int\limits_{\mu_1}^{\mu} D_1(\tilde \mu)\d\tilde\mu\right)\\
-\frac{1}{6}\log n
+\left(C_0(\mu_1)+\int\limits_{\mu_1}^{\mu} D_0(\tilde \mu)\d\tilde\mu\right)
+\widetilde r(\mu_1;n)+\mathcal{O}\(\frac{1}{n\sqrt{\mu}\,\rho(\mu)}\)+\mathcal O\left(\mu_1^{1/2}\right).
\end{multline}

\medskip

If we take in the above $0<\mu<\mu_1$ fixed, then the coefficients in the large $n$ expansion need to be independent of $\mu_1$, and this implies without any computations that
\begin{equation}\label{identities:C1C2}\ C_2'(\mu)=D_2(\mu),\qquad \ C_1'(\mu)=D_1(\mu).\end{equation}
As a consistency check, we choose however to prove these identities also in a direct way, together with the analogous identity for the $\mathcal O(1)$ term, {$C_0'(\mu)=D_0^{\mathrm{full}}(\mu)$, where $D_0^{\mathrm{full}}(\mu)$ is given in Remark \ref{D_0:full}. We will also see that $D_0(\mu)-D_0^{\mathrm{full}}(\mu)=\mathcal{O}(\mu^{-1/2})$ as $\mu\to0,$ and thus Proposition \ref{theor_main} and Remark \ref{D_0:full} are consistent.}

\medskip
\subsubsection*{$\bullet\ C_2'(\mu)=D_2(\mu)$.}
In the definition of $I_{-1}(\mu)$ in \eqref{defI}, change the integration variable $t=\mu+(a(\mu)-\mu)u$. Then we obtain
\[
I_{-1}'(\mu)=(\bt/2-1)\left( I_{-2}(\mu)+\frac{a'(\mu)-1}{a(\mu)-\mu} J(\mu) \right).
\]
Applying here the identitities (b) and (c) of Lemma \ref{lemma:I}, we have
\[
I_{-1}'(\mu)=\frac{2\pi}{\mu a(\mu)}\frac{\mu a'(\mu)-a(\mu)}{a(\mu)-\mu}+\frac{1}{\mu}\left(\frac{\bt}{2}-1+\frac{-\mu a'(\mu)+a(\mu)}{2a(\mu)}
\right)I_{-1}(\mu).
\]
Substituting here for the last fraction (using the identity (d) of Lemma \ref{lemma:I})
\[
\frac{-\mu a'(\mu)+a(\mu)}{2a(\mu)}=\frac{2\pi a'(\mu)}{a(\mu) I_{-1}(\mu)},
\]
we obtain
\be\label{dfdmu}
\rho'(\mu)=\frac{\bt}{2\pi}  I_{-1}'(\mu)=\frac{\bt/2-1}{\mu}\rho(\mu)+ \frac{\bt(a'(\mu)-1)}{\mu(a(\mu)-\mu)}.
\ee
Using this identity, it is straightforward to compute the derivative of $C_2$ in \eqref{C1234}, and we obtain
$
C_2'(\mu)=D_2(\mu).
$

\subsubsection*{$\bullet\ C_1'(\mu)=D_1(\mu)$.}To compute the derivative of $C_1(\mu)$ in \eqref{C1234}, we use \eqref{dfdmu} and \eqref{lemma:d}. This gives
$
C_1'(\mu)=D_1(\mu).
$

\subsubsection*{$\bullet\ C_0'(\mu)=D_0^{\mathrm{full}}(\mu)$.}
The proof is based on the following identities:
\begin{proposition}\label{prop:logf} We have
\begin{align}
\frac{ \partial_x f (x;\mu) _{x=\mu}  }{\rho(\mu)}&=
2\frac{\rho'(\mu)}{\rho(\mu)}+\frac{a'(\mu)-1}{a(\mu)-\mu},\label{C0-1}
\\
\frac{\partial_x f(x;\mu)_{x=a(\mu)}}{f(a(\mu);\mu)}a'(\mu)&=
\frac{2}{3}\frac{\frac{d}{d\mu}f(a(\mu);\mu)}{f(a(\mu);\mu)}-\frac{1}{3}\frac{a'(\mu)-1}{a(\mu)-\mu}.
\label{C0-2}
\end{align}
\end{proposition}

\begin{proof}
First, integrating by parts we have
\[
\rho'(\mu)=
\clu{
\frac{\bt}{4\pi\i}\int_{\Gamma} t^{\bt/2-1}\frac{dt}{(t-a)^{1/2}(t-\mu)^{3/2}}
}
=
\frac{\beta}{\pi}\left(\frac{\bt}{2}-1\right)I_{-2}(\mu)
\clu{ 
-\frac{\bt}{4\pi\i}\int_{\Gamma} \frac{t^{\bt/2-1}\ dt}{(t-\mu)^{1/2}(t-a(\mu))^{3/2}}
}.
\]
Using this and \eqref{lemma:b}, we write
\[\begin{aligned}
\rho'(\mu)&=\frac{d}{d\mu} 
f(\mu;\mu)=\frac{1}{2} \partial_x f(x;\mu)_{x=\mu}
\clu{
+
a'(\mu) \frac{\bt}{8\pi\i}\int_{\Gamma} t^{\bt/2-1}\frac{dt}{(t-\mu)^{1/2}(t-a(\mu))^{3/2}}
}\\
&=
\frac{1}{2} \partial_x f(x;\mu)_{x=\mu}-\frac{a'(\mu)}{2}\left(\partial_x f(x;\mu)_{x=\mu}-
\frac{\beta}{\pi}\left(\frac{\bt}{2}-1\right)I_{-2}(\mu)\right)\\ &=
\frac{1}{2} \partial_x f(x;\mu)_{x=\mu}(1-a'(\mu))+
\frac{a'(\mu)}{\mu}\left(\frac{\bt}{2}-1\right)\left[ \rho(\mu) -\frac{\bt}{2\pi}J(\mu)\right].
\end{aligned}
\]
Hence, using also \eqref{lemma:c}, we have, collecting the first term to match \eqref{C0-1},
\[
\frac{\partial_x f(x;\mu)_{x=\mu}}{\rho(\mu)}=2\frac{\rho'(\mu)}{\rho(\mu)}+\frac{2a'(\mu)}{1-a'(\mu)}
\left(
\frac{\rho'(\mu)}{\rho(\mu)}-\frac{1}{\mu}\left(\frac{\bt}{2}-1\right)+\frac{1}{\mu a(\mu)}\left[\frac{\bt}{\rho(\mu)}-\frac{a(\mu)-\mu}{2}\right]
\right).
\]
Using \eqref{dfdmu} and 
the fact that by \eqref{lemma:d} $\frac{\bt}{\rho(\mu)}=\frac{a(\mu)-\mu a'(\mu)}{2a'(\mu)}$, to reduce the second term to a combination
of $a(\mu)$, $a'(\mu)$, and $\mu$,
we complete the derivation of
\eqref{C0-1}. The equation \eqref{C0-2} is obtained in a similar way.
\end{proof}

By \eqref{lemma:e}, we have
 $a'(\mu)=\frac{\rho(\mu)}{f(a(\mu);\mu)}$, hence we can rewrite the expression for {$D_0^{\mathrm{full}}(\mu)$ in Remark \ref{D_0:full}}
 in the form
 \be\nonumber
 {D_0^{\mathrm{full}}(\mu)}=-\frac{1}{8}\frac{a'(\mu)-1}{a(\mu)-\mu}-\frac{\al^2}{4}\frac{(\sqrt{a(\mu)}-\sqrt{\mu})^2}{a(\mu)-\mu}
 \frac{a(\mu)-a'(\mu)\mu}{a(\mu)\mu}-
 \frac{1}{16}\frac{\partial_x f(x;\mu)_{x=\mu}}{\rho(\mu)}-
 \frac{1}{16}\frac{\partial_x f(x;\mu)_{x=a(\mu)}}{f(a(\mu);\mu)}a'(\mu).
 \ee
 The proposition just proved allows to rewrite this expression as follows,
 \be\nonumber
 {D_0^{\mathrm{full}}(\mu)}=-\frac{1}{6}\frac{a'(\mu)-1}{a(\mu)-\mu}-\frac{\al^2}{4}\frac{(\sqrt{a(\mu)}-\sqrt{\mu})^2}{a(\mu)-\mu}\
 \frac{a(\mu)-a'(\mu)\mu}{a(\mu)\mu}-
 \frac{1}{8}\frac{\rho'(\mu)}{\rho(\mu)}
 -\frac{1}{24} \frac{\frac{d}{d\mu}f(a(\mu);\mu)}{f(a(\mu);\mu)}.
 \ee
This expression has the advantage of being easily integrable w.r.t. $\mu$. Integrating and using the matching at infinity as before,
we conclude that {$C_0'(\mu)=D_0^{\mathrm{full}}(\mu).$
Besides, we see from formula \eqref{dfdmu}, Proposition \ref{prop:logf}, and Lemma \ref{lemma:expaf} that the difference $D_0^{\mathrm{full}}(\mu)-D_0(\mu)$ in the limit $\mu\to0$ is of the order $\mathcal{O}(1)$ for $\beta>1,$ of the order $\mathcal{O}(|\ln \mu|)$ for $\beta=1$, and is of the order $\mathcal{O}\(\mu^{(\beta-1)/2}\)$ for $0<\beta<1$, and thus of the order $\mathcal{O}(\mu^{-1/2})$ for all $\beta>0,$ and hence Proposition \ref{theor_main} and Remark \ref{D_0:full} are consistent.}

\medskip

Combining \eqref{eq:logHnas} with the just obtained identities $
C_j'(\mu)=D_j(\mu), j=1,2$, {$
C_0'(\mu)=D_0^{\mathrm{full}}(\mu),$ and the relation $D_0^{\mathrm{full}}(\mu)-D_0(\mu)=\mathcal{O}(\mu^{-1/2})$ as $\mu\to0,$}
we obtain
\[
\ln \widetilde H_n(\mu) = C_2(\mu)n^2+C_1(\mu)n
{-\frac16\ln n}+C_0(\mu)
+\widetilde r(\mu_1;n)+\mathcal{O}\(\frac{1}{n\sqrt{\mu}\,\rho(\mu)}\)+\mathcal O\left(\mu_1^{1/2}\right),
\]
where we recall that $|\widetilde r(\mu_1;n)|\leq c(\mu_1)\frac{\log n}{n}$, and that the error terms are
uniform in $\mu,\mu_1$ for $\left(\frac{M}{\rho(\mu)n}\right)^2\leq \mu\leq \mu_1\leq \mu_0$ as $n\to\infty$.
It follows that there exists a constant $C>0$ independent of $\mu, \mu_1$ such that
such that
\[
\left|\ln \widetilde H_n(\mu) - C_2(\mu)n^2-C_1(\mu)n{+\frac16\ln n}-C_0(\mu)
\right|\leq \frac{C}{n\sqrt{\mu}\,\rho(\mu)}+C\mu_1^{1/2}+c(\mu_1)\frac{\log n}{n}.
\]
Let $\mu=\mu(n)$ be such that $\lim_{n\to\infty}\mu(n)=0$ and such that $\left(\frac{M}{\rho(\mu(n))n}\right)^2\leq \mu(n)\leq \mu_1\leq \mu_0$ for $n$ sufficiently large. 
For any choice of $\varepsilon>0$, we can take $\mu_1>0$ such that $C\mu_1^{1/2}<\varepsilon/2$, and also such that $\mu_1\geq \mu(n)$ for $n$ sufficiently large. Given such a value $\mu_1$, we can take $n$ sufficiently large such that  
$c(\mu_1)\frac{\log n}{n}<\varepsilon/2$.
Hence, for any $\varepsilon>0$, there exists $n_0$ such that for all $n\geq n_0$,
\[
\left|\ln \widetilde H_n(\mu) - C_2(\mu)n^2-C_1(\mu)n{+\frac16\ln n}-C_0(\mu)
\right|\leq \frac{C}{n\sqrt{\mu}\,\rho(\mu)}+\varepsilon.
\]
In other words,
\[
\ln \widetilde H_n(\mu) - C_2(\mu)n^2-C_1(\mu)n{+\frac16\ln n}-C_0(\mu)
=\mathcal O\(\frac{1}{n\sqrt{\mu}\,\rho(\mu)}\)+o(1),\]
as $n\to\infty$, 
and this proves
Proposition \ref{prop:Hankel} in the case where $\mu\to 0$ and $n\to\infty$ at the same time.

\appendix

\section{Derivation of \eqref{HH}}

In this appendix, we prove the identity \eqref{HH}.
Changing the integration variable $x=t^2 2^{2/\bt}$ in the orthonormality condition 
\[
1=\int_{\mu}^{\infty} p_k^{(\mu,\al)}(x)^2 e^{-n x^{\bt/2}}x^{\al}dx
\]
we obtain that
\be\label{polconn}
\widehat P^{(\lb)}_{2k}(t)=p_k^{(\mu,\al=-1/2)}(t^2 2^{2/\bt}) 2^{1/(2\bt)},\qquad 
\widehat P^{(\lb)}_{2k+1}(t)=t p_k^{(\mu,\al=1/2)}(t^2 2^{2/\bt}) 2^{3/(2\bt)},
\ee
where $\widehat P^{(\lb)}_n(t)$ are the orthonormal polynomials satisfying
\[
\int_{\mathbb R\setminus (-\lb,\lb)} \widehat P^{(\lb)}_k(t)\widehat P^{(\lb)}_j(t) e^{-2n |t|^{\beta}}dt,\qquad \mu=2^{2/\bt}\lb^2,
\]
with leading coefficients $\varkappa^{(\lb)}_k$. We now obtain \eqref{HH} using the identity
\[
H_{2n}(\lb)=\prod_{j=0}^{2n-1}(\varkappa^{(\lb)}_j)^{-2}=\prod_{k=0}^{n-1}(\varkappa^{(\lb)}_{2k})^{-2}
(\varkappa^{(\lb)}_{2k+1})^{-2}
\]
and the
connection between the leading coefficients 
of the polynomials $\widehat P^{(\lb)}_j(t)$ and $p^{(\mu,\al)}_j(t)$ provided by \eqref{polconn}:
\[
\varkappa^{(\lb)}_{2k}=\kappa^{(\mu,\alpha=-1/2)}_k   2^{2k/\bt+1/(2\bt)},\qquad
\varkappa^{(\lb)}_{2k+1}=\kappa^{(\mu,\alpha=1/2)}_k 2^{2k/\bt+3/(2\bt)},
\]
where
\[
\prod_{k=0}^{n-1}(\kappa^{(\mu,\alpha)}_k)^{-2}=\wt H_n^{(\alpha)}(\mu).
\]

\section{Single integral representation for $\varphi.$}\label{AppSIR}
Here, we provide yet another representation for $\varphi,$ as a single integral, as opposed to the double integral representations given in \eqref{def:phi}, \eqref{varphi}.
\begin{lemma}\label{lem_varphi}
The function $\varphi$ defined in \eqref{varphi} has the following single integral representation:
\be\label{varphi_single}
\varphi(z;\mu)
=
2\ln\frac{(z-\mu)^{\frac12}+(z-a(\mu))^{\frac12}}
{\sqrt{a(\mu)-\mu}}
-
\frac{(z-a(\mu))^{\frac12}(z-\mu)^{\frac12}}{4\pi\i}
\int\limits_{\Gamma}\frac{x^{\beta/2}\,\d x}{(x-a(\mu))^{\frac12}(x-\mu)^{\frac12}(x-z)},
\ee
where $(.)^{1/2}$ and $(.)^{\beta/2}$ denote the principal branches of the roots, and $\sqrt{.}$ the positive square root. Here $\Gamma$ is a counter-clock-wise oriented loop, which encircles the segment $[\mu, a(\mu)]$ and the point $z$, but does not intersect the negative half-axis $(-\infty,0].$
\end{lemma}
\begin{proof}
Substitute the expression \eqref{f_def} into \eqref{varphi}, and change the order of integration.
We obtain (see Figure \ref{FigureGamma} for a sketch of $\Gamma$ and the positions of $x, z, \xi$)
\be\label{varphi2}
\varphi(z;\mu) = \frac{\i\beta}{8\pi}
\int\limits_{\Gamma}
t^{\beta/2-1}\left(\frac{t-\mu}{t-a(\mu)}\right)^{1/2}          \left(\int\limits_{a(\mu)}^{z}
\(\frac{\xi-a(\mu)}{\xi-\mu}\)^{\frac12}\frac{\d\xi}{t-\xi}\right) dt.
\ee
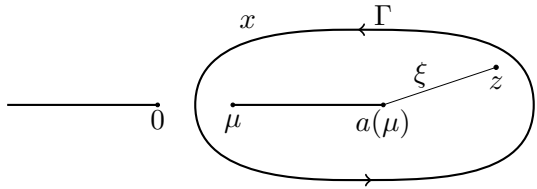
\begin{figure}[ht!]
\begin{tikzpicture}
\draw[thick] (-2,0) to (0,0);
\draw[thick] (1,0) to (3,0);
\draw [thick, decoration = {markings, mark = at position 0.25 with {\arrow{<}}}, 
decoration = {markings, mark = at position 0.75 with {\arrow{<}}}, postaction={decorate}] (0.5,0) to [out=90, in=180]  (2.5,1) to [out=0, in=90] (5,0) to [out=-90, in=0] (2.5,-1) to [out=180, in =-90] (0.5,0);
\filldraw (0,0) circle (0.7pt); \filldraw (1,0) circle (0.7pt); \filldraw (3,0) circle (0.7pt);
\node at (-0,-0.2) {$0$};
\node at (1,-0.25) {$\mu$};
\node at (3,-0.25) {$a(\mu)$};
\node at (3, 1.2){$\Gamma$};
\node at (1.2, 1.1){$x$};
\node at  (3.5,0.4){$\xi$};
\node at  (4.5,0.3){$z$};
\filldraw (4.5, 0.5) circle (0.8pt);
\draw[very thin] (3,0) to (4.5,0.5);
\end{tikzpicture}
\caption{In \eqref{varphi2}, $\Gamma$ is a positively oriented contour enclosing $[\mu, a(\mu)]$ and $z$, $\xi$ lies on the segment connecting $a(\mu)$ and $z\in\mathbb{C}\setminus(-\infty,a(\mu)],$ which lies inside $\Gamma,$ and $x$ is a point on $\Gamma.$}
\label{FigureGamma}
\end{figure}
The inner integral can be computed by elementary means, by substituting $v = \(\frac{\xi-a(\mu)}{\xi-\mu}\)^{\frac12}.$
Straightforward but long computations lead to
\[
\int\limits_{a(\mu)}^{z}\(\frac{\xi-a(\mu)}{\xi-\mu}\)^{\frac12}\frac{\d\xi}{t-\xi}
=\ln\frac{(z-\mu)^{\frac12}-(z-a(\mu))^{\frac12}}{(z-\mu)^{\frac12}+(z-a(\mu))^{\frac12}}
+\(\frac{t-a(\mu)}{t-\mu}\)^{\frac12}
\ln\frac
{\(\frac{t-a(\mu)}{t-\mu}\)^{1/2}
+
\(\frac{z-a(\mu)}{z-\mu}\)^{1/2}}
{\(\frac{t-a(\mu)}{t-\mu}\)^{1/2}
-
\(\frac{z-a(\mu)}{z-\mu}\)^{1/2}},
\]
and substituting that in \eqref{varphi2}, we obtain two integrals. We evaluate the first one explicitly using \eqref{defa}, and we integrate the second one by parts, using $(x^{\beta/2})' = \frac{\beta}{2}x^{\beta/2-1}.$
After simplifying the resulting expression, we obtain \eqref{varphi_single}.
\end{proof}
\subsubsection*{Acknowledgements}
IK and TC are grateful to Vladimir Kravtsov for bringing this model to their attention and for useful discussions.
TC and OM were supported by the Fonds de la Recherche Scientifique-FNRS under EOS project O013018F. IK was supported by the Leverhulme Trust research project grant RPG-2018-26.

\end{document}